%% file: main.tex
\documentclass[acmsmall]{acmart}

\input{preamble}

\input{macros}

\setcopyright{cc}
\setcctype{by}
\acmDOI{10.1145/3808288}
\acmYear{2026}
\acmJournal{PACMPL}
\acmVolume{10}
\acmNumber{PLDI}
\acmArticle{210}
\acmMonth{6}
\acmSubmissionID{pldi26main-p272-p}
\received{2025-11-13}
\received[accepted]{2026-04-03}

\begin{document}

\title{{Presynthesis}}
\subtitle{{Towards Scaling Up Program Synthesis with Finer-Grained Abstract Semantics}}


\author{Rui Dong}
\orcid{0000-0002-2757-2768}
\affiliation{%
  \institution{University of Michigan}
  \city{Ann Arbor}
  \country{USA}
}
\email{ruidong@umich.edu}

\author{Qingyue Wu}
\orcid{0009-0009-2888-6361}
\affiliation{%
  \institution{University of Michigan}
  \city{Ann Arbor}
  \country{USA}
}
\email{wqingyue@umich.edu}

\author{Danny Ding}
\authornote{Work done while working as an undergraduate researcher at the University of Michigan, Ann Arbor.}
\orcid{0009-0009-3437-8064}
\affiliation{%
  \institution{University of Texas at Austin}
  \city{Austin}
  \country{USA}
}
\email{dingyy@utexas.edu}

\author{Zheng Guo}
\orcid{0000-0002-0927-4011}
\affiliation{%
  \institution{University of Michigan}
  \city{Ann Arbor}
  \country{USA}
}
\email{zhgguo@umich.edu}

\author{Ruyi Ji}
\orcid{0000-0002-0150-8629}
\affiliation{%
  \institution{University of Michigan}
  \city{Ann Arbor}
  \country{USA}
}
\email{jiry@umich.edu}

\author{Xinyu Wang}
\orcid{0000-0002-1836-0202}
\affiliation{%
  \institution{University of Michigan}
  \city{Ann Arbor}
  \country{USA}
}
\email{xwangsd@umich.edu}



\input{sections/00-abstract}

\maketitle

\clearpage

\input{sections/01-intro}
\input{sections/02-overview}
\input{sections/03-alg}

\input{sections/04-sql}

\input{sections/05-eval}

\input{sections/06-related}

\input{sections/07-conc}

\clearpage

\input{sections/ack}

\section*{Data-Availability Statement}
The artifact that implements the techniques and supports the evaluation results reported in this paper is available on Zenodo~\cite{dong_2026_19078478}.

\bibliography{main}

\clearpage
\appendix

\input{sections/appendix-proof}

\input{sections/appendix-transformer}

\end{document}

%% file: preamble.tex
\usepackage{enumitem}
\usepackage[noend]{algpseudocode}
\usepackage{algorithmicx}
\usepackage{algorithm}
\usepackage{subcaption}
\usepackage{makecell}
\usepackage{booktabs}
\usepackage{multirow}
\usepackage{array}

\usepackage[colorinlistoftodos]{todonotes}

\usepackage{xspace}
\usepackage{amsmath}
\usepackage{xcolor}

\usepackage{booktabs}
\usepackage{subcaption}
\usepackage{hhline}

%% file: macros.tex
\newcommand{\frameworkname}{\textsc{Foresighter}\xspace}
\newcommand{\sqlname}{\textsc{ForesighterSQL}\xspace}
\newcommand{\stringname}{\textsc{ForesighterString}\xspace}
\newcommand{\matrixname}{\textsc{ForesighterMatrix}\xspace}

\newcommand{\xinyutodo}[1]{{\color{red}{#1}}}

\newcommand{\ruichecked}[1]{{\color{black}{#1}}}

\newcommand{\newpara}[1]{\vspace{5pt}\noindent\emph{\textbf{#1}}\hspace{4pt}}
\newcommand{\mydots}{\cdots}
\newcommand{\assign}{:=}

\newcommand{\myprocedure}[1]{\hspace{-10pt}\textbf{procedure }{#1}}
\newcommand{\myinput}{\hspace{-10pt}\textbf{input: }}
\algnewcommand\Output{\hspace{-10pt}\textbf{output: }}
\newcommand{\irule}[2]%
   {\mkern-2mu\displaystyle\frac{#1}{\vphantom{,}#2}\mkern-2mu}

\newcommand{\grammar}{{G}}
\newcommand{\grammareq}{::=}

\newcommand{\terminalsymbols}{\Sigma}
\newcommand{\nonterminalsymbols}{V}
\newcommand{\grammarsymbol}{s}
\newcommand{\startsymbol}{s_0}
\newcommand{\productions}{R}
\newcommand{\productionarrow}{\rightarrow}
\newcommand{\variablesymbol}{x}
\newcommand{\allvariablesymbols}{\terminalsymbols_{\emph{var}}}
\newcommand{\constantsymbol}{c}
\newcommand{\allconstantsymbols}{\terminalsymbols_{\emph{const}}}
\newcommand{\dsloperator}{f}

\newcommand{\fta}{\mathcal{A}}
\newcommand{\ftastates}{Q}
\newcommand{\ftafinalstates}{Q_f}
\newcommand{\ftatransitions}{\Delta}
\newcommand{\ftatransition}{\delta}
\newcommand{\ftatransitionarrow}{\rightarrow}
\newcommand{\ftaalphabet}{F}

\newcommand{\ftastate}{q}

\newcommand{\abstractdomain}{\mathcal{D}}
\newcommand{\abstractinputspace}{\mathcal{D}}
\newcommand{\abstractvalue}{a}

\newcommand{\abstracttransformer}[1]{\llbracket{#1}\rrbracket^{\sharp}}

\newcommand{\concretetransformer}[1]{\llbracket #1 \rrbracket}
\newcommand{\transformer}[1]{\llbracket #1 \rrbracket^{\sharp}}
\newcommand{\concretevalue}{c}

\newcommand{\ioex}{e}
\newcommand{\inputex}{\ioex_{\textit{in}}}
\newcommand{\outputex}{\ioex_{\textit{out}}}
\newcommand{\inputtooutput}{\mapsto}
\newcommand{\inputvartoinputval}{\mapsto}
\newcommand{\ioexs}{E}
\newcommand{\abstractinputenv}{\Gamma}

\newcommand{\prog}{P}

\newcommand{\mynull}{null}

\newcommand{\ftaslice}{\fta_{\ioex}}
\newcommand{\ftaslicee}[1]{\fta_{#1}}

\newcommand{\cubes}{\xspace\textsc{Cubes}\xspace}

\newcommand{\patsql}{\xspace\textsc{PatSQL}\xspace}
\newcommand{\blazesql}{\xspace\textsc{BlazeSQL}\xspace}
\newcommand{\blaze}{\xspace\textsc{Blaze}\xspace}

\newcommand{\toolnopresynkpred}[1]{\xspace\textsc{NoPresyn-}{#1}\xspace}
\newcommand{\toolnopresynthesis}{\xspace\textsc{NoPresyn}\xspace}
\newcommand{\toolnoslicingoracle}{\xspace\textsc{NoOracle}\xspace}

\newcommand{\toolkpred}[1]{\xspace\textsc{ForesighterSQL-}{#1}\xspace}

\newcommand{\stringkpred}[1]{\xspace\textsc{ForesighterString-}{#1}\xspace}
\newcommand{\matrixkpred}[1]{\xspace\textsc{ForesighterMatrix-}{#1}\xspace}

\newcommand{\presynthesisalg}{\textsc{Presynthesize}}
\newcommand{\synthesisalg}{\textsc{Synthesize}}
\newcommand{\buildftaalg}{\textsc{BuildOfflineFTA}}
\newcommand{\buildslicingoraclealg}{\textsc{BuildOracle}\xspace}
\newcommand{\slicealg}{\textsc{Slice}\xspace}
\newcommand{\searchalg}{\textsc{Search}}
\newcommand{\slicingoracle}{\mathcal{O}}

\newcommand{\bigland}{\bigwedge}
\newcommand{\biglor}{\bigvee}

\newcommand{\sqlquery}{Q}
\newcommand{\sqlinputtablevar}{x}
\newcommand{\sqlpredicate}{\phi}
\newcommand{\filterpred}{\sqlpredicate}
\newcommand{\joinpred}{\sqlpredicate}
\newcommand{\proj}{\text{Project}}
\newcommand{\filter}{\text{Filter}}
\newcommand{\innerjoin}{\text{IJoin}}
\newcommand{\leftjoin}{\text{LJoin}}
\newcommand{\distinct}{\text{Distinct}}
\newcommand{\groupby}{\text{GroupBy}}

\newcommand{\agglist}{L_{g}}
\newcommand{\columnlist}{L_{c}}
\newcommand{\aggexpr}{g}

\newcommand{\agg}{\textit{agg}}
\newcommand{\logicop}{\odot}
\newcommand{\sqlvalue}{v}

\newcommand{\sqllimit}{\text{Limit}}

\newcommand{\like}{\text{like}}
\newcommand{\colidx}{{col}}
\newcommand{\sqlmin}{\text{Min}}
\newcommand{\sqlmax}{\text{Max}}
\newcommand{\sqlavg}{\text{Avg}}
\newcommand{\sqlcount}{\text{Count}}
\newcommand{\sqlsum}{\text{Sum}}

\newcommand{\sqlminmax}{\text{Min/Max}}
\newcommand{\sqlsumavg}{\text{Sum/Avg}}

\newcommand{\prednumofcols}{{\texttt{NumCols}}}
\newcommand{\predsubsetofinput}{{\texttt{SubsetOf}}}
\newcommand{\predcoltype}{{\texttt{ColType}}}
\newcommand{\prednumofrows}{{\texttt{NumRows}}}

\newcommand{\predcolrange}{{\texttt{In}}}
\newcommand{\predlistcolinset}{{\texttt{In}}}
\newcommand{\predcolsused}{{\texttt{MaxCols}}}
\newcommand{\setofcols}{J}

\newcommand{\maxcol}{m}

\newcommand{\lenOfString}{{\texttt{Len}}}
\newcommand{\fromInput}{{\texttt{FromInp}}}
\newcommand{\charEqual}{{\texttt{CharEq}}}
\newcommand{\indexOf}{{\texttt{Index}}}

\newcommand{\dimNum}{{\texttt{DimNum}}}
\newcommand{\shapeOf}{{\texttt{Shape}}}
\newcommand{\elementAt}{{\texttt{ElemEq}}}
\newcommand{\vin}{{v_{\textit{in}}}}

\newcommand{\vout}{{v_{\textit{out}}}}
\newcommand{\voutprime}{{v'_{\textit{out}}}}
\newcommand{\vectorOf}{{\texttt{Vector}}}

\newcommand{\nbc}[3]{
	{\colorbox{#3}{\bfseries\sffamily\scriptsize\textcolor{white}{#1}}}
	{\textcolor{#3}{\sf\small \textit{#2}}}
}
\definecolor{rycolor}{RGB}{150, 0, 0}
\definecolor{ruicolor}{RGB}{0, 0, 150}
\newcommand\ruyicomment[1]{\nbc{RJ}{#1}{rycolor}}

\newcommand{\graybox}[1]{{\setlength{\fboxsep}{2pt}\colorbox{gray!20}{#1}}}
\newcommand{\grayboxr}[1]{  \begingroup
  \setlength{\fboxsep}{0pt} 
  \colorbox{gray!20}{\strut #1}
  \endgroup
}

\newcommand{\lang}{\mathcal L}
\newcommand{\ftarun}{\rho}
\newcommand{\ftaruns}{runs\xspace}

\newcommand{\runvar}{\rho} 
\newcommand{\runset}{\mathcal R}
\newcommand{\bottomup}{\textsc{BottomUpReachable}}
\newcommand{\inputconsistent}{\textsc{InputConsistent}}
\newcommand{\shl}{<\!\!<}
\newcommand{\bv}[1]{#1}
\newcommand{\bvlastbit}[1]{{\square #1}}

\newcommand{\myautoref}[2]{\hyperref[#2]{#1~\ref*{#2}}}
\newcommand{\algorithmref}[1]{\myautoref{Algorithm}{#1}}
\newcommand{\alglineref}[1]{\myautoref{line}{#1}}
\newcommand{\lemmaref}[1]{\myautoref{Lemma}{#1}}
\newcommand{\defref}[1]{\myautoref{Definition}{#1}}
\newcommand{\ignore}[1]{}
\newcommand{\samplestate}[2]{\langle #1, \bv{#2} \rangle}

\newcommand{\inputconsistentconstraint}{\Psi}
\newcommand{\bindsto}{\textsc{BindsTo}}

\newcommand{\sygus}{\textsc{SyGuS}}

\newcommand{\speedup}[1]{{#1}\times}

%% file: sections/00-abstract.tex
\begin{abstract}

Abstract semantics has proven to be instrumental for accelerating search-based program synthesis, by enabling the sound pruning of a set of incorrect programs (without enumerating them). 
One may expect \emph{faster} synthesis with increasingly \emph{finer-grained} abstract semantics.
Unfortunately, to the best of our knowledge, this is not the case, yet.
The reason is because, as abstraction granularity increases---while fewer programs are enumerated---pruning becomes more costly. 
This imposes a fundamental limit on the overall synthesis performance, which we aim to address in this work.

Our key idea is to introduce an \emph{offline presynthesis} phase, which consists of two steps. 
Given a DSL with abstract semantics, the first \emph{semantics modeling} step constructs a tree automaton $\fta$ for \emph{a space of inputs}---such that, for any program $\prog$ and for any considered input $I$, $\fta$ has a run that corresponds to $\prog$'s execution on $I$ under abstract semantics. 
Then, the second step builds an oracle $\slicingoracle$ for $\fta$. 
This $\slicingoracle$ enables fast pruning \emph{during synthesis}, by allowing us to efficiently find exactly those DSL programs that satisfy a given input-output example under abstract semantics.

We have implemented this \emph{presynthesis-based synthesis paradigm} in a framework, $\frameworkname$. 
On top of it, we have developed three instantiations for SQL, string transformation, and matrix manipulation. All of them significantly outperform prior work in the respective domains. 

\end{abstract}

%% file: sections/01-intro.tex
\section{Introduction}
\label{sec:intro}

This paper concerns the long-standing problem of syntax-guided program synthesis from examples, also known as programming-by-example (PBE). 
That is, 

\vspace{5pt}
\begin{center}
\parbox{.95\linewidth}{
Given a domain-specific language (DSL) whose syntax is defined by a context-free grammar and given a set $\ioexs$ of input-output examples, find a program $\prog$ in the DSL such that $\prog$ satisfies $\ioexs$ according to the DSL's semantics. 
}
\end{center}
\vspace{5pt}

This is an intellectually challenging problem with a wide range of practical applications, evidenced by the seminal FlashFill work~\cite{gulwani2011automating} among many others~\cite{solar2008program,gulwani2017program,alur2013syntax,solar2006combinatorial,wang2017program,polozov2015flashmeta,srivastava2010program}. 
A widely-used synthesis method is to \emph{search} over the given grammar for satisfying programs. 
But, as one can imagine, this is computationally intractable, due to the enormous search space of programs.

\newpara{State-of-the-art.}
To scale up synthesis, many state-of-the-art synthesizers leverage some form of \emph{abstract semantics} to prune the search space~\cite{wang2017program,feser2015synthesizing,wang2017synthesizing,feng2017component,polikarpova2016program,lee2021combining,guria2023absynthe}.
The idea is to reason about the behavior of \emph{a set of programs} to prove their inconsistency with the specification, such that they can be safely pruned away. 
For example, top-down synthesis~\cite{feng2017component} works by enumerating program sketches (i.e., incomplete programs with holes). 
A sketch can be pruned away if its abstract behavior is not consistent with the example, thereby saving us from enumerating any of its completions. 
As another example, in representation-based synthesis~\cite{wang2017program}, a graph-like data structure is first built for the given input-output example, based on the DSL's abstract semantics. 
In this graph, nodes are associated with abstract output values of programs on the given input, edges are created using abstract semantics, and paths represent programs. 
Then, synthesis boils down to searching among  ``consistent'' paths (i.e., paths that reach a target node whose associated abstract value is consistent with the desired output).
Irrelevant nodes (that do not lie on a consistent path) are not considered during search; in other words, their corresponding programs are pruned away.

\newpara{Finer-grained abstraction $\Rightarrow$ fast synthesis?}
The idea of abstraction-based pruning is profound, and underpins search-based synthesis techniques across a wide range of applications. 
In this work, we ask the following question.

\vspace{5pt}
\emph{Can we scale up synthesis by providing it with increasingly finer-grained abstract semantics?}
\vspace{5pt}

This is an important question, because a positive answer indicates that we can potentially use abstraction granularity as a dimension to continuously scale up synthesis.
Unfortunately, to our best knowledge, this is not the case, yet. 
While developing precise abstract semantics is itself very hard~\cite{johnson2024automating,wang2018learning}, an orthogonal challenge is that \emph{pruning is not always cheap.}
With a finer abstraction, while more programs can potentially be pruned away, the cost of pruning also increases.
As the abstraction becomes more granular, this overhead can offset---or even outweigh---its benefit to synthesis.
For instance, in top-down synthesis, 
evaluating a program sketch using a finer-grained abstraction
is often more costly~\cite{feng2017component}, while potentially still failing to prune it away.
In representation-based synthesis, finer-grained abstract semantics will result in a larger graph \emph{with more irrelevant nodes}~\cite{wang2017program}.
This ``graph bloat'' phenomenon is fundamentally because we cannot soundly predict, during its construction, which part of the graph is relevant, until the graph is fully built.

\begin{figure}[!t]
\centering
\includegraphics[width=0.85\linewidth]{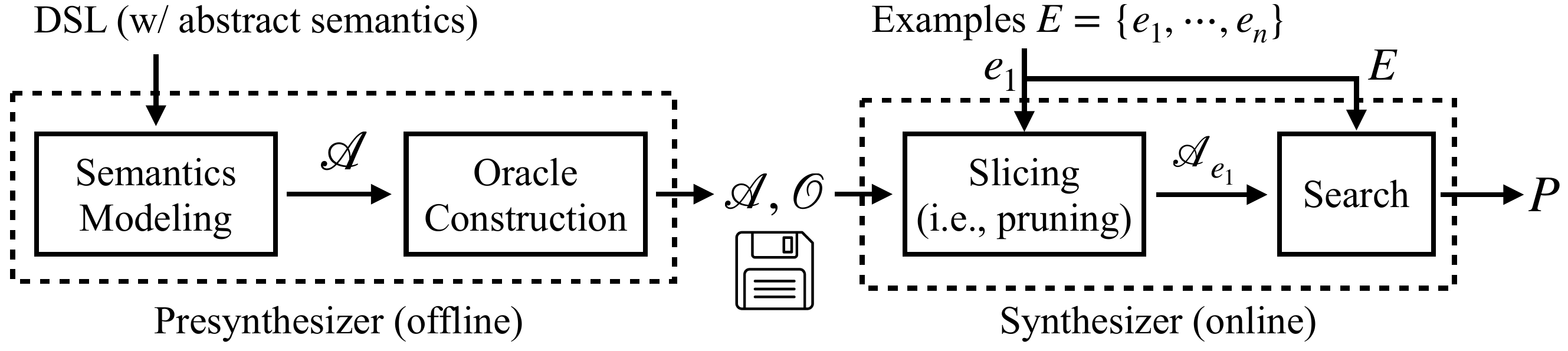}
\caption{Schematic workflow of our presynthesis-based synthesis paradigm.}
\label{fig:workflow}
\end{figure}

\newpara{Our work: presynthesis-based synthesis.}
We present a scalable pruning technique, which we believe is a crucial step towards scaling up synthesis with finer-grained abstract semantics. \autoref{fig:workflow} shows the schematic workflow. 
The key idea is to introduce an \emph{offline presynthesis} phase. 
Given a DSL with abstract semantics, we \emph{precompute} some information, which can be stored on a disk. 
To solve a stream of synthesis tasks, such information is first loaded into the synthesizer, and will be used to enable fast pruning. 
This is a paradigm shift, from existing work\footnote{{We refer interested readers to~\autoref{sec:related} for related work that utilizes preprocessing/precomputation.}} which is virtually all ``purely online'', to one that can take advantage of \emph{offline} resources for better \emph{online} performance.
Two key ideas allow us to realize this ``presynthesis-based synthesis'' paradigm.

\newpara{Key idea 1: offline abstract semantics modeling using FTAs.}
Our first idea is to model the given DSL's abstract semantics for \emph{a space of inputs} (that are of interest for future synthesis).
Specifically, given a DSL with abstract semantics, we build  a finite tree automaton (FTA) $\fta$ during presynthesis, such that, for any program $\prog$ and for any input in the considered space, $\fta$ has a run $\ftarun$ that begins from an abstraction of the given input and ends in a state labeled with $\prog$'s resulting abstract output.

\newpara{Key idea 2: fast pruning using offline-built oracle.}
Given this offline FTA $\fta$, synthesis boils down to finding a run $\ftarun$ of $\fta$ that is \emph{consistent with} the given input-output example $\ioex = \inputex \inputtooutput \outputex$. 
Here, $\ftarun$ is said to be consistent with $\ioex$, if it (i) begins from states consistent with $\inputex$ and (ii) ends in a state that matches $\outputex$. 
This can be viewed as a ``graph search'' problem; the challenge lies in that our $\fta$ is very large. 
Our idea is to first build a form of ``reachability oracle'' $\slicingoracle$ \emph{during presynthesis}.
Then, \emph{during synthesis}, $\slicingoracle$ is used to help extract a \emph{slice} $\ftaslicee{\ioex}$ (which is a sub-FTA of $\fta$) for $\ioex$, such that $\ftaslicee{\ioex}$ represents all programs that satisfy $\ioex$ under abstract semantics. 
This slicing step essentially performs pruning. 
However, unlike prior work, our slicing-based pruning is very fast, thanks to our oracle $\slicingoracle$ (which significantly reduces the ``graph bloat'' while building $\ftaslicee{\ioex}$).
Finally, we perform standard search over $\ftaslicee{\ioex}$ to find a program that satisfies all examples under concrete semantics.

\newpara{Synthesis as a service.}
As one can imagine, in order for the synthesizer to fully load in the offline artifact, more memory is needed (compared to purely online approaches). Therefore, in general, we envision our presynthesis-based synthesizer to be deployed as a \emph{service}: that is, user requests (i.e., specifications) are sent to a \emph{synthesis server}, and the synthesized programs are sent back to users. 
Note that users do not need as much memory, unless they choose to deploy the service locally.

\newpara{Implementation, instantiation, and evaluation.}
We have implemented this presynthesis-based synthesis paradigm into a framework called $\frameworkname$, and instantiated it for three different domains. In addition to outperforming prior work across all domains, it is also worth highlighting that, as the abstraction granularity increases (up to the finest available), the pruning time of our approach stays consistently low and the synthesis time continues to decrease.

\newpara{Contributions.}
This paper makes the following contributions. 
\begin{itemize}[leftmargin=*]
\item 
The concept of presynthesis, and a new presynthesis-based synthesis paradigm. 
\item 
The first presynthesis technique.
\item 
The first synthesis algorithm that utilizes presynthesis-based pruning. 
\item 
Implementation of the presynthesis-based synthesis framework, and three instantiations.
\item 
Comprehensive evaluation, with promising experimental results across the board.
\end{itemize}

%% file: sections/02-overview.tex

\section{Overview}
\label{sec:overview}

This section illustrates our approach over a simple synthesis task.

\newpara{Domain-specific language (DSL).}
Consider the following unrolled grammar\footnote{{We assume an unrolled grammar to simplify the presentation.}}, which specifies the syntax of a small bit-vector language from the $\sygus$ dataset~\cite{sygus-website,alur2013syntax}.
For presentation purposes, we consider vectors of length $4$, left shift $\shl$ and addition $+$, one single input variable $x$, and two integer constants for $k$.
We further assume only $t_2$ is the start symbol. 
The semantics is standard.
\[\small
\begin{array}{lcl}
t & \grammareq &  x \ | \ t + k \ | \ t \shl k \\ 
k & \grammareq &  1 \ | \ 2 \\ 
\end{array}
\xrightarrow{\quad\text{Unrolling}\quad}
\begin{array}{lcl}
t_2 & \grammareq &  t_1 + k \ | \ t_1 \shl k \\ 
t_1 & \grammareq &  t_0 + k \ | \ t_0 \shl k \\ 
t_0 & \grammareq &  x  
\\ 
k &  \grammareq  & 1 \ | \ 2 \\ 
\end{array}
\]

\newpara{Programming-by-example (PBE) task.}
Given this DSL, consider a PBE task with the following set $\ioexs$ of two examples. 
The PBE task is to find a DSL program that satisfies both examples in $\ioexs$.
\[
\ioexs = 
\left\{
\begin{array}{l}
\ioex_1: \{ \ x = \bv{0001} \ \} \inputtooutput \bv{0110} 
\\ 
\ioex_2: \{ \ x = \bv{0010} \ \} \inputtooutput \bv{1010}
\end{array}
\right\}
\]
In our unrolled grammar, $(x \shl 2) + 2$ is the only program that satisfies $\ioexs$.

\subsection{FTA-Based Synthesis using Abstract Semantics}
\label{sec:overview:fta-synthesis}

Let us first illustrate an existing PBE framework~\cite{wang2017program,wang2019efficient}---which our work builds upon---that is based on finite tree automata (FTAs) and abstract semantics.

\newpara{Abstract semantics.}
The abstract semantics for our bit-vector DSL is defined by a set of abstract transformers for all operators in the DSL.
A transformer for operator $\dsloperator$ accepts abstract values as arguments and returns an abstract value as its output. 
Abstract values are drawn from a predefined domain of abstract values (i.e., abstract domain).
Below are some example transformers. 
\[
\abstracttransformer{\bvlastbit{1} \shl 1} = \bvlastbit{0} 
\qquad 
\abstracttransformer{\bvlastbit{1} \shl 2} = \bvlastbit{0} 
\qquad 
\abstracttransformer{\bvlastbit{1} + 1} = \bvlastbit{0} 
\qquad
\abstracttransformer{\bvlastbit{1} + 2} = \bvlastbit{1}
\]
$\bvlastbit{b}$ denotes an (abstract) vector whose lowest bit is $b$.
For instance, the first example says ``performing left shift on a vector whose lowest bit is $1$ always yields a vector whose lowest bit is $0$''.

\newpara{PBE using FTAs and abstract semantics.}
We can perform synthesis by first building an FTA, given  a DSL with abstract semantics and given an example.~\autoref{fig:online-fta-lowest-bit} shows the FTA  $\fta_1$ for the first example $\ioex_1$, given our DSL with the aforementioned abstract semantics.
Each state $\ftastate_{\grammarsymbol}^{\abstractvalue}$ is associated with a grammar symbol $\grammarsymbol$ and an abstract value $\abstractvalue$. 
Each transition corresponds to a DSL operator $\dsloperator$, and is built using $\dsloperator$'s abstract transformer $\abstracttransformer{\dsloperator}$.
An accepting run $\ftarun$ of $\fta_1$ represents a program that satisfies $\ioex_1$ under abstract semantics. 
For instance, consider the following run $\ftarun$ of $\fta_1$
\[
\xrightarrow{ \ x \ }
\ftastate_{t_0}^{\bvlastbit{1}} 
\xrightarrow{ \ \shl_{1} \ }
\ftastate_{t_1}^{\bvlastbit{0}}
\xrightarrow{ \ +_{2} \ }
\ftastate_{t_2}^{\bvlastbit{0}}
\]
which corresponds to $(x \shl 1) + 2$. 
This program, however, does not satisfy $\ioex_1$ under the DSL's \emph{concrete} semantics.
It is not hard to figure out that $\fta_1$ has in total $12$ accepting runs, while the DSL has a total of $16$ programs.
Given $\fta_1$, we can perform search \emph{among its accepting runs}, to find a program that satisfies both $\ioex_1$ and $\ioex_2$ under concrete semantics.

\newpara{Remarks.}
While $\fta_1$ prunes away $4$ programs (which are not considered during search), they are still presented in~\autoref{fig:online-fta-lowest-bit} as runs ending at $\ftastate_{t_2}^{\bvlastbit{1}}$.
This is because states are built in a bottom-up fashion: begin with states for input variables, iteratively create more states, until all programs in the DSL are considered. 
For instance, in~\autoref{fig:online-fta-lowest-bit}, 
transition $\xrightarrow{ \ x \ } \ftastate_{t_0}^{\bvlastbit{1}}$ is first created, since $\ioex_1$ binds $x$ to $\bv{0001}$ whose lowest bit is $1$. 
Note the state $\ftastate_{t_2}^{\bvlastbit{1}}$: while ``useless'', it is still created---this is because, in general, we cannot foresee whether or not a state can lead to a final state. 
As a result, prior work first builds all ``reachable'' states from the input example and then removes those not reaching a final state (i.e., the desired output example).

\begin{figure*}[t]
\centering 
\hfill
\begin{subfigure}{0.45\textwidth}
\centering
\includegraphics[scale=0.7]{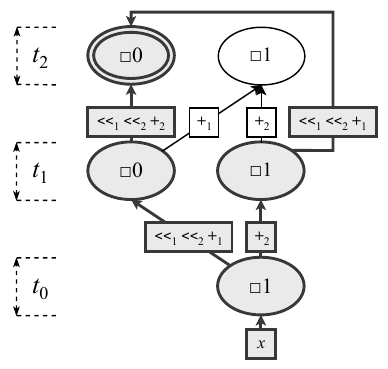}
\caption{FTA $\fta_1$ for $\ioex_1$, tracking lowest bit.}
\label{fig:online-fta-lowest-bit}
\end{subfigure}
\hfill
\begin{subfigure}{0.48\textwidth}
\centering
\includegraphics[scale=0.7]{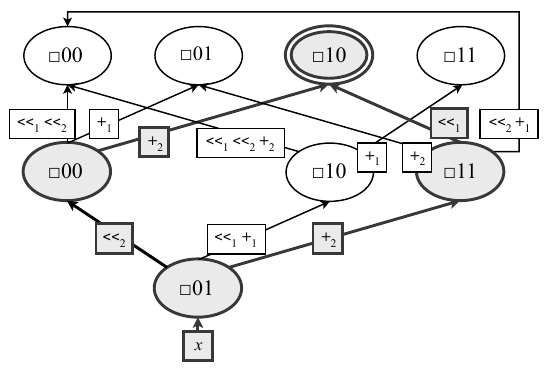}
\caption{FTA $\fta_2$ for $\ioex_1$, tracking lowest two bits.}
\label{fig:online-fta-lowest-two-bits}
\end{subfigure}
\hfill\
\caption{Two FTAs for input-output example $\ioex_1: \{ x = \bv{0001} \} \inputtooutput \bv{0110}$.
$\fta_1$ in \autoref{fig:online-fta-lowest-bit} uses an abstract domain that tracks the lowest bit for all bitvectors, 
and $\fta_2$ in \autoref{fig:online-fta-lowest-two-bits} uses a finer-grained domain that tracks the lowest two bits.
Each state is labeled with its associated abstract value, and the associated grammar symbol is shown on the far left.
For presentation purposes, transitions for $\shl$ and $+$ also include constant $k$ as part of the label (boxed), and we merge labels for transitions that share the same input and output states. 
Final states are double-circled. 
Accepting runs are highlighted with gray states and bold transitions.}
\label{fig:online-ftas}
\end{figure*}

\subsection{The FTA Bloat Phenomenon}
\label{sec:overview:bloat}

Now let us illustrate the FTA bloat issue when a finer-grained abstract semantics is used.

\newpara{Finer-grained abstract semantics.}
Consider a more precise abstract domain, which tracks the lowest two bits. The abstract transformers can be defined accordingly. Below are some examples. 
\[
\abstracttransformer{\bvlastbit{01} \shl 1} = \bvlastbit{10} 
\qquad 
\abstracttransformer{\bvlastbit{11} \shl 2} = \bvlastbit{00} 
\qquad 
\abstracttransformer{\bvlastbit{01} + 1} = \bvlastbit{10} 
\qquad
\abstracttransformer{\bvlastbit{00} + 2} = \bvlastbit{10}
\]

\newpara{Fewer accepting runs.}
\autoref{fig:online-fta-lowest-two-bits} shows the FTA $\fta_2$ for this new abstraction. 
Compared to $\fta_1$ in \autoref{fig:online-fta-lowest-bit}, $\fta_2$ is larger, since it tracks finer-grained information. As a result, $\fta_2$ has fewer accepting runs---in particular, only 2, compared to 12 in $\fta_1$. In other words, more programs are pruned away.
This makes the subsequent search faster. 

\newpara{More bloat.}
However, this additional pruning power and faster search do not come for free---$\fta_2$ has more ``useless'' states and transitions. As alluded to earlier, this FTA bloat phenomenon is due to the inability to soundly ``predict the future'' (i.e., if a current state will lead to a future final state). 
For instance, $\ftastate_{t_1}^{\bvlastbit{10}}$ in $\fta_2$ does not lead to a final state, no matter what operators and constants are used; but, in general, we must first build all of its outgoing transitions and states, to soundly tell so.
This FTA bloat issue is exacerbated when finer-grained abstract semantics is used.

\subsection{Our Approach: Presynthesis-Based Synthesis}
\label{sec:overview:our-work}

Finally, let us illustrate how we enable fast pruning without causing FTA bloat during synthesis.

\begin{figure}[t]
\hfill 
\includegraphics[scale=0.7]{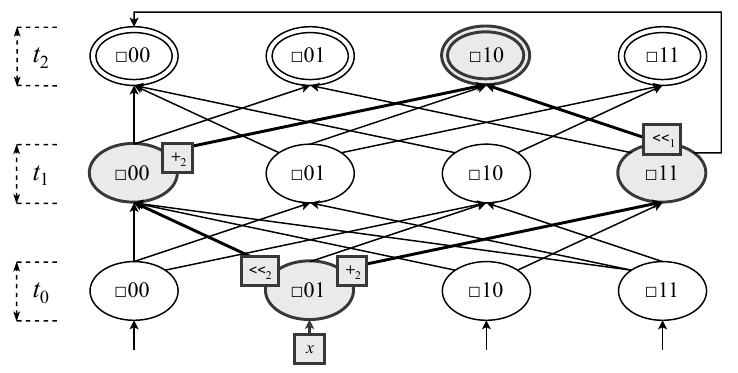}
\hfill\  
\caption{Offline FTA $\fta$ for our DSL using abstract semantics that tracks the lowest two bits.  
We highlight the part of $\fta$ that corresponds to the slice $\ftaslicee{\ioex_1}$ for the first example $\ioex_1: \{ x = \bv{0001} \} \inputtooutput \bv{0110}$, and we elide labels for transitions not in $\ftaslicee{\ioex_1}$ to simplify the presentation. 
}
\label{fig:offline-fta}
\end{figure}

\newpara{Offline FTA-based semantics modeling.}
Our FTA is built during a \emph{presynthesis} phase. 
\autoref{fig:offline-fta} shows our offline FTA $\fta$,  built using the same abstraction for $\fta_2$.
Unlike $\fta_2$, $\fta$ has 3 additional ``initial'' states for $x$, because $\fta$ needs to consider all potential input vectors. 
Another distinction is that all $t_2$ states are final, because any of them can be desired before ``online examples'' are revealed.
The rest of the FTA construction machinery remains the same. 

Our new insight is that $\fta$ essentially models the DSL's abstract semantics---in the sense that, the execution of any program $P$ on any input corresponds to an accepting run $\ftarun$ of $\fta$.
Prior work that builds FTAs \emph{online} also performs such modeling, but only for specific inputs.

\newpara{Online slicing.}
Now, consider the synthesis task for our examples $\ioexs = \{ \ioex_1, \ioex_2 \}$.
First, we extract a \emph{slice} $\ftaslicee{\ioex_1}$ out of $\fta$. 
$\ftaslicee{\ioex_1}$ is a sub-FTA of $\fta$, as highlighted in~\autoref{fig:offline-fta}. 
$\ftaslicee{\ioex_1}$'s accepting runs represent exactly those programs that satisfy $\ioex_1$ abstractly. 
Readers now might wonder: how is $\ftaslicee{\ioex_1}$ different from $\fta_2$ (see~\autoref{fig:online-fta-lowest-two-bits})? 
Indeed, they have the language. 
The key distinction is that our slicing does not suffer from the bloat issue as when building $\fta_2$. Let us illustrate our slicing process below.

Given $\fta$ and $\ioex_1$, we first find final states that are consistent with the output example: in our case,  $\ftastate_{t_2}^{\bvlastbit{10}}$ is the only such state. 
Then, we traverse $\fta$ top-down from $\ftastate_{t_2}^{\bvlastbit{10}}$ to collect its ``descendants''. 
For instance, we will collect $\ftastate_{t_2}^{\bvlastbit{00}}$ and $\ftastate_{t_2}^{\bvlastbit{11}}$, both of which belong to $\ftaslicee{\ioex_1}$, as well as $\ftastate_{t_0}^{\bvlastbit{01}}$. 
However, we will not collect $\ftastate_{t_1}^{\bvlastbit{01}}$, because it is reachable only from initial states that are \emph{not consistent} with $\ioex_1$'s input. 
Similarly, we will not collect  $\ftastate_{t_0}^{\bvlastbit{00}}$ and $\ftastate_{t_0}^{\bvlastbit{11}}$ either.
In other words, we smartly avoid \emph{irrelevant} states that end up not in  $\ftaslicee{\ioex_1}$ (which is the root cause of bloat!).
How can we foresee the future?

\newpara{Offline-built oracle.}
Our foresight comes from the use of an oracle $\slicingoracle$, that is built offline during presynthesis. 
Intuitively, $\slicingoracle$ tells us whether or not a state $\ftastate$ in $\fta$ is reachable from initial states that are consistent with $\ioex_1$'s input: if so, $\ftastate$ is relevant, and will be added to our slice $\ftaslicee{\ioex_1}$. 
In particular, $\slicingoracle$ maps every $\ftastate$ to all \emph{abstract inputs} that ``can reach'' $\ftastate$. 
For instance, $\slicingoracle$ contains the following entry
\[
\slicingoracle( \ftastate_{t_1}^{\bvlastbit{01}} )  \ \   = 
\   \ 
\big( x = \bvlastbit{00} \big) \lor 
\big( x = \bvlastbit{11} \big)
\]
which means: there is a run to $\ftastate_{t_1}^{\bvlastbit{01}}$ in $\fta$, only if $x$ binds to $\bvlastbit{00}$ or $\bvlastbit{11}$. 
Now, consider our example $\ioex_1$ where the input binds $x$ to the concrete value $\bv{0001}$. 
The reason that, during the previous top-down traversal process, we did not collect $\ftastate_{t_1}^{\bvlastbit{01}}$ is because: $\ioex_1$'s input $\bv{0001}$ is not consistent with any of the abstract values that $\slicingoracle$ maps $\ftastate_{t_1}^{\bvlastbit{01}}$ to. 
On the other hand, we have 
\[
\slicingoracle( \ftastate_{t_1}^{\bvlastbit{00}} ) = 
\big( x = \bvlastbit{00} \big) \lor 
\big( x = \bvlastbit{01} \big) \lor 
\big( x = \bvlastbit{10} \big) \lor 
\big( x = \bvlastbit{11} \big) 
\quad
\slicingoracle( \ftastate_{t_1}^{\bvlastbit{11}} ) = 
\big( x = \bvlastbit{01} \big) \lor 
\big( x = \bvlastbit{10} \big) 
\]
both of which contain a consistent abstract input. That is why both states were collected.  
Note that we check $\slicingoracle$ for final states too: in this case, $\slicingoracle$ maps $\ftastate_{t_2}^{\bvlastbit{10}}$ to an abstract input consistent with $\ioex_1$.

$\slicingoracle$ is essentially a form of \emph{reachability oracle}~\cite{cohen2003reachability,thorup2004compact}, which is designed for our synthesis setting.
It is built by processing $\fta$ in a bottom-up fashion, starting from $t_0$ (i.e., initial) states. 
\[
\slicingoracle( \ftastate_{t_0}^{\bvlastbit{00}} ) = 
\big( x = \bvlastbit{00} \big) 
\quad 
\slicingoracle( \ftastate_{t_0}^{\bvlastbit{01}} ) = 
\big( x = \bvlastbit{01} \big) 
\quad 
\slicingoracle( \ftastate_{t_0}^{\bvlastbit{10}} ) = 
\big( x = \bvlastbit{10} \big) 
\quad 
\slicingoracle( \ftastate_{t_0}^{\bvlastbit{11}} ) = 
\big( x = \bvlastbit{11} \big) 
\]
Then, we build entries for each $t_1$ state, by combining its incoming transitions disjunctively. See $\slicingoracle( \ftastate_{t_1}^{\bvlastbit{00}} ), \slicingoracle( \ftastate_{t_1}^{\bvlastbit{01}} ), \slicingoracle( \ftastate_{t_1}^{\bvlastbit{11}} )$ above as examples. Finally, $t_2$ states are processed in the same way.

\newpara{Online search.}
Our slice $\ftaslicee{\ioex_1}$ is essentially a grammar (with all programs that satisfy $\ioex_1$ abstractly).
Now, given $\ioexs = \{ \ioex_1, \ioex_2 \}$, we can in principle use any syntax-guided solver to search over $\ftaslicee{\ioex_1}$ for programs that concretely satisfy $\ioexs$.
Hence, our framework is not tied to a specific search algorithm. 

%% file: sections/03-alg.tex
\section{Presynthesis-Based Synthesis Framework}
\label{sec:alg}

\subsection{Preliminaries}
\label{sec:alg:prelim}

We first cover necessary formal background on a widely-adopted framework~\cite{wang2017program,wang2017synthesis} for syntax-guided synthesis from examples, that is based on finite tree automata (FTAs) and abstract semantics.

\newpara{Domain-specific language (DSL).} 
Following the seminal work on syntax-guided synthesis~\cite{alur2013syntax}, we assume the DSL's syntax is defined by a context-free grammar $\grammar = ( \nonterminalsymbols, \terminalsymbols, \productions, \startsymbol )$, where

\begin{itemize}[leftmargin=*]

\item 
$\nonterminalsymbols$ is a finite set of non-terminal symbols, that correspond to (sub-)programs in the DSL.

\item 
$\terminalsymbols$ is a finite set of terminal symbols, which correspond to DSL operators, input variables, and constants. 
We use $\allvariablesymbols \subseteq \terminalsymbols$ to denote the set of all input variables. 

\item 
$\productions$ is a set of production rules of the form $\grammarsymbol \rightarrow \dsloperator(\grammarsymbol_1, \mydots, \grammarsymbol_n)$, where $\grammarsymbol, \grammarsymbol_1, \mydots, \grammarsymbol_n$ are non-terminals and $\dsloperator$ is an $n$-arity DSL operator. 
Input variables and constants are viewed as nullary operators. 

\item 
$\startsymbol \in \nonterminalsymbols$ is the start symbol. 

\end{itemize}
As standard in the literature~\cite{wang2017program,solar2008program,polozov2015flashmeta}, we finitize the grammar (e.g., by finite unrolling).

The DSL's (concretize) semantics is assigned by defining, for each DSL operator $\dsloperator$, a concrete transformer $\concretetransformer{\dsloperator}$, which computes the concrete output of $\dsloperator$ given concrete argument values. 
An input variable is evaluated given a binding environment in the usual way.

\newpara{Programming-by-example (PBE).}
The PBE problem is to find a program $\prog$, from a given DSL, that satisfies a given set $\ioexs$ of input-output examples according to the DSL's concrete semantics. 
Each example $\ioex \in \ioexs$ is of the form $\inputex \inputtooutput \outputex$, where $\inputex$ is a binding environment that maps input variables to their values and $\outputex$ is the desired output for $\inputex$.

\newpara{Abstract semantics.}
Following the seminal work on abstract interpretation~\cite{cousot1977abstract},
we define the DSL's abstract semantics using abstract values and abstract transformers.  
In particular, for each production rule $\grammarsymbol \rightarrow \dsloperator(\grammarsymbol_1, \mydots, \grammarsymbol_n)$, its abstract transformer $\abstracttransformer{\dsloperator}$ returns an abstract value $\abstractvalue$ for $\grammarsymbol$ given abstract values $\abstractvalue_1, \mydots, \abstractvalue_n$ for $\dsloperator$'s argument symbols $\grammarsymbol_1, \mydots, \grammarsymbol_n$.
Following prior work~\cite{wang2017program}, we require such transformers to be sound; that is, 
\[
\text{if } 
\abstracttransformer{\dsloperator}( \abstractvalue_1, \mydots, \abstractvalue_n ) = \abstractvalue
\text{ and } 
\concretevalue_1 \in \gamma(\abstractvalue_1), \mydots, \concretevalue_n \in \gamma(\abstractvalue_n), 
\text{ then }
\concretetransformer{\dsloperator}(\concretevalue_1, \mydots, \concretevalue_n) \in \gamma(\abstractvalue). 
\]
Here, $\gamma$ is the concretization function that maps an abstract value to a set of concrete values. 
We also require an abstraction function $\alpha$ that maps a concrete value to an abstract value---which will be used to perform abstraction-based synthesis, as will be shown shortly.

\newpara{Finite tree automata (FTAs).}
FTAs~\cite{comon2008tree} accept trees, which our work will utilize to represent a set of abstract syntax trees (ASTs) for DSL programs. 
More formally, a (bottom-up) FTA is a tuple $\fta = (\ftastates, \ftaalphabet, \ftafinalstates, \ftatransitions)$, 
where $\ftastates$ is a set of states, $\ftaalphabet$ is a (ranked) alphabet, $\ftafinalstates \subseteq Q$ is a set of final states, and $\ftatransitions$ is a set of transitions of the form $\dsloperator(\ftastate_1, \mydots, \ftastate_n) \rightarrow \ftastate$ with $\dsloperator \in \ftaalphabet$ and $\ftastate_1, \mydots, \ftastate_n, \ftastate \in \ftastates$. 

$\fta$ \emph{accepts} a term $t$---which corresponds to a DSL program---if, using $\fta$'s transitions (i.e., rewrite rules), $t$ can be rewritten to one of $\fta$'s final states.  
The language of $\fta$, denoted by $\lang(\fta)$, is the set of terms accepted by $\fta$. 
Given any $t \in \lang(\fta)$, an \emph{accepting run} for $t$ maps each sub-term (which corresponds to a sub-program) of $t$ to a state in $\fta$, and maps $t$ to a final state.

\newpara{Program synthesis using FTAs and abstract semantics.}
A recent line of work~\cite{wang2019efficient} developed a technique to construct an FTA $\fta$ for a given DSL (with a given abstract semantics), such that $\lang(\fta)$ is exactly those programs that satisfy a given input-output example $\ioex$ with respect to the DSL's abstract semantics. 
$\fta$ can then be used for PBE~\cite{wang2017program,wang2017synthesis,wang2018relational,yaghmazadeh2018automated}.
Our work builds upon this approach. 
\autoref{fig:online-fta-rules} gives their FTA construction rules, which are briefly explained below.

\input{alg-figures/online-fta-rules}

In a nutshell, a state $\ftastate_{\grammarsymbol}^{\abstractvalue}$ is constructed in FTA $\fta$, if a (sub-)program at grammar symbol $\grammarsymbol$ can produce abstract value $\abstractvalue$. A transition connects argument states to a return state based on abstract semantics. 
Specifically, the first rule, \textsc{Syn-Var}, creates transitions for input variables, where $\variablesymbol$'s abstract value $\abstractvalue$ is obtained by applying $\alpha$ to the concrete value that $\inputex$ binds $\variablesymbol$ to. 
The \textsc{Syn-Prod} rule grows $\fta$ by processing every production rule. 
Finally, \textsc{Syn-Final} marks a state $\ftastate_{\startsymbol}^{\abstractvalue}$ final, if $\abstractvalue$ is consistent with $\outputex$. 
The soundness of abstract transformers ensures $\fta$ is also sound---that is, every accepting run of $\fta$ corresponds to a DSL program that satisfies $\ioex$ under the given abstraction. 
Since rules in \autoref{fig:online-fta-rules} are applied until saturation, $\lang(\fta)$ has all programs that abstractly satisfy $\ioex$. 
Hence, we can safely prune out programs not in $\fta$, without hindering synthesis completeness.

\newpara{Remarks.}
This FTA $\fta$ essentially models the abstract semantics for the given example, by mapping all programs in the DSL to runs in $\fta$. 
However, only those \emph{accepting} runs represent abstractly valid programs, while the rest cause the FTA bloat issue. Our work addresses exactly this bloat.

\subsection{Top-Level Algorithm}
\label{sec:alg:top-level}

\input{alg-figures/top-level-algs}

\algorithmref{alg:top-level} gives our top-level algorithm with two procedures, $\presynthesisalg$ and $\synthesisalg$.

\newpara{Presynthesis (offline phase).}
$\presynthesisalg$ is performed offline---\emph{once} for a given DSL.
It yields an FTA $\fta$ and its corresponding oracle $\slicingoracle$.
$\fta$ models the abstract semantics for all DSL programs, for \emph{a space of inputs} (which are of interest for future synthesis tasks).
That is,  given a program $\prog$, for any input from the considered space, $\fta$ has a corresponding run for $\prog$ which reveals $\prog$'s output (abstract) value as well as all intermediate ones. 
This is different from all prior work (to our best knowledge), which builds an FTA \emph{only} for those inputs specific to the synthesis task at hand. 
As a result, our offline $\fta$ is expected to be larger than its online counterpart from prior work, but $\fta$ is built during presynthesis. 
$\slicingoracle$ is a data structure, also built offline, that is used to accelerate online synthesis, which we will explain shortly.

\newpara{Synthesis (online phase).}
Before serving any synthesis requests, $\fta$ and $\slicingoracle$ are first loaded into a synthesis server. 
Then, given a synthesis task with a set $\ioexs$ of examples, $\synthesisalg$ will perform $\searchalg$ (as in prior work) to find a program that satisfies $\ioexs$ (with respect to the concrete semantics).
This $\searchalg$, however, is performed over a \emph{slice} $\ftaslicee{\ioex_1}$ that is built for one example $\ioex_1 \in \ioexs$.
$\ftaslicee{\ioex_1}$ is a sub-FTA of our offline $\fta$, and $\lang(\ftaslicee{\ioex_1})$ is all programs that satisfy $\ioex_1$ under abstract semantics. 
$\ftaslicee{\ioex_1}$ is computed using a new $\slicealg$ algorithm that leverages our oracle $\slicingoracle$ to address the FTA bloat.

\newpara{Addressing FTA bloat.}
Prior work can be viewed as also building a slice/FTA that represents all abstractly consistent programs. 
But, as illustrated in~\autoref{sec:overview}, they do it online (i.e., during synthesis), and suffer from the FTA bloat issue (i.e., creating too many unnecessary states/transitions). 
Note that the bloat cannot be addressed by only building an offline $\fta$. 
We must also know how to navigate in $\fta$, to avoid adding unnecessary states to the slice.
Our oracle $\slicingoracle$ is designed exactly for this purpose: our slicing algorithm, $\slicealg$, utilizes $\slicingoracle$ to compute the slice highly efficiently.

In what follows, we first describe how to build offline $\fta$, and then present how slicing works. 
Since $\searchalg$ can be done using any syntax-guided search technique, we skip its presentation.

\subsection{Offline FTA Construction}
\label{sec:alg:fta}

\autoref{fig:offline-fta-rules} gives the construction rules for $\buildftaalg$. 
It shares the same structure as~\autoref{fig:online-fta-rules}, except for two important differences. 
First, \textsc{Presyn-Var} creates  transitions for $\variablesymbol$, for all abstract values in $\abstractinputspace_{\variablesymbol}$ that $\variablesymbol$ can bind to. This is because we consider a space of inputs during presynthesis.
Second, \textsc{Presyn-Final} marks states  associated with the grammar's start symbol $\startsymbol$ as final states, regardless of the abstract value. 
This is because presynthesis does not have access to online examples yet.
On the other hand,  \textsc{Presyn-Prod} remains the same as its counterpart in~\autoref{fig:online-fta-rules}.


In what follows, we highlight a few important characteristics of this offline FTA $\fta$.

\newpara{Non-deterministic FTA to model semantics for all inputs.}
The \textsc{Presyn-Var} rule---which builds \emph{multiple} transitions for the same variable $\variablesymbol$---makes $\fta$ non-deterministic. That is, one program now corresponds to multiple runs in $\fta$. 
This non-determinism is intended, since $\fta$ aims to model the abstract semantics \emph{for all inputs} of interest, rather than for a specific one as in prior work.

\newpara{Compactness by sharing states across both programs and inputs.}
$\fta$ is compact in two ways. 
First, as in prior work, \emph{for the same input}, different (sub-)programs are merged into one state if they return the same output. 
Second, (sub-)programs that yield the same output, on \emph{different inputs}, are also merged.
In other words, state sharing in $\fta$ is not only across programs but also across inputs.

\newpara{Implication on synthesis.}
Since states in $\fta$ are shared across different inputs, simply extracting an \emph{accepting run} (as in prior work~\cite{wang2017synthesis,feng2017component}) is insufficient.
That is, given $\ioex = \inputex \inputtooutput \outputex$, an accepting run $\ftarun$ that ends in a final state $\ftastate_{\startsymbol}^{\abstractvalue}$ , where $\outputex \in \gamma(\abstractvalue)$, does not guarantee that its corresponding program $\prog$ satisfies $\ioex$ under abstract semantics. This is because $\ftarun$ might map an input variable $\variablesymbol$ in $\prog$ to a state $\ftastate_{\grammarsymbol}^{\abstractvalue}$, where 
$\abstractvalue \neq \alpha^{\abstractinputspace_{\variablesymbol}}(\inputex)$.
We will present a solution to this issue in the next \autoref{sec:alg:slicing}.

\input{alg-figures/offline-fta-rules}

\subsection{Slicing}
\label{sec:alg:slicing}

As alluded to at the end of~\autoref{sec:alg:fta}, given $\ioex = \inputex \inputtooutput \outputex$, we are interested in \emph{specific} accepting runs in $\fta$---which not only end in a state consistent with $\outputex$ but also \emph{begin with states that are consistent with $\inputex$.}
Such runs are said to be \emph{consistent with $\ioex$}, which is formally defined below. 

\begin{definition}[Example-consistent run]
Given 
an FTA $\fta$ constructed by~\autoref{fig:offline-fta-rules} for a given DSL and 
an example $\ioex = \inputex \inputtooutput \outputex$, 
consider an accepting run $\ftarun$ in $\fta$ that corresponds to program $\prog$. 
We say $\ftarun$ is consistent with $\ioex$, if the following two conditions are true.
\begin{enumerate}[leftmargin=*]
\item 
$\ftarun$ is output-consistent. 
That is,  
$\ftarun$ maps $\prog$ to a final state $\ftastate_{\startsymbol}^{\abstractvalue}$, such that 
$\outputex \in \gamma(\abstractvalue)$. 
\item 
$\ftarun$ is input-consistent. 
That is, 
$\ftarun$ maps each input variable $\variablesymbol$ in  $\prog$ to $\ftastate_{\grammarsymbol}^{\abstractvalue}$, such that 
$\abstractvalue = \alpha^{\abstractinputspace_{\variablesymbol}}\big(\inputex(\variablesymbol)\big)$, i.e., the abstract value $\abstractvalue$ correctly abstracts the concrete value that $\ioex_1$ binds $\variablesymbol$ to.
\end{enumerate}
It should be obvious that $\ftarun$'s corresponding program $\prog$ must satisfy $\ioex$ according to the DSL's abstract semantics that is used to build $\fta$.
\end{definition}

\newpara{Oracle-guided slicing algorithm.}
Now, we are ready to present our $\slicealg$ algorithm; see~\autoref{fig:online-slicing-rules}. Recall how $\slicealg$ is used in our synthesis algorithm (see~\alglineref{alg:top-level:slice} of~\algorithmref{alg:top-level}): it takes the offline FTA $\fta$ and an example $\ioex$ (and an  oracle $\slicingoracle$), and returns a \emph{slice} $\ftaslicee{\ioex}$ which is formally defined below.

\begin{definition}[Slice]
Given 
an FTA $\fta = (\ftastates, \ftaalphabet, \ftafinalstates, \ftatransitions)$ constructed by~\autoref{fig:offline-fta-rules} for a given DSL and an example $\ioex$, a slice $\ftaslicee{\ioex} = (\ftastates', \ftaalphabet, \ftafinalstates', \ftatransitions')$ is a \emph{sub-FTA} of $\fta$ (that is, $\ftastates' \subseteq \ftastates$, $\ftafinalstates' \subseteq \ftafinalstates$, $\ftatransitions' \subseteq \ftatransitions$), such that an accepting run $\ftarun$ of $\fta$ is consistent with $\ioex$ iff $\ftarun$ is an accepting run of $\ftaslicee{\ioex}$. 
\end{definition}

Because example-consistent runs always correspond to abstractly consistent programs, $\lang(\ftaslicee{\ioex})$ is exactly the set of programs that satisfy $\ioex$ under the DSL's abstract semantics.

\input{alg-figures/slicing-rules}

\input{alg-figures/slicing-oracle-rules}

Let us now explain how~\autoref{fig:online-slicing-rules} extracts $\ftaslicee{\ioex}$ from $\fta$. 
Our basic idea is to find states that are involved in at least one example-consistent run of $\fta$---such states belong to $\ftaslicee{\ioex}$. 
The \textsc{Output} rule first identifies final states of $\ftaslicee{\ioex}$: $\ftastate_{\startsymbol}^{\abstractvalue}$ is a final state of $\ftaslicee{\ioex}$, if (i) $\ftastate_{\startsymbol}^{\abstractvalue}$ is a final state in $\fta$ where $\abstractvalue$ is consistent with the output example $\outputex$; (ii) there is a run in $\fta$ to $\ftastate_{\startsymbol}^{\abstractvalue}$ that is consistent with $\inputex$. 
Note that (ii) is checked by $\inputconsistent$, which additionally takes as input the oracle $\slicingoracle$; we will explain how this check works shortly. 
The \textsc{Transition} rule serves as the recursive case, which grows $\ftaslicee{\ioex}$ by processing transitions in $\fta$. 
Specifically, given a state $\ftastate$ in $\ftaslicee{\ioex}$, we add an incoming transition of $\ftastate$ and its argument states $\ftastate_1, \mydots, \ftastate_n$ to $\ftaslicee{\ioex}$, if $\inputconsistent$ holds for every $\ftastate_i$. 
Finally, when reaching an input variable, the \textsc{Input} rule is used to collect those transitions $\variablesymbol \ftatransitionarrow \ftastate_{\grammarsymbol}^{\abstractvalue}$ where $\abstractvalue$ correctly abstracts $\variablesymbol$'s concrete value in $\inputex$.

Now, let us explain how $\inputconsistent$ works. 
Our approach first constructs an oracle $\slicingoracle$ during presynthesis for the offline FTA $\fta$ (see~\alglineref{alg:top-level:slicing-oracle} in~\algorithmref{alg:top-level}). 
At a high level, $\slicingoracle$ maps every state $\ftastate$ in $\fta$ to a constraint $\inputconsistentconstraint$: 
conceptually, $\inputconsistentconstraint$ describes, for every run $\ftarun$ to $\ftastate$, the \emph{abstract input} $\abstractinputenv$ that $\ftarun$ is consistent with. 
Here, $\abstractinputenv$ is a binding environment that maps input variables to abstract values. 
In other words, if $\slicingoracle(\ftastate)$ does not hold for some $\abstractinputenv$, then we know there is no run in $\fta$ to $\ftastate$ that is consistent with $\abstractinputenv$; this is exactly how $\inputconsistent$ utilizes $\slicingoracle$.
In particular, as shown in~\autoref{fig:online-slicing-rules}, $\inputconsistent$ first abstracts the input example $\inputex$ to an abstract input $\abstractinputenv$. Then, the check simply boils down to checking if $\abstractinputenv$ is a satisfying assignment for $\slicingoracle(\ftastate)$.

\newpara{Building oracle.}
\autoref{fig:slicing-oracle-rules} presents the $\buildslicingoraclealg$ procedure. 
For each state $\ftastate_{\grammarsymbol}^{\abstractvalue}$ in $\fta$, $\inputconsistentconstraint_{\ftatransition}$ creates a constraint for each incoming transition of $\ftastate_{\grammarsymbol}^{\abstractvalue}$, and their disjunction gives the constraint $\inputconsistentconstraint$ for $\ftastate_{\grammarsymbol}^{\abstractvalue}$. Intuitively, this means, $\inputconsistentconstraint$ is true (i.e., there is an input-consistent run to $\ftastate_{\grammarsymbol}^{\abstractvalue}$), as long as one of its incoming transitions lies on an input-consistent run. 
$\inputconsistentconstraint_{\ftatransition}$ considers two cases. 
In the base case where $\ftatransition = \variablesymbol \rightarrow \ftastate_{\grammarsymbol}^{\abstractvalue}$, $\inputconsistentconstraint_{\ftatransition}$ is simply $\variablesymbol = \abstractvalue$, meaning the abstract input must bind $\variablesymbol$ to $\abstractvalue$ in order for an input-consistent run to reach $\ftastate_{\grammarsymbol}^{\abstractvalue}$ via $\ftatransition$. 
In the recursive case, $\inputconsistentconstraint_{\ftatransition}$ conjoins $\slicingoracle(\ftastate_i)$ across all argument states. This is because $\ftastate_{\grammarsymbol}^{\abstractvalue}$ is reachable via $\ftatransition$ only if all of its argument states are reachable.

\newpara{Addressing FTA bloat.}
Let us now take a step back, and revisit the FTA bloat problem. 
The reason we can significantly reduce the bloat is because our slicing process is guarded by $\inputconsistent$, which ensures all states and transitions added to $\ftaslicee{\ioex}$ are necessary. 
One straw-man approach is to remove the $\inputconsistent$ check: (i) first traverse $\fta$ top-down to obtain $\fta'$ that contains all output-consistent runs, and (ii) then retain only runs in $\fta'$ that are input-consistent. 
The problem is that (i) will cause bloat, which is exactly what $\inputconsistent$ addresses. 

\newpara{Implementation.}
Our implementation incorporates a couple of optimizations. 
One is to simplify every constraint $\inputconsistentconstraint$, during $\buildslicingoraclealg$ (see \autoref{fig:slicing-oracle-rules}), into its disjunctive normal form, where each disjunct encodes a valid abstract input. 
This allows $\inputconsistentconstraint$ to be represented as a set of abstract inputs; then $\inputconsistent$ (see \autoref{fig:online-slicing-rules}) can be done by checking ``if $\slicingoracle(\ftastate)$ contains $\alpha(\inputex)$''.
We further optimize the $\textsc{Output}$ rule in \autoref{fig:online-slicing-rules}, to avoid going over every state in $\ftafinalstates$. 
The idea is to build a map $M$ \emph{offline} for $\fta$: $M$ maps each valid abstract input $\abstractinputenv$ to a set $S$ of final states, such that: for any $\abstractinputenv$, $M(\abstractinputenv)$ contains $\ftastate$ iff $\slicingoracle(\ftastate)$ contains $\abstractinputenv$.
Then, \textsc{Output} first does a look-up $M \big( \alpha(\inputex) \big)$ to find input-consistent final states, followed by checking against $\outputex$.

\subsection{Theorem}
\label{sec:alg:theorems}


\begin{theorem}[Soundness and Completeness] 
Suppose $\ftaslicee{\ioex_1}$ is the slice returned by  $\slicealg$ at~\alglineref{alg:top-level:slice} of~\algorithmref{alg:top-level}, given a DSL with abstract semantics and a set $\ioexs = \{ \ioex_1, \mydots, \ioex_n \}$ of examples. 
Then $\lang(\ftaslicee{\ioex_1})$ is exactly the set of DSL programs that satisfy $\ioex_1$ under abstract semantics. 
That is, $\slicealg$ is both sound (i.e., every accepting run of $\ftaslicee{\ioex_1}$ represents a DSL program that satisfies $\ioex_1$ under abstract semantics) and complete (i.e., every DSL program that satisfies $\ioex_1$ under abstract semantics corresponds to an accepting run in $\ftaslicee{\ioex_1}$).
\label{theorem:slice-sound-complete} 
\end{theorem}

    
\begin{proof}
Proofs can be found in~\autoref{sec:appendix:proof}.
\end{proof}

%% file: alg-figures/online-fta-rules.tex
\begin{figure}[!t]
\small 
\centering

\[
\arraycolsep=3pt\def\arraystretch{1}
\begin{array}{c}

\textsc{[Syn-Var]}
\ 
\irule{
\begin{array}{c}
\ 
\variablesymbol  \in \allvariablesymbols
\quad\quad  
\grammarsymbol \productionarrow \variablesymbol \in \productions
\quad\quad
\abstractvalue = \alpha^{\abstractinputspace_{\variablesymbol}} \big( \inputex(x) \big)
\ 
\end{array}
}{
\ftastate_{s}^\abstractvalue \in \ftastates 
\quad 
\variablesymbol \in \ftaalphabet
\quad 
\variablesymbol \ftatransitionarrow \ftastate_{s}^\abstractvalue \in \ftatransitions
}

\quad \quad

\textsc{[Syn-Final]}
\ 
\irule{
\ 
\ftastate_{\startsymbol}^{\abstractvalue} \in \ftastates
\quad\quad
\outputex \in \gamma(\abstractvalue)
\ 
}{
\ftastate_{\startsymbol}^{\abstractvalue} \in \ftafinalstates
}

\\ \\

\textsc{[Syn-Prod]}
\ 
\irule{
\ 
\grammarsymbol \productionarrow \dsloperator( \grammarsymbol_1, \mydots, \grammarsymbol_n ) \in \productions

\quad \quad 

\ftastate_{\grammarsymbol_i}^{\abstractvalue_i} \in \ftastates

\quad \quad 

\abstractvalue = 
\transformer{\dsloperator}(\abstractvalue_1, \mydots, \abstractvalue_n) 

\
}{

\ftastate_{\grammarsymbol}^\abstractvalue \in \ftastates 

\quad

\dsloperator \in \ftaalphabet

\quad 

\dsloperator (\ftastate_{\grammarsymbol_1}^{\abstractvalue_1}, \mydots, \ftastate_{\grammarsymbol_n}^{\abstractvalue_n} ) \ftatransitionarrow \ftastate_{\grammarsymbol}^{\abstractvalue} \in \ftatransitions

}

\end{array}
\]
\caption{Rules for constructing $\fta = ( \ftastates, \ftaalphabet, \ftafinalstates, \ftatransitions)$, 
given 
(i) 
a DSL whose grammar is $\grammar = ( \nonterminalsymbols, \terminalsymbols, \startsymbol, \productions )$ and whose abstract semantics is defined by abstract transformers $\abstracttransformer{\cdot}$
and (ii) 
an input-output example $\ioex = \inputex \inputtooutput \outputex$. $\alpha^{\abstractinputspace_{\variablesymbol}}$ is the abstraction function over the abstract domain of variable $\variablesymbol$, while $\gamma$ is the concretization function.}
\label{fig:online-fta-rules}
\end{figure}

%% file: alg-figures/top-level-algs.tex
\begin{figure}[!t]
\vspace{-5pt}
\small 
\begin{algorithm}[H]
\caption{Top-level algorithm for presynthesis-based synthesis.}
\label{alg:top-level}
\begin{algorithmic}[1]

\Statex\myprocedure{\presynthesisalg($\grammar$)}
\vspace{1pt}

\Statex\myinput{$\grammar$ is a context-free grammar that is associated with abstract semantics.}

\Statex\Output{An offline FTA $\fta$ and its oracle $\slicingoracle$.} 
\vspace{2pt}

\State 
$\fta \assign \buildftaalg( \grammar )$;
\label{alg:top-level:fta}
\State 
$\slicingoracle \assign \buildslicingoraclealg( \fta )$;
\label{alg:top-level:slicing-oracle}

\State 
\Return $( \fta, \slicingoracle )$;

\vspace{5pt}

\Statex\myprocedure{\synthesisalg($\ioexs, \fta, \slicingoracle$)}
\vspace{1pt}

\Statex\myinput{$\ioexs = \{ \ioex_1, \mydots, \ioex_n \}$ is a set of $n$ input-output examples.} 

\Statex\myinput{$\fta$ and $\slicingoracle$ are the FTA and oracle build by presynthesis.}

\Statex\Output{A program $\prog$ that satisfies $\ioexs$ under concrete semantics, or $\mynull$ indicating no such program is found.}
\vspace{2pt}

\State
$\ftaslicee{\ioex_1} \assign \slicealg( \fta, \ioex_1, \slicingoracle )$;
\label{alg:top-level:slice}

\State 
$\prog \assign \searchalg( \ftaslicee{\ioex_1}, \ioexs )$;
\label{alg:top-level:search}
\State
\Return $\prog$;

\end{algorithmic}
\end{algorithm}
\vspace{-10pt}
\end{figure}

%% file: alg-figures/offline-fta-rules.tex
\begin{figure}[!t]
\small 
\centering

\[
\arraycolsep=3pt\def\arraystretch{1}
\begin{array}{c}

\textsc{[Presyn-Var]} 
\ 
\irule{
\ 

\variablesymbol  \in \allvariablesymbols

\quad \quad 

\grammarsymbol \rightarrow \variablesymbol \in \productions

\quad\quad 

\abstractvalue \in \abstractinputspace_{\variablesymbol}

\ 

}{
\ftastate_{\grammarsymbol}^{\abstractvalue} \in \ftastates 

\quad 

\variablesymbol \in \ftaalphabet

\quad 

\variablesymbol \ftatransitionarrow \ftastate_{\grammarsymbol}^{\abstractvalue} \in \ftatransitions

}
\qquad \quad 
\textsc{[Presyn-Final]}
\ 
\irule{
\ftastate_{\startsymbol}^{\abstractvalue} \in \ftastates
}{
\ \ 
\ftastate_{\startsymbol}^{\abstractvalue} \in \ftafinalstates
\ \ 
}

\\ \\

\textsc{[Presyn-Prod]}
\ 
\irule{
\ 
\grammarsymbol \productionarrow \dsloperator( \grammarsymbol_1, \mydots, \grammarsymbol_n ) \in \productions

\quad \quad 

\ftastate_{\grammarsymbol_i}^{\abstractvalue_i} \in \ftastates

\quad \quad 

\abstractvalue = 
\transformer{\dsloperator}(\abstractvalue_1, \mydots, \abstractvalue_n) 

\
}{

\ftastate_{\grammarsymbol}^\abstractvalue \in \ftastates 

\quad

\dsloperator \in \ftaalphabet

\quad 

\dsloperator (\ftastate_{\grammarsymbol_1}^{\abstractvalue_1}, \mydots, \ftastate_{\grammarsymbol_n}^{\abstractvalue_n} ) \ftatransitionarrow \ftastate_{\grammarsymbol}^{\abstractvalue} \in \ftatransitions

}
\end{array}
\]
\caption{Rules for $\buildftaalg$, which construct an FTA $\fta = ( \ftastates, \ftaalphabet, \ftafinalstates, \ftatransitions)$, given 
a DSL whose grammar is $\grammar = ( \nonterminalsymbols, \terminalsymbols, \startsymbol, \productions )$ and whose abstract semantics is defined with abstract transformers $\abstracttransformer{\cdot}$. 
$\abstractinputspace_{\variablesymbol}$ gives the set of abstract values for input variable $\variablesymbol$ (i.e., the space of inputs of interest for $\variablesymbol$).
}
\label{fig:offline-fta-rules}
\end{figure}

%% file: alg-figures/slicing-rules.tex
\begin{figure}[!t]
\small 

\centering
\[
\arraycolsep=3pt\def\arraystretch{1}
\begin{array}{c}

\textsc{[Output]}
\
\irule{

\ \ 

\ftastate_{\startsymbol}^{\abstractvalue} \in \ftafinalstates

\quad \quad 

\outputex \in \gamma(\abstractvalue)

\quad \quad

\inputconsistent(\ftastate_{\startsymbol}^{\abstractvalue}, \inputex, \slicingoracle)

\ \ 

}{
\ftastate_{\startsymbol}^{\abstractvalue} \in \ftastates'
\quad 
\ftastate_{\startsymbol}^{\abstractvalue} \in \ftafinalstates'
} 

\\
\\

\textsc{[Transition]} 
\ 
\irule{
\ \

\ftastate \in \ftastates'

\quad \quad 

\dsloperator ( \ftastate_1, \mydots, \ftastate_n ) \ftatransitionarrow \ftastate \in \ftatransitions

\quad \quad

\forall i \in [1,n]: \inputconsistent(\ftastate_i, \inputex, \slicingoracle)

\ \ 
}{
\ftastate_1, \mydots, \ftastate_n \in \ftastates'
\quad 
\dsloperator ( \ftastate_1, \mydots, \ftastate_n ) \ftatransitionarrow \ftastate \in \ftatransitions'
}

\\ \\ 

\textsc{[Input]}
\
\irule{

\ \ 

\ftastate_{\grammarsymbol}^{\abstractvalue} \in \ftastates'

\quad \quad 

\variablesymbol \ftatransitionarrow \ftastate_{\grammarsymbol}^{\abstractvalue} \in \ftatransitions

\quad \quad 

\abstractvalue = \alpha^{\abstractinputspace_{\variablesymbol}} \big( \inputex(\variablesymbol) \big)

\ \ 

}{
\variablesymbol \ftatransitionarrow \ftastate_{\grammarsymbol}^{\abstractvalue} \in \ftatransitions'
}

\\ \\

\textbf{where}\ \ \inputconsistent(\ftastate, \inputex, \slicingoracle) = 

\Big(
\text{True} 
\ \ 
\text{iff} 
\ \ 
\alpha(\inputex)
\models \slicingoracle(\ftastate) 
\Big)

\ \ 
\text{for }
\alpha(\inputex) = 
\Big\{ 
\variablesymbol \inputtooutput \alpha^{\abstractinputspace_{\variablesymbol}} \big( \inputex(\variablesymbol) \big)
\Big\} 

\end{array}
\]

\caption{Rules for $\slicealg$, which constructs a slice/FTA $\ftaslice = (\ftastates', \ftaalphabet, \ftafinalstates', \ftatransitions')$, given (i) an FTA $\fta = (\ftastates, \ftaalphabet, \ftafinalstates, \ftatransitions)$, (ii) an input-output example $\ioex = \inputex \inputtooutput \outputex$, and (iii) an oracle $\slicingoracle$ for $\fta$.}
\label{fig:online-slicing-rules}
\end{figure}

%% file: alg-figures/slicing-oracle-rules.tex
\begin{figure}[!t]
\small 
\centering

\begin{gather*}
\slicingoracle ( \ftastate_{\grammarsymbol}^{\abstractvalue} )
= \biglor\limits_{\ftatransition \in \ftatransitions}
\inputconsistentconstraint_{\ftatransition}
\quad 
\textbf{where }
\inputconsistentconstraint_{\ftatransition} =
\left\{
\begin{array}{llll}
\variablesymbol = \abstractvalue
& 
\text{ if }
\ftatransition = \variablesymbol \rightarrow \ftastate_{\grammarsymbol}^{\abstractvalue}
\\[5pt]
\bigland_{i \in [1, n]} \slicingoracle( \ftastate_i )
& 
\text{ if }
\ftatransition = \dsloperator(\ftastate_1, \mydots, \ftastate_n) \productionarrow \ftastate_{\grammarsymbol}^{\abstractvalue}
\\
\end{array}
\right.
\end{gather*}
\caption{Rules for $\buildslicingoraclealg$, which constructs an oracle $\slicingoracle$ by processing each state $\ftastate_{\grammarsymbol}^{\abstractvalue}$ in $\fta$.}
\label{fig:slicing-oracle-rules}
\end{figure}

%% file: sections/04-sql.tex

\section{Implementation and Instantiations}
\label{sec:sql}

This section presents an implementation of our presynthesis-based synthesis framework, followed by instantiations of the framework for three different domains. 

\subsection{$\frameworkname$: Implementing Presynthesis-Based Synthesis}
\label{sec:sql:framework-impl}

\input{alg-figures/sql-search-rules}

While the presynthesis module is parametrized with a domain-specific abstraction, our implementation of the synthesis component is domain-general. 
In particular, our $\searchalg$ algorithm is based on the standard bottom-up search technique with equivalence class reduction~\cite{wang2017synthesis}, which we briefly explain below.
Our implementation of the entire framework is called $\frameworkname$.

Recall our $\searchalg$ procedure (see~\alglineref{alg:top-level:search} of~\algorithmref{alg:top-level}): it takes as input a set $\ioexs$ of examples and a slice/FTA $\ftaslicee{\ioex}$ for one example $\ioex \in \ioexs$, and returns a program that satisfies all examples in $\ioexs$ under concrete semantics. 
Our implementation constructs an FTA $\fta'$ given $\ftaslicee{\ioex}$ and $\ioex$, using rules from~\autoref{fig:sql-cfta-rules}. 
This is a bottom-up process which ``concretizes'' $\ftaslicee{\ioex}$ using the DSL's concrete semantics: \textsc{Search-Var} builds transitions for $\variablesymbol$, \textsc{Search-Transition} iteratively grows $\fta'$, and finally \textsc{Search-Final} marks final states. 
Whenever a final state $\ftastate$ is produced, we pause the concretization process, and immediately enumerate accepting runs (i.e., queries) to $\ftastate$: if a program passing $\ioexs$ is found, then $\searchalg$ terminates; 
otherwise, concretization resumes, and the aforementioned process repeats. We return $\mynull$ if no satisfying program is found (within the given timeout).

\subsection{$\sqlname$: Instantiating $\frameworkname$ for SQL}
\label{sec:sql:sql}

We first show our query DSL, then present the abstraction (i.e., abstract domain and transformers) used for presynthesis, and finally explain key optimizations in brief for the SQL domain.

\newpara{Query language.}
\autoref{fig:sql-syntax} gives the syntax (which is an expressive subset of SQL).
A query $\sqlquery$ takes as input one or more table-typed variables, and returns a table.
Operators all have standard semantics. 
$\proj_{\columnlist}(\sqlquery)$ selects columns $\columnlist$ from $\sqlquery$'s result, where $\columnlist$ is a list of column indices. 
$\filter_{\sqlpredicate}(\sqlquery)$ retains rows that satisfy a filter condition $\sqlpredicate$. 
Here, $\sqlpredicate$ supports comparison between two columns, or between a column and a value (string or numeric). 
$\groupby_{\columnlist, \agglist}(\sqlquery)$ first groups all rows from $\sqlquery$ by $\columnlist$, then performs aggregation for each group based on $\agglist$, and finally outputs a table with the aggregated content for all groups.
$\distinct_{\columnlist}(\sqlquery)$ removes duplicate rows according to $\columnlist$.
$\innerjoin_{\sqlpredicate}(\sqlquery_1, \sqlquery_2)$ and $\leftjoin_{\sqlpredicate}(\sqlquery_1, \sqlquery_2)$ perform inner join and left join respectively. 
When $\sqlpredicate$ is used as a join condition, without loss of generality, we assume the column index $\colidx$ that appears on the left (resp. right) hand side of $\logicop$ always refers to a column in $\sqlquery_1$ (resp. $\sqlquery_2$); this is solely to simplify the later presentation of our abstract semantics.

\input{alg-figures/sql-syntax}

\newpara{Abstract domain.}
Like prior work~\cite{wang2017program}, we apply predicate abstraction, which abstracts a concrete value into a \emph{conjunction} of predicates. 
A predicate always refers to a grammar symbol $\grammarsymbol$ whose value is being abstracted, and may also involve constants and input variables. 
Specifically, we define the following predicate \emph{templates} to abstract a query $\sqlquery$'s return table. 
\begin{itemize}[leftmargin=*]

\item 
Our first predicate template is $\prednumofrows^{\sqlquery} ( \sigma_1, \sigma_2)$, 
where  $\sqlquery$ is the \emph{Query} grammar symbol and  $\sigma_i$ is the summation over some input variables and constants. 
This predicate is true, if the number of rows in $\sqlquery$ falls in the range $[ \sigma_1, \sigma_2 ]$.

\item 
Our second template is $\prednumofcols^\sqlquery(k)$, 
where $k$ is a constant. It means ``query $\sqlquery$ yields a table with exactly $k$ columns''. 

\item 
The next one is $\predsubsetofinput^\sqlquery( \setofcols, \sqlinputtablevar)$, 
where $\setofcols$ is a set of column indices (constants) and  $\sqlinputtablevar$ is an input variable.
This predicate is true, if ($\forall j \in \setofcols$)  the $j$-th column  of $\sqlquery$'s return table is a subset of some column from $\sqlinputtablevar$'s table (i.e., the table that $\sqlinputtablevar$ binds to). 

\item 
The last template is $\predcoltype^\sqlquery(\setofcols, t)$, stating that ($\forall j \in \setofcols$) the $j$-th column of $\sqlquery$  has type $t$. Here, $t$ is either numeric or string.

\end{itemize}

Abstracting $\agglist$ and $\columnlist$ boils down to abstracting the column index $\colidx$. 
This is done by a predicate $\predcolrange^{\colidx}(\setofcols)$, meaning $\colidx$ is in the set $\setofcols$ of column indices.
Then, $\columnlist$ is abstracted by $\predlistcolinset^{\columnlist}([ \setofcols_1, \mydots, \setofcols_l ])$, meaning $\columnlist$ has exactly $l$ columns and its $i$-th column ($\forall i \in [1, l]$) is in $\setofcols_i$.
$\agglist$ can be abstracted in a similar manner. 
These abstractions are crucial for propagating $\sqlquery$'s abstract values through query operators (which heavily make use of $\columnlist$ and $\agglist$). 

The join/filter condition $\sqlpredicate$ is abstracted using $\predcolsused^{\sqlpredicate}( \maxcol_1, \maxcol_2 )$, where $\maxcol_1$ (resp. $\maxcol_2$) is the max column index involved in $\sqlpredicate$ that appear on the left (resp. right) hand side of $\logicop$. 
Specifically, any $\colidx$ used in $\sqlpredicate$ is still abstracted by $\predcolrange$.
Then, $\maxcol_1$ (resp. $\maxcol_2$) for $\sqlpredicate$ is computed by taking the max of $\predcolrange$ predicates for each $\colidx$ that is on the left (resp. right) hand side of any $\logicop$ in $\sqlpredicate$.

Finally, $\sqlname$'s abstract input space (namely, $\abstractinputspace_{\variablesymbol}$ in~\autoref{fig:offline-fta-rules}) considers up to 3 input variables each with up to 15 columns. In other words, $k$ in $\prednumofcols$ ranges from 1 to 15.

\input{alg-figures/sql-transformers}

\newpara{Abstract transformers.}
Following prior work~\cite{wang2017program},  we use a generic transformer for each operator $\dsloperator$, which takes an abstract value---that is a \emph{conjunction} of predicates---for each of $\dsloperator$'s arguments and returns an abstract value.  Therefore, it suffices to define an abstract transformer for every operator $\dsloperator$ given an \emph{atomic} predicate for each argument.~\autoref{fig:sql-transformers} gives a representative list of transformers. 
\[\small
\abstracttransformer{\dsloperator}
\Big( 
\bigland_{i_1} \varphi_{i_1}, 
\mydots, 
\bigland_{i_n} \varphi_{i_n}
\Big)
\assign 
\bigland_{i_1} \mydots \bigland_{i_n}
\abstracttransformer{\dsloperator}
(
\varphi_{i_1}, \mydots, \varphi_{i_n}
)
\vspace{-5pt}
\]

\newpara{Optimizations.}
$\sqlname$ incorporates a few optimizations to improve the performance for both presynthesis and synthesis.
First, we reduce the abstract domain by defining a canonical order on column indices in $\predsubsetofinput$ and $\predcoltype$ (e.g., numeric  columns before string columns). 
This does not hinder completeness, since the desired order can always be recovered with a projection, but it significantly speeds up presynthesis.
On top of this canonicalization, we slightly adjust our $\searchalg$ to find a query $\sqlquery$ whose return table \emph{subsumes} all columns in the example output table, and then infer the top-level projection given $\sqlquery$.
Our final optimization is around further accelerating presynthesis. 
Recall that a set $\setofcols$ of columns is tracked in our abstraction. The canonical column order allows us to represent $\setofcols$ as an interval, making presynthesis faster without losing completeness.

\subsection{$\stringname$: Instantiating $\frameworkname$ for String Transformation}

While SQL is the focus domain, to fully demonstrate our generality, we also briefly present two more instantiations, one for string transformation and another for matrix manipulation (next section). 
We use DSLs from the $\blaze$ work~\cite{wang2017program}, and focus our presentation on the abstraction. 

Same as $\blaze$, we use a predicate template to track the length of strings. Our second template $\fromInput^S(J, x)$ tracks if the $j$-th character ($\forall j \in J$) of string $S$ appears in the input variable $x$.
The third and last template $\charEqual^S(J, x, k)$ refines $\fromInput^S(J, x)$: it requires the $j$-th character in $S$ to match the $k$-th character in $x$. 
The atomic transformers can be found in~\autoref{sec:appendix:stringtransformer}.



\subsection{$\matrixname$: Instantiating $\frameworkname$ for Matrix Manipulation}

Following $\blaze$, we have a predicate template $\dimNum^T(k)$ to track the number of dimensions for a matrix/tensor, and another $\shapeOf^T(v)$ that tracks the tensor's shape (i.e., the $i$-th dimension of tensor $T$ is $v[i]$). 
We also have a third template $\elementAt^T(\vout, x, \vin)$, indicating the element in $T$ at position $\vout$ is equal to the element in an input tensor $x$ at position $\vin$. 
The corresponding abstract transformers can be found in~\autoref{sec:appendix:matrixtransformer}.

%% file: alg-figures/sql-search-rules.tex
\begin{figure}[!t]
\small 
\centering

\[
\arraycolsep=3pt\def\arraystretch{1}
\begin{array}{c}

\textsc{[Search-Var]}
\ 
\irule{
\begin{array}{c}
\ 
\variablesymbol \ftatransitionarrow \ftastate_{\grammarsymbol}^{\abstractvalue} \in \ftatransitions
\quad\quad  
\concretevalue = \inputex(\variablesymbol)
\ 
\end{array}
}{
\ftastate_{\grammarsymbol}^{\concretevalue} \in \ftastates' 
\quad 
\variablesymbol \ftatransitionarrow \ftastate_{\grammarsymbol}^{\concretevalue} \in \ftatransitions'
}

\quad \quad

\textsc{[Search-Final]}
\ 
\irule{
\ 
\ftastate_{\startsymbol}^{\concretevalue} \in \ftastates'
\quad\quad
\concretevalue = \outputex 
\ 
}{
\ftastate_{\startsymbol}^{\concretevalue} \in \ftafinalstates'
}

\\ \\

\textsc{[Search-Transition]}
\ 
\irule{
\ 
\dsloperator (\ftastate_{\grammarsymbol_1}^{\abstractvalue_1}, \mydots, \ftastate_{\grammarsymbol_n}^{\abstractvalue_n} ) \ftatransitionarrow \ftastate_{\grammarsymbol}^{\abstractvalue} \in \ftatransitions
\quad \quad 
\ftastate_{\grammarsymbol_i}^{\concretevalue_i} \in \ftastates'
\quad \quad 
\concretevalue_i \in \gamma(\abstractvalue_i)
\quad \quad 
\concretevalue = 
\concretetransformer{\dsloperator}(\concretevalue_1, \mydots, \concretevalue_n) 

\
}{

\ftastate_{\grammarsymbol}^{\concretevalue} \in \ftastates' 
\quad
\dsloperator (\ftastate_{\grammarsymbol_1}^{\concretevalue_1}, \mydots, \ftastate_{\grammarsymbol_n}^{\concretevalue_n} ) \ftatransitionarrow \ftastate_{\grammarsymbol}^{\concretevalue} \in \ftatransitions'

}

\end{array}
\]
\caption{
Rules for constructing a (concrete) FTA $\fta' = ( \ftastates', \ftaalphabet, \ftafinalstates', \ftatransitions')$, 
given a slice/FTA $\ftaslicee{\ioex} = ( \ftastates, \ftaalphabet, \ftafinalstates, \ftatransitions)$ for an input-output example $\ioex = \inputex \inputtooutput \outputex$.
}
\label{fig:sql-cfta-rules}
\end{figure}

%% file: alg-figures/sql-syntax.tex
\begin{figure}[!t]
\small 
\centering
\[
\begin{array}{rrcl}

\textit{Query} &
\sqlquery &
\grammareq &

\sqlinputtablevar 
\ | \ 
\proj_{\columnlist}(\sqlquery) 
\ | \ 
\filter_{\filterpred}(\sqlquery) 
\ | \ 
\groupby_{\columnlist, \agglist}(\sqlquery) 
\\

& & | &
\distinct_{\columnlist}(\sqlquery)  
\ | \ 
\innerjoin_{\joinpred}(\sqlquery,\sqlquery) 
\ | \ 
\leftjoin_{\joinpred}(\sqlquery, \sqlquery) 

\\

\textit{Column List} & 
\columnlist & 
\grammareq &
[ \colidx, \mydots, \colidx ]

\\

\textit{Aggregation List} &
\agglist &
\grammareq &
[ \aggexpr, \mydots, \aggexpr ]

\\

\textit{Aggregation} &
\aggexpr &
\grammareq &
\agg(\colidx)

\\

\textit{Condition} &
\sqlpredicate &
\grammareq &
\colidx \logicop \colidx
\ | \ 
\colidx \logicop \sqlvalue
\ | \ 
\sqlpredicate \land \sqlpredicate
\ | \ 
\sqlpredicate \lor \sqlpredicate

\\

\textit{Operator} &
\logicop &
\grammareq &
\leq  
\ | \ 
< 
\ | \ 
= 
\ | \ 
\neq 
\ | \ 
> 
\ | \ 
\geq 
\ | \ 
\like

\end{array}
\]
\[
\begin{array}{c}
\sqlinputtablevar \in \textbf{Input Variables} 
\quad 
\colidx \in \textbf{Column Indices} 
\quad 
\sqlvalue \in \textbf{Values}
\quad
\agg \in \{ \sqlmin, \sqlmax, \sqlsum, \sqlcount, \sqlavg \}
\end{array}
\]
\vspace{-10pt}
\caption{Syntax of our query DSL (which is an expressive subset of SQL).}
\label{fig:sql-syntax}
\vspace{-10pt}
\end{figure}

%% file: alg-figures/sql-transformers.tex
\begin{figure}[!t]
\fontsize{8.5pt}{10pt}\selectfont
\centering
\renewcommand{\arraystretch}{1.2}
\setlength{\arraycolsep}{1pt} 
\[
\begin{array}{rlll}

\abstracttransformer{\proj}
\big( 
\prednumofcols^{\sqlquery}( k ),  \ 
\predlistcolinset^{\columnlist} ( [ \setofcols_1, \mydots, \setofcols_l ] ) 
\big)
& \assign 
& \prednumofcols^{\sqlquery'} ( l ) 
\ \  
\text{if } 
\bigland_{i \in [1,l]}
\big( \max ( \setofcols_i ) \leq k \big)
\\ 

\abstracttransformer{\proj}
\big( 
\predsubsetofinput^{\sqlquery}( \setofcols, \sqlinputtablevar),  \ 
\predlistcolinset^{\columnlist} ( [ \setofcols_1, \mydots, \setofcols_l ] ) 
\big)
& \assign 
& \predsubsetofinput^{\sqlquery'} (\setofcols', \sqlinputtablevar) 
\ \ 
\text{where } \setofcols' = \{ i \ | \ \setofcols_i \subseteq \setofcols \}
\\ 

\abstracttransformer{\proj}
\big( 
\predcoltype^{\sqlquery}( \setofcols, t),  \ 
\predlistcolinset^{\columnlist} ( [ \setofcols_1, \mydots, \setofcols_l ] ) 
\big)
& \assign 
& \predcoltype^{\sqlquery'} (\setofcols', t) 
\ \ 
\text{where } \setofcols' = \{ i \ | \ \setofcols_i \subseteq \setofcols \}
\\[4pt]

\abstracttransformer{\filter}
\big( 
\prednumofrows^{\sqlquery}(   \sigma_1, \sigma_2  ),  \ 
\sqlpredicate
\big)
& \assign 
& \prednumofrows^{\sqlquery'}(  0, \sigma_2  ) 
\\

\abstracttransformer{\filter}
\big( 
\prednumofcols^{\sqlquery}( k ),  \ 
\predcolsused^{\sqlpredicate} ( \maxcol_1, \maxcol_2 ) 
\big)
& \assign 
& \prednumofcols^{\sqlquery'} ( k ) 
\ \  
\text{if } \max(\maxcol_1, \maxcol_2) \leq k
\\

\abstracttransformer{\filter}
\big( 
\predcoltype^{\sqlquery}( \setofcols, t),  \ 
\sqlpredicate
\big)
& \assign 
& \predcoltype^{\sqlquery'}( \setofcols, t) 
\\[4pt]

\abstracttransformer{\innerjoin}
\big( 
\prednumofrows^{\sqlquery_1}( \sigma_{1, 1}, \sigma_{1, 2} ),  \ 
\prednumofrows^{\sqlquery_2}( \sigma_{2, 1}, \sigma_{2, 2} ),  \ 
\sqlpredicate
\big)
& \assign 
& \prednumofrows^{\sqlquery'}(  0, \sigma_{1, 2} \times \sigma_{2, 2}  ) 
\\

\abstracttransformer{\innerjoin}
\big( 
\prednumofcols^{\sqlquery_1}( k_1 ),  \ 
\prednumofcols^{\sqlquery_2}( k_2 ),  \ 
\predcolsused^{\sqlpredicate} ( \maxcol_1, \maxcol_2 ) 
\big)
& \assign 
& \prednumofcols^{\sqlquery'} ( k_1 + k_2 ) 
\ \  
\text{if } \maxcol_i \leq k_i 
\\

\abstracttransformer{\innerjoin}
\big( 
\predcoltype^{\sqlquery_1}( \setofcols, t),  \ 
\sqlquery_2, \ 
\sqlpredicate
\big)
& \assign 
& \predcoltype^{\sqlquery'}( \setofcols, t) 
\\[4pt]

\abstracttransformer{ \predcolrange^{\colidx_1}(\setofcols_1) \logicop \predcolrange^{\colidx_2}(\setofcols_2) }
& \assign 
& \predcolsused^{\sqlpredicate} \big( \max(\setofcols_1), \max(\setofcols_2) \big)
\\

\abstracttransformer{ \predcolsused^{\sqlpredicate_1} ( \maxcol_{1,1}, \maxcol_{1, 2} ) \land \predcolsused^{\sqlpredicate_2} ( \maxcol_{2,1}, \maxcol_{2,2} ) }
& \assign 
& \predcolsused^{\sqlpredicate} \big( \max( \maxcol_{1, 1}, \maxcol_{2, 1} ), \max( \maxcol_{1, 2}, \maxcol_{2, 2} ) \big)
\\

\end{array}
\]
\vspace{-10pt}
\caption{Representative abstract transformers for operators in \autoref{fig:sql-syntax}.}
\label{fig:sql-transformers}
\vspace{-5pt}
\end{figure}

%% file: sections/05-eval.tex

\section{Evaluation}
\label{sec:eval}

This section presents a series of experiments designed to answer the following questions. 
Note that the first four questions are centered around our focus domain (i.e., SQL), while RQ5 is designed to evaluate the generality of $\frameworkname$ (in string and matrix domains).

\begin{itemize}[leftmargin=*]
\item 
\textbf{RQ1}: 
How effective is $\sqlname$ at solving example-based SQL synthesis tasks?
\item 
\textbf{RQ2}: 
How does $\sqlname$ compare with state-of-the-art techniques? 
\item 
\textbf{RQ3}:
How important are the key ideas (namely, presynthesis and oracle)?
\item 
\textbf{RQ4}: 
How does $\sqlname$ scale as the underlying abstract semantics becomes finer?
\item 
\textbf{RQ5}: 
How does $\stringname$ and $\matrixname$ perform in practice?
\end{itemize}

\newpara{SQL Benchmarks.}
We collected all \ruichecked{3,866} benchmarks from $\cubes$~\cite{brancas2024towards,brancas2022cubes,cubes-tool}, which, to the best of our knowledge, contains all example-based SQL synthesis benchmarks in the literature~\cite{orvalho2020squares,patsql,tran2009query,wang2017synthesizing} to date. 
Each benchmark consists of a ``ground-truth'' query $P$ and one input-output example.
We augmented each benchmark with \ruichecked{16} additional examples, 
following the same setup in $\cubes$~\cite{brancas2022cubes}. 
This is because the original example led to many synthesized queries that are unintended (i.e., not equivalent to $P$).
At the end, we curated \ruichecked{3,817} benchmarks\footnote{\ruichecked{49 benchmarks were not included, since their ground-truth queries were not executable on at least one of the inputs.}}, each with \ruichecked{17} examples in total.

\newpara{String and matrix benchmarks.}
We use the same benchmarks from prior work~\cite{wang2017program}. In particular, we have 108 $\sygus$ string transformation benchmarks and 39 matrix manipulation benchmarks.

\newpara{Testbed.}
Our experiments are conducted on machines with Xeon Gold 6154 CPU and 192GB RAM running Red Hat 8.10. 
SQLite 3.46 is used to evaluate SQL queries. 

\subsection{RQ1: $\sqlname$ Results}
\label{sec:eval:RQ1}

\newpara{Setup.}
To evaluate $\sqlname$, we first run its presynthesizer to build an offline FTA $\fta$ and the oracle $\slicingoracle$. 
Then, given a benchmark with a set $\ioexs$ of 17 examples, we run $\sqlname$'s synthesizer (given $\fta$, $\slicingoracle$, and $\ioexs$), and record: 
(i) whether or not the benchmark is solved (i.e., if a query that satisfies $\ioexs$ is found) within \ruichecked{our timeout (10 minutes)}, 
and (ii) if so, the synthesis time. 
During this process, we also log detailed information (such as time breakdown and FTA size).

\input{eval-figures/RQ1}

\newpara{End-to-end results.}
Let us first go over $\sqlname$'s end-to-end performance; see ~\autoref{table:RQ1:end-to-end}. 
Overall, it successfully solved \ruichecked{2,604} benchmarks using a median of \ruichecked{0.01} seconds. 
\ruichecked{82\%} of them can be solved within \ruichecked{5 seconds}. 
The average solving time is \ruichecked{22.0} seconds.

\newpara{Synthesis results.}
\autoref{table:RQ1:slice} and~\autoref{table:RQ1:search} further present the synthesis time breakdown.
While $\slicealg$ always terminates in \ruichecked{0.5} seconds (see~\autoref{table:RQ1:slice}), $\searchalg$ time tends to be longer (see~\autoref{table:RQ1:search}). Despite this, $\searchalg$ can still finish in \ruichecked{5 seconds} for \ruichecked{82\%} of the benchmarks.

Let us dive into more details, beginning with~\autoref{table:RQ1:search}. 
Recall that $\searchalg$ (i) first builds a concrete FTA (CFTA), based on concrete semantics, for the given slice, and (ii) then enumerates accepting runs of this CFTA and along the way checking their corresponding queries against all examples, until a satisfying program is found. 
It turns out that, for most of the benchmarks, over \ruichecked{60\%} of the total $\searchalg$ time ends up being spent on SQLite-related work. 
The resulting CFTA is typically small, with a median of \ruichecked{160} states; \ruichecked{and for over 72\% of the benchmarks, it has under 10K states.}
Not all such states, however, lead to a final state. Only those accepting runs (i.e., reaching a final  state) are enumerated and checked.
For most benchmarks, the first accepting run yields a query that passes all \ruichecked{17} examples; \ruichecked{and for over 88\% of them, no more than 1,000 queries are checked.}

Now let us turn to \autoref{table:RQ1:slice}.
The main take-away is that our slices are small, with a median of \ruichecked{37} states and \ruichecked{47} transitions. 
For over \ruichecked{90\%} of the benchmarks, the slice has under  \ruichecked{1,000 states and 3,000 transitions.}
In addition to the small size, slicing does not require running abstract transformers---we just collect states/transitions from the offline FTA.
This also explains why $\slicealg$ is very fast.

\newpara{Presynthesis results.}
As shown in \autoref{table:RQ1:presyn}, our offline FTA can be built in \ruichecked{5.0 hours}, consuming \ruichecked{51.0 GB} of disk space. Its oracle takes \ruichecked{7.2 hours}  to be constructed and occupies \ruichecked{16.8GB} of disk space. 
We believe this resource requirement is affordable.

\newpara{Failure analysis.}
Despite significantly stronger performance than state-of-the-art, there are 1,213 benchmarks (out of 3,817 total) that $\sqlname$ was not able to solve (within 10 minutes). 
A closer inspection reveals that DSL expressiveness is a major reason, accounting for 653 of the unsolved benchmarks. Notable unsupported features include HAVING clause (149 benchmarks), scalar subqueries (98), INTERSECT (97), EXCEPT (66), and complex LIMIT (60). 
Another common cause is FTA concretization timeout (543 benchmarks). Ground-truth queries for them typically have complicated structures (such as many levels of JOINs and complex filter predicates), and some of them involve long SQLite execution time (due to large input/intermediate tables).
Other reasons include non-determinism in certain operators (such as ORDER BY LIMIT, which, in the presence of ties, may return a different output table from the user-provided one), and abstract input space coverage (e.g., input example has more columns than the max our offline supports).

\subsection{RQ2: $\sqlname$ vs. State-of-the-Art}
\label{sec:eval:RQ2}

\input{eval-figures/RQ2}

\newpara{Baselines.}
We compare $\sqlname$ with \emph{all} state-of-the-art \emph{bespoke} SQL synthesizers (to the best of our knowledge).
The first one is $\patsql$~\cite{patsql}, which combines program sketching and enumeration.
The other is $\cubes$~\cite{brancas2024towards,cubes-tool,brancas2022cubes}, which uses SMT-based enumeration with SQL-specific pruning and search hints (such as aggregation functions to be used; which we do not require).

Besides SQL synthesizers, we also include a third baseline, called $\blazesql$, which implements the abstraction-refinement-based synthesis framework from the $\blaze$ paper~\cite{wang2017program}.
In particular, $\blazesql$ first builds an FTA $\fta_0$ using an initial (coarse) abstraction.
$\fta_i$ is then refined into $\fta_{i+1}$, by incorporating new predicates  mined from a spurious program extracted from $\fta_i$. 
This iterative refinement process terminates when a program that passes all 17 examples is extracted. 
$\blazesql$ uses the same DSL as $\sqlname$, as well as the same predicate templates and transformers.

\newpara{Setup.}
We use the same setup (as in RQ1) to run all baselines \ruichecked{(using the same 10-minute timeout)}. For each baseline, we record the solved benchmarks and their corresponding solving times.

\newpara{Results.}
\autoref{fig:RQ2:solved-varying-timeout} summarizes the comparison. 
$\sqlname$ consistently solves more benchmarks than all baselines. 
Using a 600s timeout, it solves \ruichecked{15\%} more than the best baseline, $\cubes$. 
This improvement is significant for a few reasons. 
First, $\sqlname$ is an instantiation of our framework with a pretty standard (online) search technique, whereas $\cubes$ is specialized to SQL and incorporates domain-specific guidance.
Second, $\cubes$ uses additional hints (e.g., candidate aggregation functions), which $\sqlname$ does not use. 
Finally, for those benchmarks solved by both, $\sqlname$ is significantly faster, with a median of \ruichecked{$\speedup{90}$} speed-up and a geometric mean of \ruichecked{$\speedup{51}$}.
On the other hand, $\cubes$ can solve 235 benchmarks that $\sqlname$ cannot. 87 of them involve unsupported features (e.g., certain set operations). For the rest, $\sqlname$ timed out during concretization. This suggests clear room to further improve our approach.

The next best baseline is $\patsql$. $\sqlname$ significantly outperforms it by solving \ruichecked{33\%} more benchmarks.
These results highlight the advantage of abstraction-based pruning and our presynthesis-based technique over purely online synthesizers.
While $\blazesql$ utilizes abstract semantics, it actually solves the least number of benchmarks. In particular, it spends up to \ruichecked{391K} refinement steps to solve a benchmark, with an average of \ruichecked{8K} steps across  solved benchmarks. 
This may suggest that, in our case, predicates discovered locally do not prune away many programs.

\input{eval-figures/RQ3}

\subsection{RQ3: Ablation Studies}
\label{sec:eval:RQ3}

\newpara{Ablations.}
Our ablations aim to evaluate the two underpinning ideas of our framework---namely, presynthesis and oracle-based slicing.
The first ablation, called $\toolnopresynthesis$, removes the presynthesis step: that is, it performs $\slicealg$, given $\ioex_1$ and the grammar, by building $\ftaslicee{\ioex_1}$ \emph{from scratch}. In other words, $\toolnopresynthesis$ is a standard FTA-based synthesizer using abstract semantics (see~\autoref{sec:overview:fta-synthesis}).
Our second ablation, $\toolnoslicingoracle$, still has presynthesis but  does not build the oracle. Instead, the slice $\ftaslicee{\ioex_1}$ is computed by traversing the offline FTA using two passes. The first bottom-up pass obtains all states reachable from $\ioex_1$'s input, followed by a second top-down pass which removes unnecessary states (that do not reach an output-consistent final state).

\newpara{Setup.}
We use the same setup as in RQ2 to run all ablations.

\newpara{$\sqlname$ vs. $\toolnopresynthesis$.}
See results in \autoref{table:RQ3}.
Let us first compare $\sqlname$ with $\toolnopresynthesis$.
Overall, $\sqlname$ solves more benchmarks in less time. This is primarily due to $\toolnopresynthesis$'s significant slowdown on $\slicealg$. Building the slice online during synthesis can take up to \ruichecked{247} seconds. 
As illustrated in~\autoref{sec:overview:bloat}, this slowdown is caused by the FTA bloat issue; that is, an excessive number of unnecessary states and transitions are created (online).
A more detailed analysis of this bloat will be presented in the next~\autoref{sec:eval:RQ4}.

\newpara{$\sqlname$ vs. $\toolnoslicingoracle$.}
We observe a similar comparison with $\toolnoslicingoracle$, which is  also due to $\slicealg$ being slow.
Careful readers might wonder why it is even slower than $\toolnopresynthesis$, given that $\toolnoslicingoracle$ has additional access to an offline FTA $\fta$ (while $\toolnopresynthesis$ builds the slice  \emph{from scratch}). 
The reason is that the offline $\fta$ includes states for all inputs (not just $\ioex_1$'s input). This makes $\fta$ large, to the extent that computing $\ftaslicee{\ioex_1}$ by \emph{traversing} $\fta$ (without our oracle) ends up being more expensive than building $\ftaslicee{\ioex_1}$ from scratch! 
Specifically, recall that $\toolnoslicingoracle$ uses a bottom-up pass to collect states $\ftastate$ in $\fta$ that are reachable from $\ioex_1$'s input. During this pass, it iteratively processes $\ftastate$'s outgoing transitions to discover new reachable states. 
But the vast majority of such transitions are ``illegal'': that is, at least one argument state is not reachable from $\ioex_1$'s input.
Processing such transitions ends up wasting much of the time.
So our take-away is: only building an offline $\fta$ is far from sufficient; we must also use $\fta$ wisely (our oracle seems to be a good way).

\newpara{Sensitivity to abstract input space.}
Recall that $\sqlname$ considers up to 3 input tables each with up to 15 columns. To understand the sensitivity of $\sqlname$'s performance to this tuning knob, we create two ablations, one by shrinking the space down to up to 10 columns and the other with up to 20. 
The latter is able to express only a handful of additional benchmarks, but the corresponding offline FTA is $\speedup{1.3}$ bigger than $\sqlname$ and oracle $\speedup{2.1}$ bigger. 
The former, on the other hand, solves 2,478 benchmarks (95\% of those solved by $\sqlname$), whose offline FTA is $\speedup{3.4}$ smaller and oracle $\speedup{7.1}$ smaller.
Further reducing the abstract input space (e.g., up to 5 columns) would significantly hinder the solving capability.

\subsection{RQ4: Scaling}
\label{sec:eval:RQ4}

\newpara{Our hypothesis: ``scaling law'' for synthesis.}
Akin to the ``scaling law''~\cite{kaplan2020scaling} in deep learning---that is, model performance seems to keep improving as model size, dataset size, and computational resources spent on training increase---we hypothesize that a similar phenomenon also exists in the context of (search-based) program synthesis. 
Specifically, our hypothesis is that, in principle, finer-grained abstraction should lead to faster synthesis, \emph{provided sufficient computational resources}.
We believe that presynthesis is a promising building block towards realizing this ``scaling law'' for synthesis. 
Hence, this section will present empirical data on how synthesis scales as a function of abstraction granularity, aiming to provide preliminary validation for our hypothesis.

\input{eval-figures/RQ4-figure}

\newpara{Setup.}
We create 5 versions of $\sqlname$ with \emph{increasing} abstraction granularity. 
Specifically, we consider the following ordered list of predicate templates:
\[
\prednumofrows, \quad 
\prednumofcols, \quad 
\predsubsetofinput, \quad 
\predcoltype.
\]
$\toolkpred{k}$ ($k$ ranging from 0 to 4) is a variant of $\sqlname$, that uses the first $k$ templates to build the offline FTA. 
For instance, $\toolkpred{0}$ essentially performs $\searchalg$  on the given grammar, without any abstraction-based pruning. 
$\toolkpred{2}$ tracks a table's ``shape'' (i.e., number of rows and columns) as the abstraction, but not table content. 
$\toolkpred{4}$ (i.e., $\sqlname$) uses the finest abstraction, which additionally tracks if a column is a subset of some input table and the column's data type.

As a comparison, we also create 5 versions of $\toolnopresynthesis$ in the same fashion as above, where $\toolnopresynkpred{k}$ uses the same abstraction as $\toolkpred{k}$. This allows us to see how a purely online FTA-based approach scales, and compare it with the presynthesis-based counterpart.

Note that $\toolkpred{0}$ and $\toolnopresynkpred{0}$ are the same---because they both directly perform $\searchalg$  over the grammar (without presynthesis). 
They also serve as a very good ``baseline'', which performs standard bottom-up search (using FTA-based equivalence-class reduction)~\cite{wang2017synthesis}. 
Thus, we will normalize other tools' results to theirs, to facilitate easy inspection of the scaling trend.

Finally, for each $\toolkpred{k}$, we first run its presynthesizer and then use its synthesizer to solve benchmarks (using the same 10-minute timeout). 
The same timeout is used for $\toolnopresynkpred{k}$.

\input{eval-figures/RQ4-cfta-table}

\newpara{$\slicealg$ time and total time vs. abstraction granularity.}
\autoref{fig:RQ4:scaling:slicing} presents how $\slicealg$ scales as the abstract semantics becomes more granular, and \autoref{fig:RQ4:scaling:total} shows how the total time scales. 
As shown by \autoref{fig:RQ4:scaling:slicing}, $\toolnopresynthesis$'s $\slicealg$ time skyrockets as $k$ increases, whereas $\sqlname$ is much more scalable with consistently low slicing time.
This also explains why $\sqlname$'s total time continues to improve (see \autoref{fig:RQ4:scaling:total}) as the abstraction granularity increases, whereas $\toolnopresynthesis$'s total time has a U-shaped scaling curve.

\input{eval-figures/RQ4-offline-scaling}

\newpara{Peak slice size vs. abstraction granularity.}
\autoref{table:RQ4} presents how the peak slice size changes as the abstract semantics is more granular.
As we can see, $\sqlname$'s slice size remains quite stable, whereas $\toolnopresynthesis$'s explodes with increasing $k$. 
This is exactly the FTA bloat phenomenon that  \autoref{sec:overview:bloat} illustrates: 
(i) while $\toolnopresynkpred{k}$ eventually gives the same slice as $\toolkpred{k}$, it creates many unnecessary states and transitions, causing the bloat; 
(ii) this bloat is exacerbated, when finer-grained abstract semantics is used.
This is exactly why $\toolnopresynthesis$'s slicing time scales poorly in \autoref{fig:RQ4:scaling:slicing}, eventually leading to the U-shaped scaling curve of total time in \autoref{fig:RQ4:scaling:total}.

\newpara{Offline scaling.}
Switching gears,~\autoref{table:RQ4-offline} presents $\sqlname$'s offline scaling. Comparing adjacent columns in~\autoref{table:RQ4-offline:scaling}, one might notice the offline FTA grows unevenly. For example, the number of transitions from 1 to 2 increases roughly by $\speedup{30}$, but from 2 to 3 only about $\speedup{2}$, and then from 3 to 4 around $\speedup{80}$. 
This is because from 2 to 3, the abstraction precision increment is not as much as the other two, with 3 to 4 being the most. 
To better quantify the granularity change, instead of using only $k$, we measure the ``state split factor'' from $k$ to $k+1$. That is, for each state $\ftastate$ in $\fta_k$ (offline FTA built with $k$ templates), how many states $\ftastate'$ in $\fta_{k+1}$ that $\ftastate$ get split to (i.e., the abstract value in $\ftastate'$ refines that in $\ftastate$).~\autoref{table:RQ4-offline:split-factor} reports some statistics. As we can see, the split factor grows the most from 3 to 4, which explains the uneven FTA growth observed in~\autoref{table:RQ4-offline:scaling}.

\input{eval-figures/RQ5}

\subsection{RQ5: $\stringname$ and $\matrixname$ Results}
\label{sec:eval:RQ5}

\autoref{table:RQ5} reports main results for $\stringname$ and $\matrixname$. 

\newpara{Comparing with $\blaze$.}
For both domains, our technique (using the finest abstraction; i.e., $k=3$) solves the same set of benchmarks as $\blaze$~\cite{wang2017program} and is faster. 

\newpara{Importance of presynthesis.}
Presynthesis is crucial, as $\toolnopresynthesis$ is significantly slower. 

\newpara{Online and offline scaling.}
A similar online scaling trend as in SQL domain is observed: increasing $k$ leads to faster online synthesis. 
On the other hand, the offline cost also grows, which is consistent with the offline scaling behavior for SQL. 
Overall, for both domains, the offline overhead is affordable even with our finest (i.e., $k=3$) abstraction.

%% file: eval-figures/RQ1.tex
\begin{table}[!t]
\small 
\centering
\caption{$\sqlname$ results.
First, the end-to-end benchmark-solving results are summarized in~\autoref{table:RQ1:end-to-end}.
Then, $\slicealg$ and $\searchalg$ (i.e., two main components of synthesis) results are given in~\autoref{table:RQ1:slice} and~\autoref{table:RQ1:search}, respectively.
Finally,~\autoref{table:RQ1:presyn} reports presynthesis results (for $\buildftaalg$ and $\buildslicingoraclealg$).}
\label{table:RQ1}

\begin{subtable}{.27\linewidth}
\centering
\caption{End-to-end results.  Time statistics are calculated over all solved benchmarks.}
\label{table:RQ1:end-to-end}

\begin{tabular}{ c | cc }

\toprule

\multirow{2}{*}{\#Solved} & 
\multicolumn{2}{c}{Time (sec)}

\\ 

&
\emph{median} & \emph{avg}  

\\ 
\midrule

\ruichecked{2,604} & \ruichecked{0.01} & \ruichecked{22.0}

\\ 
\bottomrule

\end{tabular}
\end{subtable}
\hfill
\begin{subtable}{0.7\linewidth}
\centering
\caption{$\slicealg$ results across all solved benchmarks. We first present the (median, average, max) slicing time. Then, we give the (median, average, max) number of states in the slice, followed by the same statistics for transitions.}
\label{table:RQ1:slice}

\begin{tabular}{ ccc | ccc | ccc }

\toprule

\multicolumn{3}{c|}{Total $\slicealg$ time (sec)} & 
\multicolumn{3}{c|}{\#States} & 
\multicolumn{3}{c}{\#Transitions} 

\\ 

\emph{median} & \emph{avg} &  \emph{max} & 
\emph{median} & \emph{avg} &  \emph{max} & 
\emph{median} & \emph{avg} & \emph{max}

\\
\midrule

\ruichecked{0.001} & \ruichecked{0.01} &  \ruichecked{0.5} & 
\ruichecked{37} & \ruichecked{234} & \ruichecked{5.0K} & 
\ruichecked{47} & \ruichecked{1.6K} & \ruichecked{93.2K} 

\\ 

\bottomrule

\end{tabular}
\end{subtable}

\vspace{15pt}

\begin{subtable}{\linewidth}
\centering
\caption{$\searchalg$ results across all solved benchmarks. 
We first present the (median, average, max)  total search time across solved benchmarks. 
We then report the (median, average, max) percentage of the total time spent on SQLite-related work (e.g., evaluating queries and communicating with SQLite).
Next, ``\#States in CFTA'' presents the (median, average, max) number of states, upon termination of $\searchalg$, that are created while concretizing the slice into a concrete FTA (CFTA). 
Finally, ``\#Queries checked'' reports the (median, average, max) number of accepting runs of the CFTA that are enumerated, each corresponding to a query.}
\label{table:RQ1:search}

\begin{tabular}{ ccc | ccc | ccc | ccc }

\toprule

\multicolumn{3}{c|}{Total $\searchalg$ time (sec)} & 
\multicolumn{3}{c|}{SQLite time / total time} & 
\multicolumn{3}{c|}{\#States in CFTA} & 
\multicolumn{3}{c}{\#Queries checked} 

\\ 

\emph{median} & \emph{avg} &  \emph{max} & 
\emph{median} & \emph{avg} &  \emph{max} & 
\emph{median} & \emph{avg} &  \emph{max} & 
\emph{median} & \emph{avg} &  \emph{max}

\\
\midrule

\ruichecked{0.01} & \ruichecked{22.0} &  \ruichecked{581.2} & 
\ruichecked{64.7}\% & \ruichecked{60.7}\%  & \ruichecked{96.6}\% & 
\ruichecked{160}& \ruichecked{108.2K} & \ruichecked{8.6M} & 
\ruichecked{1} &  \ruichecked{32.8K} & \ruichecked{6.9M}

\\ 

\bottomrule

\end{tabular}
\end{subtable}

\vspace{15pt}

\begin{subtable}{\linewidth}
\centering
\caption{Presynthesis results. For $\buildftaalg$, we report the FTA construction time and the resulting FTA's size. For $\buildslicingoraclealg$, we report the oracle construction time and its size.}
\label{table:RQ1:presyn}

\begin{tabular}{ cccc | cc }

\toprule

\multicolumn{4}{c|}{$\buildftaalg$} 
& 
\multicolumn{2}{c}{$\buildslicingoraclealg$} 

\\

Time & Size on disk & \#States & \#Transitions & 
\ \ Time & Size on disk \

\\ 
\midrule

\ruichecked{5.0 hrs} & \ruichecked{51.0 GB} & \ruichecked{454.8K} & \ruichecked{192.6M} & 
\ruichecked{7.2 hrs} & \ruichecked{16.8 GB} 

\\ 
\bottomrule

\end{tabular}
\end{subtable}

\end{table}

%% file: eval-figures/RQ2.tex
\input{eval-figures/RQ2-figure}

%% file: eval-figures/RQ2-figure.tex
\begin{figure}[!t]
\centering
\includegraphics[height=4.2cm]{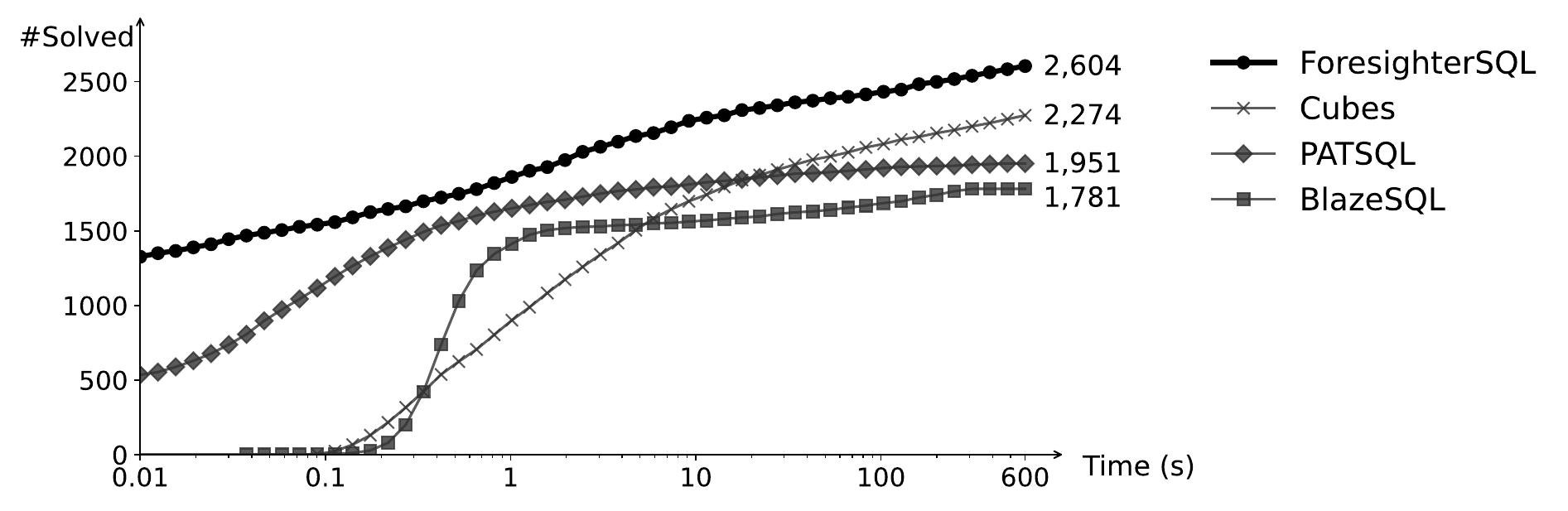}
\captionsetup{skip=10pt} 
\caption{Number of solved benchmarks when varying the timeout (up to 600 seconds), for $\sqlname$ and all baselines.
A data point (X, Y) means Y benchmarks can be solved using X seconds as the timeout. Note that y-axis is in linear scale while we use log scale for x-axis. 
We also label the number of benchmarks solved (with 600 seconds as the timeout) for each tool, right next to its line.}
\label{fig:RQ2:solved-varying-timeout}
\end{figure}

%% file: eval-figures/RQ3.tex
\begin{table}[!t]
\centering
\small 
\caption{$\sqlname$ vs. ablations. We first present the number of solved benchmarks. Then, we report the (median, average) synthesis time across all solved benchmarks. The next column compares $\sqlname$ with each ablation over benchmarks solved by both. Finally, we present the $\slicealg$ time for each tool.}
\label{table:RQ3}

\begin{tabular}{ r | c | rr | c | rrr }

\toprule

& 
\multirow{2}{*}{\#Solved} & 
\multicolumn{2}{c|}{Total time (sec)} & 
\multirow{2}{*}{\shortstack{Median speed-up\\over ablation}} & 
\multicolumn{3}{c}{$\slicealg$ time (sec)}

\\ 
& & 
\emph{median} & \emph{avg} &  
&
\emph{median} & \emph{avg} &  \emph{max}

\\ \midrule

$\sqlname$ &
\ruichecked{2,604} & 
\ruichecked{0.01} &  \ruichecked{22.0} & 
-- & 
\ruichecked{0.001} &  \ruichecked{0.01} & \ruichecked{0.5} 

\\ 

$\toolnopresynthesis$ &  
\ruichecked{2,587} & 
\ruichecked{9.3} &  \ruichecked{40.1} & 
$\ruichecked{\speedup{413.1}}$ & 
\ruichecked{9.3}& \ruichecked{21.5} &  \ruichecked{246.8}

\\ 

$\toolnoslicingoracle$ & 
\ruichecked{2,539} & 
\ruichecked{34.1} &  \ruichecked{85.0} & 
$\ruichecked{\speedup{614.4}}$ & 
\ruichecked{34.0} & \ruichecked{73.6} & \ruichecked{405.5}

\\
\bottomrule

\end{tabular}
\end{table}

%% file: eval-figures/RQ4-figure.tex
\begin{figure}[!t]
\centering
\begin{minipage}[b]{.49\linewidth}
\centering
\includegraphics[height=5cm]{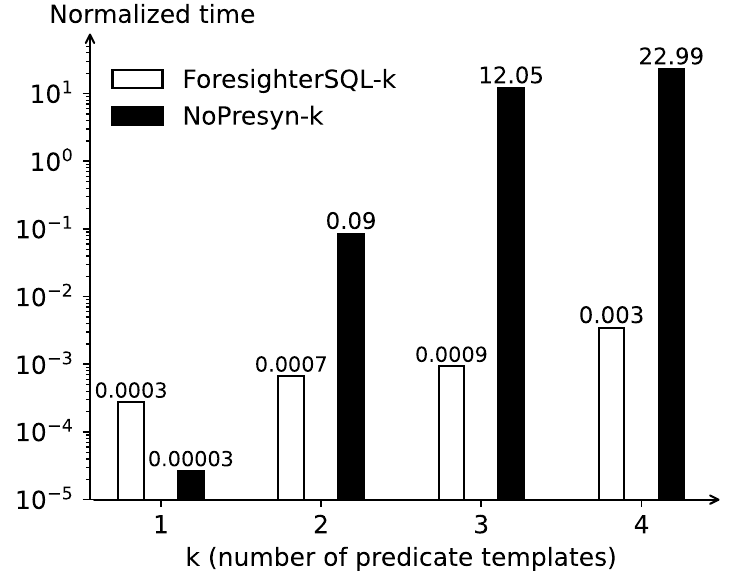}
\captionsetup{skip=1pt} 
\subcaption{Normalized $\slicealg$ time (y-axis) vs. $k$ (x-axis).}
\label{fig:RQ4:scaling:slicing}
\end{minipage}
\hfill
\begin{minipage}[b]{.49\linewidth}
\centering
\includegraphics[height=5cm]{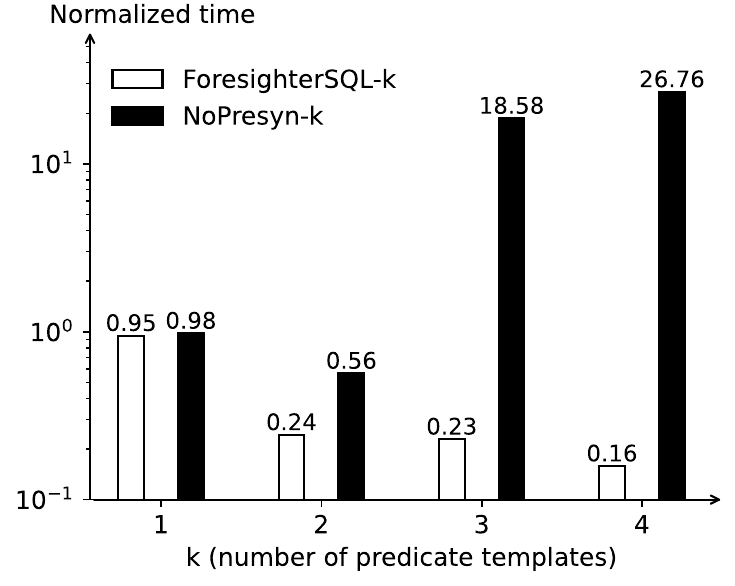}
\captionsetup{skip=1pt} 
\subcaption{Normalized total time (y-axis) vs. $k$ (x-axis).}
\label{fig:RQ4:scaling:total}
\end{minipage}
\caption{Online scaling for $\toolkpred{k}$ and $\toolnopresynkpred{k}$, with $k$ from 1 to 4. \autoref{fig:RQ4:scaling:slicing} shows the scaling for $\slicealg$, and \autoref{fig:RQ4:scaling:total} presents the scaling for the end-to-end synthesis process (i.e., $\slicealg + \searchalg$). Both figures consider benchmarks solved by $\toolkpred{4}$ (which solves the most). For an unsolved benchmark, we use the timeout (600s) as the ``solving time'' to plot \autoref{fig:RQ4:scaling:total}, and use the corresponding slicing time to plot \autoref{fig:RQ4:scaling:slicing}. For each benchmark, all tools' times (both slicing and total) are normalized to $\toolnopresynkpred{0}$'s \emph{total} time. Finally, for both figures, we present the geometric mean of normalized times across all benchmarks. For instance, $\toolkpred{4}$'s normalized total time is \ruichecked{0.16} in \autoref{fig:RQ4:scaling:total}, which means its total time is 16\% of $\toolnopresynkpred{0}$'s (representing a geometric mean of \ruichecked{$\speedup{6.3}$} speed-up).
Finally, $k=0$ is not presented: their slicing time should be interpreted as 0 and total time as 1, after normalization.}
\label{fig:RQ4:scaling}
\end{figure}

%% file: eval-figures/RQ4-cfta-table.tex
\begin{table}[!t]
\centering
\small
\caption{Peak slice size change, when increasing $k$ from 1 to 4, for $\toolkpred{k}$ and $\toolnopresynkpred{k}$. 
The ``Peak \#states'' column reports that, for each tool, the number of states \emph{ever} created by  $\slicealg$.
Same as \autoref{fig:RQ4:scaling}, we consider all benchmarks solved by $\toolkpred{4}$. 
Given a benchmark, each tool's peak number of states is normalized to $\toolkpred{0}$'s number.
Finally, for each tool, we report the geometric mean across all benchmarks. 
The ``Peak \#transitions'' column is calculated in the same manner but for transitions.}
\label{table:RQ4}

\begin{tabular}{ r | rrrr | rrrr }

\toprule

& 
\multicolumn{4}{c|}{Peak \#states} & 
\multicolumn{4}{c}{Peak \#transitions}

\\ 
& 
$k = 1$ & $k = 2$ &  $k = 3$  & $k = 4$ & 
$k = 1$ & $k = 2$ &  $k = 3$  & $k = 4$

\\ \midrule

$\toolkpred{k}$ &

1.15 & 1.06 & 1.13 & 1.37 & 
1.11 & 0.98 & 1.06 & 1.50

\\ 

$\toolnopresynkpred{k}$ &  
2.16 & 16.83 & 468.42 & 1533.07 & 
2.31 & 17.24 & 404.89 & 1445.50

\\
\bottomrule

\end{tabular}

\end{table}

%% file: eval-figures/RQ4-offline-scaling.tex
\begin{table}[!t]
\centering
\small
\caption{Offline scaling for $\sqlname$.}
\label{table:RQ4-offline}

\begin{subtable}{.61\linewidth}
\centering
\caption{Offline artifact scaling when increasing $k$ from 1 to 4.}
\vspace{6pt}
\label{table:RQ4-offline:scaling}

\begin{tabular}{ c | r | cccc }

\toprule & & 
$k = 1$ & $k = 2$ &  $k = 3$  & $k = 4$ 

\\ 
\midrule

\multirow{4}{*}{\makecell[r]{Offline \\ FTA}}
&
\#States &
1.5K & 21.2K & 62.0K & 454.8K 

\\ 

& 
\#Transitions &  
36.7K & 1.1M & 2.4M & 192.6M

\\ 

& 
Size on disk & 
4MB & 238MB & 504MB & 51.0GB

\\ 

& 
Build time & 
0.2s & 7s & 27s & 5.0h

\\ 
\midrule

\multirow{2}{*}{\makecell[r]{Offline \\ oracle}}

&
Size on disk &
5MB & 13MB & 4.7GB & 16.8GB

\\ 

&
Build time &
1s & 33s & 71s & 7.2h 

\\
\bottomrule

\end{tabular}

\end{subtable}
\hfill
\begin{subtable}{.37\linewidth}
\centering
\caption{State split factor statistics from $k$ to $k+1$, across all states in offline FTA $\fta_k$ with $k$ templates. For instance, increasing $k$ from 1 to 2, a state in $\fta_1$ is split into a median of 17 states in $\fta_2$.}
\label{table:RQ4-offline:split-factor}

\begin{tabular}{r | rrr}

\toprule 
 & $1 \rightarrow 2$ & $2 \rightarrow 3$ & $3 \rightarrow 4$
\\ 
\midrule

median 
& 17 & 2 & 72 
\\

75th 
& 45 & 4 & 882 
\\

90th 
& 79 & 6 & 3,375
\\

\bottomrule

\end{tabular}
\end{subtable}
\end{table}

%% file: eval-figures/RQ5.tex
\begin{table}[!t]
\fontsize{8.5pt}{10pt}\selectfont
\centering
\caption{$\stringkpred{k}$ and $\matrixkpred{k}$ results, in comparison with $\blaze$ and $\toolnopresynthesis$.}
\label{table:RQ5}

\setlength{\tabcolsep}{3pt}
\begin{tabular}{ c | r | rr | rrrr | rr }

\toprule

\multicolumn{2}{c|}{} 
& \multicolumn{2}{c|}{\textbf{Online}} & 
\multicolumn{6}{c}{\textbf{Offline}} 
\\ 
\cmidrule{3-10}
\multicolumn{2}{c|}{}
& 
\multicolumn{2}{c|}{$\synthesisalg$} 
& 
\multicolumn{4}{c|}{$\buildftaalg$} 
& 
\multicolumn{2}{c}{$\buildslicingoraclealg$} 
\\ 
\multicolumn{2}{c|}{}
& 
\#Solved & Avg. time (s) & 
Time (s) & Size (disk) & \#States & \#Transitions & 
Time (s) & Size (disk)

\\ 
\midrule\midrule

\multirow{6}{*}{\rotatebox{90}{\textbf{String}}}
& 
$k=1$ & 
82 & 0.751 & 
0.01 & 1.0 MB & 764 & 17.8K & 
0.03 & 212.6 KB

\\ 

& 
$k=2$ & 
90 & 0.107 & 
4.3 & 338.0 MB & 1.6M & 3.3M & 
6.9 & 427.4 MB

\\ 

& 
$k=3$ & 
90 & 0.094 & 
504.0 & 8.0 GB & 35.4M & 75.1M & 
116.6 & 6.5 GB

\\ 
\cmidrule{2-10}

& 
$\blaze$ & 

90 & 0.296 
& -- & -- & -- & -- 
& -- & -- 

\\
\cmidrule{2-10}

& 
$\toolnopresynthesis$
& 

90 & 3.883
& -- & -- & -- & -- 
& -- & --

\\ 
\midrule\midrule

\multirow{6}{*}{\rotatebox{90}{\textbf{Matrix}}}
& 
$k=1$ & 
39 & 12.376 & 
0.002 & 52.8 KB & 118 & 685 & 
0.0006 & 6.92 KB

\\ 

& 
$k=2$ & 
39 & 0.101 & 
66.6 & 3.8 GB & 307.7K & 67.2M & 
177.7 & 206.8 MB

\\ 

& 
$k=3$ & 
39 & 0.086 & 
472.1 & 20.4 GB & 4.1M & 354.7M & 
424.0 & 1.4 GB

\\ 
\cmidrule{2-10}

& 
$\blaze$ & 

39 & 1.449 
& -- & -- & -- & -- 
& -- & -- 

\\
\cmidrule{2-10}

& 
$\toolnopresynthesis$
& 

39 & 6.698
& -- & -- & -- & -- 
& -- & --

\\ 
\bottomrule

\end{tabular}
\end{table}

%% file: sections/06-related.tex

\section{Related Work}
\label{sec:related}


\newpara{Abstraction-based synthesis.}
Abstract semantics has been widely used in the program synthesis literature, across a broad range of domains, such as data transformation~\cite{feser2015synthesizing,feng2017component,wang2017program}, SQL query generation~\cite{wang2017synthesizing}, syntax-guided synthesis~\cite{guria2023absynthe,10.1145/3591288}, among others~\cite{guo2019program,10.1145/3632858,farzan2021counterexample}.
Building upon the literature, we contribute a new method that models a DSL's abstract semantics in an offline phase.

\newpara{Representation-based synthesis.}
Our abstract semantics modeling is carried out by constructing an FTA $\fta$, such that every accepting run of $\fta$ reflects a program's execution on a potential abstract input.
This is related to a long line of work on representation-based synthesis, including the seminal FlashFill work~\cite{gulwani2011automating} (which uses version space algebras as the underlying representation) among many others~\cite{wang2017synthesis,wang2017program,wang2018relational,miltner2022bottom,koppel2022searching,li2024efficient,gvero2013complete,mandelin2005jungloid,sypet,guo2019program} (which utilize other representations, such as simple graphs, hypergraphs, Petri nets, and FTAs).
Our novelty lies in \emph{when} to build the representation. We build it offline for semantics modeling, whereas virtually all prior work builds it online.

\newpara{Graph reachability in program synthesis.}
Our FTAs can be viewed as a form of hypergraphs, and our synthesis algorithm can be viewed as solving a form of graph search problem. 
Hence, our approach is also related to prior work that uses a graph-based formulation of synthesis. 
Notable examples include 
\textsc{Prospector}~\cite{mandelin2005jungloid}, 
\textsc{Blaze}~\cite{wang2017program}, \textsc{SyPet}~\cite{sypet}, 
\textsc{Hoogle+}~\cite{guo2019program}, 
among others~\cite{reinking2015type,joshi2002denali}. 
Besides technical differences (e.g., the specific problem formulation), a major distinction lies in the scale of the graph. 
For instance, unlike \textsc{Prospector} which builds a ``small'' graph (8 MB on disk), our FTA is significantly larger (51 GB on disk). 
This motivates new graph search techniques, as opposed to reusing standard ones (as in \textsc{Prospector}).

\newpara{Reachability oracle.}
Given this large offline-built FTA, our online search algorithm leverages a form of reachability oracle~\cite{cohen2003reachability,thorup2004compact} (a precomputed data structure to help efficiently answer online queries about if a path exists between any two nodes in a graph).
Such oracles have been used for efficiently processing queries over massive databases (e.g., in biomedical and social networks)~\cite{sakr2012overview,jin2008efficiently,valstar2017landmark,peng2020answering,chen2022dlcr,li2022fast,li2021complexity,li2022efficient}.
To the best of our knowledge, we are the first to develop reachability oracles for FTAs and in the context of program synthesis.

\newpara{Preprocessing.}
At a high level, our work leverages preprocessing (or precomputation) to accelerate online synthesis.
This is related to a large body of work that applies preprocessing in different areas (e.g., quantum computing~\cite{huggins2023accelerating,ravi2022cafqa}, cryptography~\cite{rainbow-table,window-method}, and algorithms~\cite{thorup2005approximate,patrascu2010distance}).
Preprocessing has also been used in programming languages---e.g., to help   efficiently answer queries about context-free language reachability~\cite{li2022fast,li2021complexity,li2022efficient} in the context of program analysis.
In addition to precomputing our oracle, building the offline FTA can also be viewed as preprocessing (i.e., transforming a given DSL into a large FTA).
To our best knowledge, we are the first to use preprocessing for (search-based) synthesis.
While \emph{neural} synthesis techniques (e.g., based on LLMs) can also be viewed as having a preprocessing component (i.e., training the neural model), they are quite different from our work: we model the DSL's abstract semantics, whereas they model the distribution of programs.

%% file: sections/07-conc.tex
\section{Conclusion and Discussion}
\label{sec:conc}

We developed the first \emph{presynthesis-based} synthesis framework,  and implemented it in $\frameworkname$. 
The evaluation on three domains demonstrated promising results. 
We hope that our work could fuel the next paradigm shift, from crafting sophisticated online synthesis algorithms to using compute resources (e.g., offline time/space, online memory/parallelism), for further synthesis scale-up. 


\newpara{Necessity of presynthesis.}
Thoughtful readers might wonder if our offline presynthesis phase is necessary: for instance, would a purely online approach that \emph{incrementally} grows a ``global'' FTA (shared across tasks) \emph{on-the-fly}, for every incoming synthesis task, be fast enough?
We suspect that a trivial incremental approach would be much slower (as it still suffers from the FTA bloat issue). A better solution might exist, but we believe it is non-trivial and its performance is also unclear. For example, using an \emph{overly engineered} abstract input space that covers \emph{exactly} our SQL benchmarks (i.e., ``overfitting''), the resulting FTA takes 70 minutes to build. 
This time would be spread across all benchmarks, causing a significant slowdown for many of them. 
Note that this time is the theoretical lower bound, \emph{assuming no additional overhead}---which seems quite hard to realize. 
For instance, one challenge is, given a new task, how to maximally avoid visiting existing transitions and creating additional ones.  
A further problem is how to incrementally update the oracle on-the-fly. 
We believe these challenges could motivate new incremental algorithms for synthesis (potentially built upon our presynthesis idea), which is an interesting future direction to pursue.

%% file: sections/ack.tex
\section*{Acknowledgments} \label{sec:ack}

We would like to thank the anonymous PLDI reviewers, for their insightful feedback. 
We thank Greg Bodwin and Nicole Wein, for helping us navigate the literature on graph reachability oracle and providing useful pointers. 
We would also like to extend our thanks to Chenglong Wang and Yuepeng Wang, for their comments on an earlier draft of this work; and to Anhong Guo, Nicole Wein, and Xu Wang, for their help with experimentation resources. 
This research was supported by the National Science Foundation (NSF) under Grant Numbers 
CCF-2210832, CCF-2318937, CCF-2236233, and CCF-2123654. 
Zheng Guo was supported by Schmidt Sciences.

%% file: sections/appendix-proof.tex
\section{Proof of \autoref{theorem:slice-sound-complete}}
\label{sec:appendix:proof}
This section supplies the proof for \autoref{theorem:slice-sound-complete}, which states the soundness and completeness of our algorithm.
For clarity, let us start by re-stating some important notations.
\begin{itemize}
    \item Given an FTA $\mathcal A = (Q, F, Q_f, \Delta)$ and a (sub-)program $P$, an FTA \emph{run} $\rho$ is a mapping $\textit{Sub}(P) \mapsto Q$ compatible with $\Delta$; in other words, for every sub-program $P' = f(P_1', \mydots, P'_n)$ in $P$, $f(\rho(P_1'), \mydots, \rho(P'_n)) \to \rho(P)$ must always be a transition in $\Delta$.
    \item We call $\rho$ \emph{an accepting run} if it ends at a final state; that is, we have $\rho(P) \in Q_f$ for the associated program $P$.
    \item Recall that in our offline FTA, every state is in the form of $q_s^a$ for a grammar symbol $s$ and an abstract value $a$. 
    Then, we call a run $\rho$ (of the offline FTA) \emph{input-consistent} with an input example $\inputex$ if $\rho$ maps each input variable $x$ in $P$ to $q_s^a$ such that $a = \alpha^{\abstractinputspace_{\variablesymbol}} ( \inputex(\variablesymbol) )$.
\end{itemize}

\newpara{Lemma 1.} Our proof relies on three auxiliary lemmas.
First, \lemmaref{lemma:IO-consistent} states that any input-consistent runs must also be consistent with the abstract output of their associated programs.

\begin{lemma} \label{lemma:IO-consistent} Let $\mathcal A$ be the offline FTA constructed over abstract semantics $\abstracttransformer{\cdot}$, and $\rho$ be a run in $\mathcal A$ that is input-consistent with an input example $\inputex$.
Let $P$ be the program associated with $\rho$, and $q_{s}^a$ be $\rho(P)$ (i.e., the state that $\ftarun$ ends at).
Then, the abstract value $a$ must be equal to the abstract output $\abstracttransformer{P}(\inputex)$.
\end{lemma}
\begin{proof}
We prove by induction on the structure of the associated program $P$. It can either be a variable $x$ or a composite program $f(P_1, \mydots, P_n)$. 
\begin{itemize}
    \item When $P$ is a variable $x$, let $q_s^a$ be the state $\rho(x)$. By the input-consistency of $\rho$, this abstract value $a$ must be $\alpha^{\mathcal D_x}(\inputex(x))$, thus must match the abstract output of $P$.
    \item When $P$ is a composite program $f(P_1, \mydots, P_n)$, let $\rho_i$ be the sub-run of $\rho$ that is associated with $P_i$, $q_{s_i}^{a_i}$ be the end state of $\rho_i$ (i.e., $\rho(P_i)$), and $q_s^a$ be the end state of $\rho$ (i.e., $\rho(P)$). Then, $\rho_i$ is also input-consistent with $\inputex$ (by the input-consistency of $\rho$), and $f(q_{s_1}^{a_1}, \mydots, q_{s_n}^{a_n}) \to q_s^a$ must be a transition in $\mathcal A$ (by the definition of FTA runs). 
    Note that this transition can only be constructed by the \textsc{Presyn-Prod} rule in \autoref{fig:offline-fta-rules}.
    Therefore, by the process of this rule, the induction hypothesis, and the definition of the abstract transformer, we have:
    $$
    a = \abstracttransformer{f}(a_1, \mydots, a_n) = \abstracttransformer{f} \Big( \abstracttransformer{P_1}(\inputex), \mydots, \abstracttransformer{P_n}(\inputex) \Big) = \abstracttransformer{P}(\inputex)
    $$
\end{itemize}
Therefore, the induction holds, and the lemma is proved. 
\end{proof}

\newpara{Lemma 2.} Then, the second lemma \autoref{lemma:offline-fta-completeness} shows the \textit{completeness} of the offline FTA; that is, the offline FTA always comprises an input-consistent run for any program and any input.

\begin{lemma}
\label{lemma:offline-fta-completeness}
    let $\mathcal A = (Q, F, Q_f, \ftatransitions)$ be the offline FTA constructed from grammar $\grammar = (\nonterminalsymbols, \terminalsymbols, \startsymbol, \productions )$ and its associated abstract semantics $\abstracttransformer{\cdot}$.
    Then, for any program $P \in \lang(\grammar)$ and any input $\inputex$, there must be an accepting run $\rho$ in $\mathcal A$ for $P$, such that is $\ftarun$ \emph{input consistent} with $\inputex$. 
\end{lemma}
\begin{proof} We prove the following claim, which generalizes the lemma to sub-programs. 
\begin{itemize}
\item \textbf{Claim.} For any grammar symbol $s \in V$ and any sub-program $P$ that can be expanded from $s$, there must be a run $\rho$ in $\mathcal A$ for $P$, such that (1) $\ftarun$ is input-consistent with $\inputex$, and (2) $\ftarun$ ends at a state $q_s^a$. 
\end{itemize} 
To see how this claim generalizes the original lemma, let us assume $s$ is the start symbol $s_0$. Then, according to the claim: for any program in $P$, there must be a corresponding input-consistent run $\rho$ that ends at a state $q_{s_0}^a$.
All such states must be final states in $Q_f$ by the \textsc{Presyn-Final} rule (\autoref{fig:offline-fta-rules}); hence, $\rho$ must also be an accepting run. Therefore, the original lemma is implied.

We prove this claim by induction on the structure of the sub-program $P$. This program $P$ can either be a variable $x$ or a composite program $f(P_1, \mydots, P_n)$.
\begin{itemize}
    \item When $P$ is a variable $x$, there must be production $s \to x$ in $\productions$. Then, by the \textsc{Presyn-Var} rule, the offline FTA must comprise a transition $x \to q_s^a$ for any abstract value $a$. 
    By taking $a$ as $\alpha^{\mathcal D_x}(\inputex(x))$, this single transition forms an input-consistent run $\rho$ of $P$ that ends at $q_s^a$.
    \item When $P$ is a composite program $f(P_1, \mydots, P_n)$, there must be a production $s \to f(s_1, \mydots, s_n)$ in $\productions$, where each $P_i$ is expanded from $s_i$.
    By the induction hypothesis, for each $P_i$, there must be an input-consistent run $\rho_i$ of $P_i$ that ends at a state $q_{s_i}^{a_i}$ for some abstract value $a_i$.
    Then, by the \textsc{Presyn-Prod} rule, the offline FTA must comprise a transition $f(q_{s_1}^{a_1}, \mydots, q_{s_n}^{a_n}) \to q_s^a$ for $a = \abstracttransformer{f}(a_1, \mydots, a_n)$.
    Clearly, the union of all $\rho_i$ and $\rho(P) \mapsto q_{s}^a$ forms an input-consistent run $\rho$ of $P$ that ends at $q_s^a$.
\end{itemize}

Therefore, the induction holds, then the claim is proved. Hence the original lemma.
\end{proof}

\newpara{Lemma 3.} Our last lemma \lemmaref{lemma:oracle-preciseness} states the preciseness of our oracle. That is, the oracle constraint holds if and only if an input-consistent run exists.

\begin{lemma}
\label{lemma:oracle-preciseness}
Let $\mathcal A$ be the offline FTA, and $\slicingoracle$ be the oracle constructed for $\mathcal A$.
Then, for any state $q$ and any input $\inputex$, $\alpha(\inputex) \models \mathcal O(q)$ (i.e., \textup{\textsc{InputConsistent}}$(q, \inputex, \mathcal O)$ in \autoref{fig:slicing-oracle-rules}) if and only if there exists an input-consistent run $\rho$ that ends at $q$.
\end{lemma}
\begin{proof} We prove the two directions separately. 

For the \emph{if} direction, we need to show that the oracle constraint holds if an input-consistent run $\rho$ exists. We prove this by induction on the structure of the associated program $P$---it can either be a variable or a composite program.
\begin{itemize}
    \item When $P$ is a variable $x$, the current state $q$ must be in the form of $q_s^a$ for some grammar symbol $s$ and abstract value $a = \alpha^{\mathcal D_x}(\inputex(x))$ (by the input-consistency of $\rho$), and there must be a transition $x \to q$ in $\Delta$ (by the definition of run).
    Let $\Delta_q$ be the set of transitions ending at $q$ and $\delta^*$ be the transition $x \to q$. By the rule for \textsc{BuildOracle} (\autoref{fig:slicing-oracle-rules}), we have:
    \begin{align*}
\alpha(\inputex) \models \mathcal O(q) \ \ \Longleftarrow \ \ \alpha(\inputex) \models  \bigvee_{\delta \in \Delta_q} \Psi_{\delta} \ \ &\Longleftarrow \ \ \alpha(\inputex) \models \Psi_{\delta^*} \\
&\Longleftarrow \ \ \alpha(\inputex) \models \left(x = \alpha^{\mathcal D_x}(\inputex(x))\right)  \ \ \Longleftarrow \ \ \text{True}
    \end{align*}
    \item When $P$ is a composite program $f(P_1, \mydots, P_n)$, let $\rho_i$ be the sub-run of $\rho$ associated with $P_i$, and let $q_i$ be the end state of $\rho_i$ (i.e., $\rho(P_i)$). Then, there must be a transition $f(q_1, \mydots, q_n) \to q$ in $\Delta$ (by the definition of run) and $\rho_i$ must also be input-consistent with $\inputex$ (by the input-consistency of $\rho$). By the induction hypothesis, we have $\alpha(\inputex) \models \mathcal O(q_i)$ for each $i \in [1, n]$.
    
    Similar to the previous case, let $\Delta_q$ be the set of transitions ending at $q$ and $\delta^*$ be this transition $f(q_1, \mydots, q_n) \to q$, then we have:
    \begin{align*}
\alpha(\inputex) \models \mathcal O(q) \ \ &\Longleftarrow \ \ \alpha(\inputex) \models \Psi_{\delta^*} \ \ \Longleftarrow \ \ \alpha(\inputex) \models \bigwedge_{i \in [1, n]} \mathcal O(q_i) \ \ \Longleftarrow \ \ \text{True}
    \end{align*}
\end{itemize}
Therefore, the induction holds, and the \emph{if} direction is proved.

\smallskip 

Then, for the \textit{only-if} direction, we need to show that an input-consistent run $\rho$ exists if the oracle constraint holds, i.e., $\alpha(\inputex) \models \mathcal O(q)$. 
Because the grammar (which our offline FTA $\fta$ is built for) is finite, to simplify our proof, we assume the grammar is acyclic (i.e., no recursive productions, which can be achieved by unrolling). Then, the resulting $\fta$ obviously is also acyclic. 
Therefore, we can prove this direction by induction on the topological order of the states in $\mathcal A$, from initial states to final states. 

Let $\Delta_q$ be the set of transitions which all end at $q$. By the rule for \textsc{BuildOracle}, $\alpha(\inputex) \models \mathcal O(q)$ implies that there is a transition $\delta^* \in \Delta_q$ such that $\alpha(\inputex) \models \Psi_{\delta^*}$. This transition has two possible forms:
\begin{itemize}
    \item If $\delta^* = x \rightarrow q$, its constraint $\Psi_{\delta^*}$ is $x = a$ for $a$ being the abstract value associated with $q$. Since $\alpha(\inputex) \models \Psi_{\delta^*}$, this abstract value $a$ must be equal to $\alpha^{\mathcal D_x}(\inputex(x))$; hence the mapping $\{x \mapsto q\}$ directly forms an input-consistent run ending at $q$.
    \item If $\delta^* = f(q_1, \mydots, q_n) \to q$, its constraint $\Psi_{\delta^*}$ is $\bigwedge_{i \in [1, n]} \mathcal O(q_i)$. Since $\alpha(\inputex) \models \Psi_{\delta^*}$, we have $\alpha(\inputex) \models \mathcal O(q_i)$ for each $i \in [1, n]$. Then, by the induction hypothesis, there is an input-consistent run $\rho_i$ ending at $q_i$ for each $i \in [1, n]$; hence the union of these runs and $\rho(f(P_1, \mydots, P_n)) \mapsto q$ forms an input-consistent run ending at $q$, where $P_i$ is the program associated with $\rho_i$.
\end{itemize}
Therefore, the induction holds, and the \emph{only-if} direction is proved. 

Combining both directions, the lemma follows.
\end{proof}

\newpara{\autoref{theorem:slice-sound-complete}.} Now, we are ready to prove the target theorem \autoref{theorem:slice-sound-complete}.

\begin{theorem}[\autoref{theorem:slice-sound-complete}] Suppose $\ftaslicee{\ioex_1}$ is the slice returned by  $\slicealg$ at~\alglineref{alg:top-level:slice} of~\algorithmref{alg:top-level}, given a DSL with abstract semantics and an example $e_1$. 
Then $\lang(\ftaslicee{\ioex_1})$ is exactly the set of DSL programs that satisfy $\ioex_1$ under abstract semantics. 
That is, $\slicealg$ is both sound (i.e., every accepting run of $\ftaslicee{\ioex_1}$ represents a DSL program that satisfies $\ioex_1$ under abstract semantics) and complete (i.e., every DSL program that satisfies $\ioex_1$ under abstract semantics corresponds to an accepting run in $\ftaslicee{\ioex_1}$).
\end{theorem}
\begin{proof}
    We prove the soundness and completeness separately. 
    
    To prove the \textit{soundness}, we first show that \emph{every run $\rho$ in $\mathcal A_{e_1}$ is input consistent with the input $\inputex$ of $e_1$}. This claim can be proved by induction on the structure of the program $P$ associated with $\rho$.
    \begin{itemize}
        \item When $P$ is a variable $x$, the run $\rho$ is formed by a single transition $x \to q_s^a$.
        By the \textsc{Input} rule in our slicing algorithm (\autoref{fig:online-slicing-rules}), $a$ must be equal to $\alpha^{\mathcal D_x}(\inputex(x))$ (otherwise this transition will not be available in $\mathcal A_{e_1}$), such that $\rho$ must be input-consistent with $\inputex$.
        \item When $P$ is a composite program $f(P_1, \mydots, P_n)$, let $\rho_i$ be the sub-run of $\rho$ associated with $P_i$. By the induction hypothesis, each $\rho_i$ is input-consistent with $\inputex$, such that their union $\rho$ must also be input-consistent with $\inputex$.
    \end{itemize}
Therefore, the induction holds, and every run $\rho$ in $\mathcal A_{e_1}$ is input consistent with $\inputex$.
Then, let $\rho$ be any accepting run in $\mathcal A_{e_1}$, $P$ be the associated program, and $q_{s_0}^a$ be the end state $\rho(P)$.
\begin{itemize}
    \item By the previous claim, $\rho$ is input consistent with $\inputex$; hence the associated abstract value $a$ must be equal to $\abstracttransformer{P}(\inputex)$ (by \lemmaref{lemma:IO-consistent}).
    \item Since $q_{s_0}^a$ is a final state in $\mathcal A_{e_1}$, $a$ must matches the output of $e_1$ (by the \textsc{Output} rule in our slicing algorithm).
\end{itemize}
Therefore, program $P$ must satisfy $e_1$ under abstract semantics, which implies the soundness.

\smallskip 
\vspace{10pt}

For the \textit{completeness}, let $P$ be any program that satisfies $e_1$ under abstract semantics. According to \lemmaref{lemma:offline-fta-completeness}, there must be an accepting run $\rho$ of $P$ that is input-consistent with $\inputex$. The remaining task is to show that this run $\rho$ is also available in the slice $\mathcal A_{e_1}$, i.e., the states in $\rho$ and their induced transitions are all inside $\mathcal A_{e_1}$. To achieve this, we introduce the following claim.
\begin{itemize}
    \item \textbf{Claim.} For any sub-program $P'$ of $P$, if $\rho(P')$ is in the slice $\mathcal A_{e_1}$, then the sub-run of $\rho$ associated with $P'$ is also in $\mathcal A_{e_1}$.
\end{itemize}

If this claim holds, then by taking $P'$ as the entire program $P$, we have that the whole run $\rho$ is in $\mathcal A_{e_1}$ when the end state $q_{s}^a = \rho(P)$ is in $\mathcal A_{e_1}$. This condition is always true because this state must satisfy all three prerequisites in the \textsc{Output} rule of our slicing algorithm:
\begin{itemize}
    \item $q_{s}^a \in Q_f$ because $\rho$ is an accepting run in the offline FTA $\mathcal A$.
    \item $\outputex \in \gamma(a)$ is ensured in two steps. First, since $\rho$ is input-consistent with $\inputex$, the abstract value $a$ must be equal to $\abstracttransformer{P}(\inputex)$ by \lemmaref{lemma:IO-consistent}. Second, since $P$ satisfies $e_1$ under abstract semantics, we have $\outputex \in \gamma(\abstracttransformer{P}(\inputex)) = \gamma(a)$.
    \item $\inputconsistent(q_{s}^a, \inputex, \slicingoracle)$ holds by \lemmaref{lemma:oracle-preciseness} because $\rho$ is input-consistent with $\inputex$.
\end{itemize}

At last, we prove the claim by induction on the structure of the sub-program $P'$. 
\begin{itemize}
    \item When $P'$ is a variable $x$, the sub-run of $\rho$ comprises a single state $q_s^a$ and induces a single transition $x \to q_s^a$.
    Since $\rho$ is input-consistent with $\inputex$, we have $a = \alpha^{\mathcal D_x}(\inputex(x))$; hence this transition must be in $\mathcal A_{e_1}$ when $q_s^a$ is so, by the \textsc{Input} rule in our slicing algorithm.
    \item When $P'$ is a composite program $f(P_1', \mydots, P_n')$, let $q$ be the state $\rho(P')$, $\rho_i$ be the sub-run of $\rho$ associated with $P_i'$, and $q_i$ be the end state of $\rho_i$.
    Let us check the prerequisites for applying the \textsc{Transition} rule (\autoref{fig:online-slicing-rules}) to the state $q$.
    \smallskip

    \begin{enumerate} 
        \item $q \in Q'$ is assumed by the hypothesis of the claim.
        \item $f(q_1, \mydots, a_n) \rightarrow q \in \Delta$ is satisfied since $\rho$ is a run on the offline FTA $\mathcal A$.
        \item $\inputconsistent(q_i, \inputex, \slicingoracle)$ holds for each $i \in [1, n]$---this is ensured by \autoref{lemma:oracle-preciseness}, as $\rho_i$ is an input consistent run ended at $q_i$. 
    \end{enumerate}
    \smallskip 

    Therefore, this rule must have been executed, which adds the transition $f(q_1, \mydots, q_n) \to q$ (which is induced by $\rho$) and all sub-states $q_i$. 
    Then, by the induction hypothesis, all sub-runs $\rho_i$ are also in $\mathcal A_{e_1}$. Hence, the whole sub-run of $\rho$ associated with $P'$ must also be in $\mathcal A_{e_1}$.
\end{itemize}
So the induction holds, and the claim is proved. Consequently, the completeness is also proved. 

\vspace{10pt}

Combining both soundness and completeness, our theorem follows. 

\end{proof}

%% file: sections/appendix-transformer.tex
\clearpage

\section{Abstract transformers for $\sqlname$}
\label{sec:appendix:sqltransformer}

This section presents the complete list of abstract transformers in $\sqlname$. We re-use notations from \autoref{sec:sql:sql}. In addition, we abstract each aggregation expression $\aggexpr = \text{agg}(\colidx)$ 
using $\predcolrange^{\aggexpr}(\setofcols^{\agg})$, 
where the superscript $\agg$ records the aggregation type: 
$\sqlminmax$ means the aggregation is $\sqlmin$ or $\sqlmax$, 
$\sqlsumavg$ means the aggregation is $\sqlsum$ or $\sqlavg$, and 
$\sqlcount$ means the aggregation is $\sqlcount$. 
This annotation is necessary because 
the transformers for $\groupby$ distinguish aggregation types 
(e.g., applying $\sqlsumavg$ to a string-typed column yields $\bot$, 
whereas $\sqlminmax$ preserves the original type). 
The aggregation list $\agglist$ is abstracted by 
$\predlistcolinset^{\agglist}([\setofcols^{\agg_1}_1, \mydots, \setofcols^{\agg_n}_n])$, 
analogously to the column list abstraction $\predlistcolinset^{\columnlist}$.

\begin{figure}[H]
\footnotesize 
\centering
\renewcommand{\arraystretch}{1.2}
\setlength{\arraycolsep}{1pt} 
\[
\begin{array}{rlll}

\abstracttransformer{\proj}
\big(
\prednumofrows^{\sqlquery}(  \sigma_1, \sigma_2 ), \ 
\columnlist
\big)
& \assign
& \prednumofrows^{\sqlquery'}( \sigma_1, \sigma_2  )
\\

\abstracttransformer{\proj}
\big( 
\prednumofcols^{\sqlquery}( k ),  \ 
\predlistcolinset^{\columnlist} ( [ \setofcols_1, \mydots, \setofcols_l ] ) 
\big)
& \assign 
& \prednumofcols^{\sqlquery'} ( l ) 
\ \  
\text{if } 
\bigland_{i \in [1,l]}
\big( \max ( \setofcols_i ) \leq k \big)
\\ 

\abstracttransformer{\proj}
\big( 
\prednumofcols^{\sqlquery}( k ),  \ 
\predlistcolinset^{\columnlist} ( [ \setofcols_1, \mydots, \setofcols_l ] ) 
\big)
& \assign 
& \bot
\ \  
\text{if } 
\biglor_{i \in [1,l]}
\big( \max ( \setofcols_i ) > k \big)
\\ 

\abstracttransformer{\proj}
\big( 
\predsubsetofinput^{\sqlquery}( \setofcols, \sqlinputtablevar),  \ 
\predlistcolinset^{\columnlist} ( [ \setofcols_1, \mydots, \setofcols_l ] ) 
\big)
& \assign 
& \predsubsetofinput^{\sqlquery'} (\setofcols', \sqlinputtablevar) 
\ \ 
\text{where } \setofcols' = \{ i \ | \ \setofcols_i \subseteq \setofcols \}
\\ 

\abstracttransformer{\proj}
\big( 
\predcoltype^{\sqlquery}( \setofcols, t),  \ 
\predlistcolinset^{\columnlist} ( [ \setofcols_1, \mydots, \setofcols_l ] ) 
\big)
& \assign 
& \predcoltype^{\sqlquery'} (\setofcols', t) 
\ \ 
\text{where } \setofcols' = \{ i \ | \ \setofcols_i \subseteq \setofcols \}
\\[2pt] 
\end{array}
\]
\caption{Abstract transformers for $\proj$}
\label{fig:appendix:projtransformer}
\end{figure}

\begin{figure}[H]
\footnotesize
\centering
\renewcommand{\arraystretch}{1.2}
\setlength{\arraycolsep}{1pt} 
\[
\begin{array}{rlll}

\abstracttransformer{\filter}
\big( 
\prednumofrows^{\sqlquery}(  \sigma_1, \sigma_2 ),  \ 
\sqlpredicate
\big)
& \assign 
& \prednumofrows^{\sqlquery'}(  0, \sigma_2  ) 
\\

\abstracttransformer{\filter}
\big( 
\prednumofcols^{\sqlquery}( k ),  \ 
\predcolsused^{\sqlpredicate} ( \maxcol_1, \maxcol_2 ) 
\big)
& \assign 
& \prednumofcols^{\sqlquery'} ( k ) 
\ \  
\text{if } \max(\maxcol_1, \maxcol_2) \leq k
\\ 

\abstracttransformer{\filter}
\big( 
\prednumofcols^{\sqlquery}( k ),  \ 
\predcolsused^{\sqlpredicate} ( \maxcol_1, \maxcol_2 ) 
\big)
& \assign 
& \bot
\ \  
\text{if } \max(\maxcol_1, \maxcol_2) > k
\\

\abstracttransformer{\filter}
\big( 
\predsubsetofinput^{\sqlquery}( \setofcols, \sqlinputtablevar),  \ 
\sqlpredicate
\big)
& \assign 
& \predsubsetofinput^{\sqlquery'} (\setofcols, \sqlinputtablevar) 
\\

\abstracttransformer{\filter}
\big( 
\predcoltype^{\sqlquery}( \setofcols, t),  \ 
\sqlpredicate
\big)
& \assign 
& \predcoltype^{\sqlquery'}( \setofcols, t) 
\\[2pt] 
\end{array}
\]
\caption{Abstract transformers for $\filter$}
\label{fig:appendix:filtertransformer}
\end{figure}

\begin{figure}[H]
\footnotesize
\centering
\renewcommand{\arraystretch}{1.2}
\setlength{\arraycolsep}{1pt} 
\[
\begin{array}{rlll}
\abstracttransformer{\groupby}
\big(
\prednumofrows^{\sqlquery}(  \sigma_1, \sigma_2 ), \ 
\predlistcolinset^{\columnlist} ( [] ), \ 
\agglist
\big)
& \assign
& \prednumofrows^{\sqlquery'}(1,1)
\\

\abstracttransformer{\groupby}
\big(
\prednumofrows^{\sqlquery}(  \sigma_1, \sigma_2 ), \ 
\predlistcolinset^{\columnlist} ( [\setofcols_1, \mydots, \setofcols_l] ), \ 
\agglist )
\big)
& \assign
& \prednumofrows^{\sqlquery'}(\sigma_1,\sigma_2) \text{ if } l \geq 1
\\[5pt]

\abstracttransformer{\groupby}
\big( 
\prednumofcols^{\sqlquery}( k ),  \ 
\predlistcolinset^{\columnlist} ( [ \setofcols_1, \mydots, \setofcols_l ] ), \ 
\predlistcolinset^{\agglist} ( [ \setofcols^{\agg_1}_1, \mydots, \setofcols^{\agg_n}_n ] )
\big)
& \assign 
&  \prednumofcols^{\sqlquery'} (l+n) 
\\

\multicolumn{3}{c}{\text{if} \ \big( \bigland_{i \in [1,l]}
\big( \max ( \setofcols_i ) \leq k \big) \big) 
\land
\big(\bigland_{i \in [1,n]}
\big( \max ( \setofcols^{\agg_i}_i ) \leq k \big)\big)  }

\\[5pt]

\abstracttransformer{\groupby}
\big( 
\prednumofcols^{\sqlquery}( k ),  \ 
\predlistcolinset^{\columnlist} ( [ \setofcols_1, \mydots, \setofcols_l ] ), \ 
\predlistcolinset^{\agglist} ( [ \setofcols^{\agg_1}_1, \mydots, \setofcols^{\agg_n}_n ] )
\big)
& \assign 
& \bot
 
\\ 

\multicolumn{3}{c}{\text{if} \ \big( \biglor_{i \in [1,l]}
\big( \max ( \setofcols_i ) > k \big) \big) 
\lor
\big(\biglor_{i \in [1,n]}
\big( \max ( \setofcols^{\agg_i}_i ) > k \big)\big)  }

\\[5pt]

\abstracttransformer{\groupby}
\big( 
\predsubsetofinput^{\sqlquery}( \setofcols, \sqlinputtablevar),  \ 
\predlistcolinset^{\columnlist} ( [ \setofcols_1, \mydots, \setofcols_l ] ), \ 
\predlistcolinset^{\agglist} ( [ \setofcols^{\agg_1}_1, \mydots, \setofcols^{\agg_n}_n ] )
\big)
& \assign 
& \predsubsetofinput^{\sqlquery'} (\setofcols', \sqlinputtablevar) 

\\ 

\multicolumn{3}{c}{\text{where } \setofcols' = \{ i \ | \ \setofcols_i \subseteq \setofcols \} \cup \{ j + l \ | \ \setofcols_j^{\sqlminmax} \subseteq \setofcols \} }

\\[5pt]

\abstracttransformer{\groupby}
\big( 
\predcoltype^{\sqlquery}( \setofcols, \emph{string}),  \ 
\columnlist, \ 
\predlistcolinset^{\agglist} ( [ \setofcols^{\agg_1}_1, \mydots, \setofcols^{\agg_n}_n ] )
\big)
& \assign 
& \bot \ \text{ if } \ \setofcols_{i}^{\sqlsumavg} \cap \setofcols \neq \emptyset

\\[5pt]

\abstracttransformer{\groupby}
\big( 
\predcoltype^{\sqlquery}( \setofcols, t),  \ 
\predlistcolinset^{\columnlist} ( [ \setofcols_1, \mydots, \setofcols_l ] ), \ 
\predlistcolinset^{\agglist} ( [ \setofcols^{\agg_1}_1, \mydots, \setofcols^{\agg_n}_n ] )
\big)
& \assign 
& \predcoltype^{\sqlquery'}(\setofcols',t) \land \predcoltype^{\sqlquery'}(\setofcols'',\emph{numeric}) 
\\

\multicolumn{3}{c}{
\text{where }
\setofcols' = \{i \ | \ \setofcols_i \subseteq \setofcols \} \ \cup \ \{l + j \ | \ \setofcols_j^{\sqlminmax} \subseteq \setofcols\} \text{ and } \setofcols'' =  \ \{l + k \ | \ \setofcols_k^{\sqlsumavg} \subseteq \setofcols \} \ \cup \ \{ l + m \ | \ \agg_m = \sqlcount \}
}

\\[5pt]

\end{array}
\]
\caption{Abstract transformers for $\groupby$}
\label{fig:appendix:groupbytransformer}
\end{figure}

\begin{figure}[H]
\footnotesize
\centering
\renewcommand{\arraystretch}{1.2}
\setlength{\arraycolsep}{1pt} 
\[
\begin{array}{rlll}

\abstracttransformer{\distinct}
\big(
\prednumofrows^{\sqlquery}(  \sigma_1, \sigma_2 ), \ 
\columnlist
\big)
& \assign
& \prednumofrows^{\sqlquery'}( \sigma_1, \sigma_2  )
\\

\abstracttransformer{\distinct}
\big( 
\prednumofcols^{\sqlquery}( k ),  \ 
\predlistcolinset^{\columnlist} ( [ \setofcols_1, \mydots, \setofcols_l ] ) 
\big)
& \assign 
& \prednumofcols^{\sqlquery'} ( l ) 
\ \  
\text{if } 
\bigland_{i \in [1,l]}
\big( \max ( \setofcols_i ) \leq k \big)
\\ 

\abstracttransformer{\distinct}
\big( 
\prednumofcols^{\sqlquery}( k ),  \ 
\predlistcolinset^{\columnlist} ( [ \setofcols_1, \mydots, \setofcols_l ] ) 
\big)
& \assign 
& \bot
\ \  
\text{if } 
\biglor_{i \in [1,l]}
\big( \max ( \setofcols_i ) > k \big)
\\ 

\abstracttransformer{\distinct}
\big( 
\predsubsetofinput^{\sqlquery}( \setofcols, \sqlinputtablevar),  \ 
\predlistcolinset^{\columnlist} ( [ \setofcols_1, \mydots, \setofcols_l ] ) 
\big)
& \assign 
& \predsubsetofinput^{\sqlquery'} (\setofcols', \sqlinputtablevar) 
\ \ 
\text{where } \setofcols' = \{ i \ | \ \setofcols_i \subseteq \setofcols \}
\\ 

\abstracttransformer{\distinct}
\big( 
\predcoltype^{\sqlquery}( \setofcols, t),  \ 
\predlistcolinset^{\columnlist} ( [ \setofcols_1, \mydots, \setofcols_l ] ) 
\big)
& \assign 
& \predcoltype^{\sqlquery'} (\setofcols', t) 
\ \ 
\text{where } \setofcols' = \{ i \ | \ \setofcols_i \subseteq \setofcols \}
\\[2pt] 
\end{array}
\]
\caption{Abstract transformers for $\distinct$}
\label{fig:appendix:distincttransformer}
\end{figure}

\begin{figure}[H]
\footnotesize
\centering
\renewcommand{\arraystretch}{1.2}
\setlength{\arraycolsep}{1pt} 
\[
\begin{array}{rlll}
\abstracttransformer{\innerjoin}
\big( 
\prednumofrows^{\sqlquery_1}( \sigma_{1, 1}, \sigma_{1, 2} ),  \ 
\prednumofrows^{\sqlquery_2}( \sigma_{2, 1}, \sigma_{2, 2} ),  \ 
\sqlpredicate
\big)
& \assign 
& \prednumofrows^{\sqlquery'}( 0, \sigma_{1, 2} \times \sigma_{2, 2} ) 
\\

\abstracttransformer{\innerjoin}
\big( 
\prednumofcols^{\sqlquery_1}( k_1 ),  \ 
\prednumofcols^{\sqlquery_2}( k_2 ),  \ 
\predcolsused^{\sqlpredicate} ( \maxcol_1, \maxcol_2 ) 
\big)
& \assign 
& \prednumofcols^{\sqlquery'} ( k_1 + k_2 ) 
\ \  
\text{if } \maxcol_i \leq k_i
\\ 

\abstracttransformer{\innerjoin}
\big( 
\prednumofcols^{\sqlquery_1}( k_1 ),  \ 
\prednumofcols^{\sqlquery_2}( k_2 ),  \ 
\predcolsused^{\sqlpredicate} ( \maxcol_1, \maxcol_2 ) 
\big)
& \assign 
& \bot
\ \  
\text{if } \big( \maxcol_1 > k_1 \big) \lor \big( \maxcol_2 > k_2 \big)
\\

\abstracttransformer{\innerjoin}
\big( 
\predsubsetofinput^{\sqlquery_1}( \setofcols, \sqlinputtablevar),  \ 
\sqlquery_2, \ 
\sqlpredicate
\big)
& \assign 
& \predsubsetofinput^{\sqlquery'}( \setofcols, \sqlinputtablevar) 
\\

\abstracttransformer{\innerjoin}
\big( 
\prednumofcols^{\sqlquery_1}( k ),  \
\predsubsetofinput^{\sqlquery_2}( \setofcols, \sqlinputtablevar),  \ 
\sqlpredicate
\big)
& \assign 
& \predsubsetofinput^{\sqlquery'}( \setofcols', \sqlinputtablevar) 
\ \ \text{where }
\setofcols' = \{ j+k \ | \  j \in \setofcols \} 
\\

\abstracttransformer{\innerjoin}
\big( 
\predcoltype^{\sqlquery_1}( \setofcols, t),  \ 
\sqlquery_2, \ 
\sqlpredicate
\big)
& \assign 
& \predcoltype^{\sqlquery'}( \setofcols, t) 
\\

\abstracttransformer{\innerjoin}
\big( 
\prednumofcols^{\sqlquery_1}( k ),  \
\predcoltype^{\sqlquery_2}( \setofcols, t),  \ 
\sqlpredicate
\big)
& \assign 
& \predcoltype^{\sqlquery'}( \setofcols', t) 
\ \ \text{where }
\setofcols' = \{ j+k \ | \  j \in \setofcols \} 
\\[2pt]

\end{array}
\]
\caption{Abstract transformers for $\innerjoin$}
\label{fig:appendix:innerjointransformer}
\end{figure}

\begin{figure}[H]
\footnotesize
\centering
\renewcommand{\arraystretch}{1.2}
\setlength{\arraycolsep}{1pt} 
\[
\begin{array}{rlll}
\abstracttransformer{\leftjoin}
\big( 
\prednumofrows^{\sqlquery_1}( \sigma_{1, 1}, \sigma_{1, 2} ),  \ 
\prednumofrows^{\sqlquery_2}( \sigma_{2, 1}, \sigma_{2, 2} ),  \ 
\sqlpredicate
\big)
& \assign 
& \prednumofrows^{\sqlquery'}( \sigma_{1, 1}, \sigma_{1, 2} \times \sigma_{2, 2} ) 
\\

\abstracttransformer{\leftjoin}
\big( 
\prednumofcols^{\sqlquery_1}( k_1 ),  \ 
\prednumofcols^{\sqlquery_2}( k_2 ),  \ 
\predcolsused^{\sqlpredicate} ( \maxcol_1, \maxcol_2 ) 
\big)
& \assign 
& \prednumofcols^{\sqlquery'} ( k_1 + k_2 ) 
\ \  
\text{if } \maxcol_i \leq k_i
\\ 

\abstracttransformer{\leftjoin}
\big( 
\prednumofcols^{\sqlquery_1}( k_1 ),  \ 
\prednumofcols^{\sqlquery_2}( k_2 ),  \ 
\predcolsused^{\sqlpredicate} ( \maxcol_1, \maxcol_2 ) 
\big)
& \assign 
& \bot
\ \  
\text{if } \big( \maxcol_1 > k_1 \big) \lor \big( \maxcol_2 > k_2 \big)
\\

\abstracttransformer{\leftjoin}
\big( 
\predsubsetofinput^{\sqlquery_1}( \setofcols, \sqlinputtablevar),  \ 
\sqlquery_2, \ 
\sqlpredicate
\big)
& \assign 
& \predsubsetofinput^{\sqlquery'}( \setofcols, \sqlinputtablevar) 
\\

\abstracttransformer{\leftjoin}
\big( 
\predcoltype^{\sqlquery_1}( \setofcols, t),  \ 
\sqlquery_2, \ 
\sqlpredicate
\big)
& \assign 
& \predcoltype^{\sqlquery'}( \setofcols, t) 
\\

\abstracttransformer{\leftjoin}
\big( 
\prednumofcols^{\sqlquery_1}( k ),  \
\predcoltype^{\sqlquery_2}( \setofcols, t),  \ 
\sqlpredicate
\big)
& \assign 
& \predcoltype^{\sqlquery'}( \setofcols', t) 
\ \ \text{where }
\setofcols' = \{ j+k \ | \  j \in \setofcols \} 
\\[2pt]

\end{array}
\]
\caption{Abstract transformers for $\leftjoin$}
\label{fig:appendix:leftjointransformer}
\end{figure}

\begin{figure}[H]
\footnotesize
\centering
\renewcommand{\arraystretch}{1.2}
\setlength{\arraycolsep}{1pt} 
\[
\begin{array}{rlll}

\abstracttransformer{ [\predcolrange^{\colidx_1}(\setofcols_1), \mydots, \predcolrange^{\colidx_l}(\setofcols_l)]}
& \assign 
& \predcolrange^{\columnlist}\big([\setofcols_1, \dotsb, \setofcols_l]\big)

\\

\abstracttransformer{ [\predcolrange^{\aggexpr_1}(\setofcols^{\agg_1}_1), \mydots, \predcolrange^{\aggexpr_l}(\setofcols^{\agg_l}_l)]}
& \assign 
& \predcolrange^{\agglist}\big([\setofcols^{\agg_1}_1, \dotsb, \setofcols^{\agg_l}_l]\big)

\end{array}
\]
\caption{Abstract transformers for lists}
\label{fig:appendix:liststransformer}
\end{figure}

\begin{figure}[H]
\footnotesize
\centering
\renewcommand{\arraystretch}{1.2}
\setlength{\arraycolsep}{1pt} 
\[
\begin{array}{rlll}

\abstracttransformer{ \predcolrange^{\colidx_1}(\setofcols_1) \logicop \predcolrange^{\colidx_2}(\setofcols_2) }
& \assign 
& \predcolsused^{\sqlpredicate} \big( \max(\setofcols_1), \max(\setofcols_2) \big)
\\

\abstracttransformer{ \predcolrange^{\colidx_1}(\setofcols_1) \logicop \sqlvalue }
& \assign 
& \predcolsused^{\sqlpredicate} \big( \max(\setofcols_1), 0 \big)
\\ 

\abstracttransformer{ \predcolsused^{\sqlpredicate_1} ( \maxcol_{1,1}, \maxcol_{1, 2} ) \land \predcolsused^{\sqlpredicate_2} ( \maxcol_{2,1}, \maxcol_{2,2} ) }
& \assign 
& \predcolsused^{\sqlpredicate} \big( \max( \maxcol_{1, 1}, \maxcol_{2, 1} ), \max( \maxcol_{1, 2}, \maxcol_{2, 2} ) \big)
\\
\abstracttransformer{ \predcolsused^{\sqlpredicate_1} ( \maxcol_{1,1}, \maxcol_{1, 2} ) \lor \predcolsused^{\sqlpredicate_2} ( \maxcol_{2,1}, \maxcol_{2,2} ) }
& \assign 
& \predcolsused^{\sqlpredicate} \big( \max( \maxcol_{1, 1}, \maxcol_{2, 1} ), \max( \maxcol_{1, 2}, \maxcol_{2, 2} ) \big)
\\

\end{array}
\]
\caption{Abstract transformers for predicates}
\label{fig:appendix:predtransformer}
\end{figure}

\begin{figure}[H]
\footnotesize
\centering
\renewcommand{\arraystretch}{1.2}
\setlength{\arraycolsep}{1pt} 
\[
\begin{array}{rlll}

\abstracttransformer{\sqlmin(\predcolrange^{\colidx}(J))}
& \assign
& \predcolrange^{\aggexpr}(J^{\sqlminmax})
\\

\abstracttransformer{\sqlmax(\predcolrange^{\colidx}(J))}
& \assign
& \predcolrange^{\aggexpr}(J^{\sqlminmax})
\\

\abstracttransformer{\sqlsum(\predcolrange^{\colidx}(J))}
& \assign
& \predcolrange^{\aggexpr}(J^{\sqlsumavg})
\\

\abstracttransformer{\sqlavg(\predcolrange^{\colidx}(J))}
& \assign
& \predcolrange^{\aggexpr}(J^{\sqlsumavg})
\\

\abstracttransformer{\sqlcount(\predcolrange^{\colidx}(J))}
& \assign
& \predcolrange^{\aggexpr}(J^{\sqlcount})
\\

\end{array}
\]
\caption{Abstract transformers for aggregations}
\label{fig:appendix:aggtransformer}
\end{figure}

\section{Abstract transformers for $\stringname$}
\label{sec:appendix:stringtransformer}

This section presents a representative list of abstract transformers for $\stringname$.

\input{alg-figures/string-transformers.tex}

\section{Abstract transformers for $\matrixname$}
\label{sec:appendix:matrixtransformer}

This section presents a representative list of abstract transformers for $\matrixname$.

\input{alg-figures/tensor-transformers.tex}

%% file: alg-figures/string-transformers.tex
\begin{figure}[H]
\fontsize{8.5pt}{10pt}\selectfont
\centering
\renewcommand{\arraystretch}{1.2}
\setlength{\arraycolsep}{1pt} 
\[
\begin{array}{rlll}
\abstracttransformer{\text{Concat}}
\big( 
\lenOfString^{S_1}(k_1), \lenOfString^{S_2}(k_2)
\big)
& \assign 
& \lenOfString^{S'}(k_1 + k_2)
\\ 

\abstracttransformer{\text{Concat}}
\big( 
\fromInput^{S_1}(J, x), S_2
\big)
& \assign 
& \fromInput^{S'}(J, x)
\\ 

\abstracttransformer{\text{Concat}}
\big( 
\lenOfString^{S_1}(k), \fromInput^{S_2}(J, x)
\big)
& \assign 
& \fromInput^{S'}(J', x)
\ \ 
\text{where } J' = \{ k + j \ | \ j \in J \}

\\
    
\abstracttransformer{\text{Substr}}
\big( 
\lenOfString^{S}(k), \indexOf^{I_l}(j_l), \indexOf^{I_r}(j_r)
\big)
& \assign 
& \lenOfString^S(j_r - j_l)
\ \ \text{if } j_l \leq j_r \leq k
\\ 
\abstracttransformer{\text{Substr}}
\big( 
\charEqual^{S}(J, x, k), \indexOf^{I_l}(j_l), \indexOf^{I_r}(j_r)
\big)
& \assign 
& \charEqual^{S'}(J', x, k)
\ \ \text{where } J' = \{j - j_l\ |\ j \in J \cap [j_l, j_r)\}

\end{array}
\]
\vspace{-10pt}
\caption{Representative abstract transformers for string transformation.}
\label{figure:string-transformer}
\vspace{-10pt}
\end{figure}

%% file: alg-figures/tensor-transformers.tex
\begin{figure}[H]
\fontsize{8.5pt}{10pt}\selectfont
\centering
\renewcommand{\arraystretch}{1.2}
\setlength{\arraycolsep}{1pt} 
\[
\begin{array}{rlll}
\abstracttransformer{\text{Reshape}}
\big( 
\dimNum^T(k), \vectorOf^V(v)
\big)
& \assign 
& \dimNum^{T'}(\textit{size}(v))
\\ 
\abstracttransformer{\text{Reshape}}
\big( 
\shapeOf^{T}(v_1), \vectorOf^V(v_2)
\big)
& \assign 
& \shapeOf^{T'}(v_2)
\\

\abstracttransformer{\text{Permute}}
\big( 
\shapeOf^T(v), \vectorOf^V(v_p)
\big)
& \assign 
& \shapeOf^{T'}(v') \ \ \text{where } v'[i] = v[v_p[i]]
\\ 
\abstracttransformer{\text{Permute}}
\big( 
\elementAt^T(\vout, x, \vin), \vectorOf^V(v_p)
\big)
& \assign 
& \elementAt^{T'}(\voutprime, x, \vin) \ \ \text{where } \voutprime[i] = \vout[v_p[i]]
\\

\abstracttransformer{\text{FlipUD}}
\big( \shapeOf^T(v) \wedge \elementAt^T(\vout, x, \vin)
\big)
& \assign 
& \elementAt^T(\voutprime, x, \vin)
\\
\multicolumn{3}{r}{\text{where } \voutprime[0] = v[0] - \vout[0] - 1 \text{ and } \voutprime[i] = \vout[i] \text{ for } i > 0}

\end{array}
\]
\vspace{-10pt}
\caption{Representative abstract transformers for matrix manipulation.}
\label{figure:matrix-transformer}
\end{figure}

%% file: main.bib
@software{dong_2026_19078478,
  author       = {Dong, Rui and
                  Wu, Qingyue and
                  Ding, Danny and
                  Guo, Zheng and
                  Ji, Ruyi and
                  Wang, Xinyu},
  title        = {Foresighter (Presynthesis: Towards Scaling Up
                   Program Synthesis with Finer-Grained Abstract
                   Semantics)
                  },
  month        = mar,
  year         = 2026,
  publisher    = {Zenodo},
  doi          = {10.5281/zenodo.19078478},
  url          = {https://doi.org/10.5281/zenodo.19078478},
}

@article{wang2017program,
  title={{Program Synthesis using Abstraction Refinement}},
  author={Wang, Xinyu and Dillig, Isil and Singh, Rishabh},
  journal={Proceedings of the ACM on Programming Languages},
  volume={2},
  number={POPL},
  pages={1--30},
  year={2017},
  publisher={ACM New York, NY, USA},
  doi={https://doi.org/10.1145/3158151},
}

@article{wang2017synthesis,
  title={{Synthesis of data completion scripts using finite tree automata}},
  author={Wang, Xinyu and Dillig, Isil and Singh, Rishabh},
  journal={Proceedings of the ACM on Programming Languages},
  volume={1},
  number={OOPSLA},
  pages={1--26},
  year={2017},
  publisher={ACM New York, NY, USA},
  doi={https://doi.org/10.1145/3133886},
}

@inproceedings{wang2018learning,
  title={Learning Abstractions for Program Synthesis},
  author={Wang, Xinyu and Anderson, Greg and Dillig, Isil and McMillan, Kenneth L},
  booktitle={International Conference on Computer Aided Verification},
  pages={407--426},
  year={2018},
  organization={Springer}
}

@article{wang2018relational,
  title={Relational program synthesis},
  author={Wang, Yuepeng and Wang, Xinyu and Dillig, Isil},
  journal={Proceedings of the ACM on Programming Languages},
  volume={2},
  number={OOPSLA},
  pages={1--27},
  year={2018},
  publisher={ACM New York, NY, USA},
  doi={https://doi.org/10.1145/3276525},
}

@phdthesis{wang2019efficient,
  title={An efficient programming-by-example framework},
  author={Wang, Xinyu},
  year={2019}
}

@book{solar2008program,
  title={Program synthesis by sketching},
  author={Solar-Lezama, Armando},
  year={2008},
  publisher={University of California, Berkeley}
}

@article{yaghmazadeh2018automated,
  title={Automated migration of hierarchical data to relational tables using programming-by-example},
  author={Yaghmazadeh, Navid and Wang, Xinyu and Dillig, Isil},
  journal={VLDB},
  volume={11},
  number={5},
  pages={580--593},
  year={2018},
  publisher={VLDB Endowment}, 
doi={https://doi.org/10.1145/3187009.3177735},
}

@inproceedings{wang2017synthesizing,
  title={Synthesizing highly expressive SQL queries from input-output examples},
  author={Wang, Chenglong and Cheung, Alvin and Bodik, Rastislav},
  booktitle={Proceedings of the 38th ACM SIGPLAN Conference on Programming Language Design and Implementation},
  pages={452--466},
  year={2017},
}

@article{gulwani2011automating,
  title={Automating string processing in spreadsheets using input-output examples},
  author={Gulwani, Sumit},
  journal={ACM Sigplan Notices},
  volume={46},
  number={1},
  pages={317--330},
  year={2011},
  publisher={ACM New York, NY, USA},
  doi={https://doi.org/10.1145/1926385.1926423},
}

@misc{comon2008tree,
  title={Tree automata techniques and applications},
  author={Comon, Hubert and Dauchet, Max and Gilleron, R{\'e}mi and Jacquemard, Florent and Lugiez, Denis and L{\"o}ding, Christof and Tison, Sophie and Tommasi, Marc},
  year={2008}
}

@inproceedings{cousot1977abstract,
  author    = {Cousot, Patrick and Cousot, Radhia},
  title     = {Abstract Interpretation: A Unified Lattice Model for Static Analysis of Programs by Construction or Approximation of Fixpoints},
  booktitle = {Proceedings of the 4th ACM SIGACT-SIGPLAN Symposium on Principles of Programming Languages},
  pages     = {238--252},
  year      = {1977},
  address   = {Los Angeles, California},
  publisher = {ACM Press},
  doi       = {10.1145/512950.512973}
}

@article{joshi2002denali,
  title={{Denali: A goal-directed superoptimizer}},
  author={Joshi, Rajeev and Nelson, Greg and Randall, Keith},
  journal={ACM SIGPLAN Notices},
  volume={37},
  number={5},
  pages={304--314},
  year={2002},
  publisher={ACM New York, NY, USA}
}

@inproceedings{alur2013syntax,
  author={Alur, Rajeev and Bodik, Rastislav and Juniwal, Garvit and Martin, Milo M. K. and Raghothaman, Mukund and Seshia, Sanjit A. and Singh, Rishabh and Solar-Lezama, Armando and Torlak, Emina and Udupa, Abhishek},
  booktitle={2013 Formal Methods in Computer-Aided Design}, 
  title={Syntax-guided synthesis}, 
  year={2013},
  pages={1-8},
  doi={10.1109/FMCAD.2013.6679385}
}

@article{feng2017component,
  title={Component-based synthesis of table consolidation and transformation tasks from examples},
  author={Feng, Yu and Martins, Ruben and Van Geffen, Jacob and Dillig, Isil and Chaudhuri, Swarat},
  journal={ACM SIGPLAN Notices},
  volume={52},
  number={6},
  pages={422--436},
  year={2017},
  publisher={ACM New York, NY, USA}
}

@article{guo2019program,
  title={Program synthesis by type-guided abstraction refinement},
  author={Guo, Zheng and James, Michael and Justo, David and Zhou, Jiaxiao and Wang, Ziteng and Jhala, Ranjit and Polikarpova, Nadia},
  journal={Proceedings of the ACM on Programming Languages},
  volume={4},
  number={POPL},
  pages={1--28},
  year={2019},
  publisher={ACM New York, NY, USA}
}

@article{feser2015synthesizing,
  title={Synthesizing data structure transformations from input-output examples},
  author={Feser, John K and Chaudhuri, Swarat and Dillig, Isil},
  journal={ACM SIGPLAN Notices},
  volume={50},
  number={6},
  pages={229--239},
  year={2015},
  publisher={ACM New York, NY, USA},
  doi={https://doi.org/10.1145/2737924.2737977},
}

@article{polikarpova2016program,
  title={Program synthesis from polymorphic refinement types},
  author={Polikarpova, Nadia and Kuraj, Ivan and Solar-Lezama, Armando},
  journal={ACM SIGPLAN Notices},
  volume={51},
  number={6},
  pages={522--538},
  year={2016},
  publisher={ACM New York, NY, USA}
}

@inproceedings{polozov2015flashmeta,
  title={Flashmeta: A framework for inductive program synthesis},
  author={Polozov, Oleksandr and Gulwani, Sumit},
  booktitle={Proceedings of the 2015 ACM SIGPLAN International Conference on Object-Oriented Programming, Systems, Languages, and Applications},
  pages={107--126},
  year={2015}
}

@article{miltner2022bottom,
  title={Bottom-up synthesis of recursive functional programs using angelic execution},
  author={Miltner, Anders and Nu{\~n}ez, Adrian Trejo and Brendel, Ana and Chaudhuri, Swarat and Dillig, Isil},
  journal={Proceedings of the ACM on Programming Languages},
  volume={6},
  number={POPL},
  pages={1--29},
  year={2022},
  publisher={ACM New York, NY, USA},
  doi={https://doi.org/10.1145/3498682},
}

@article{koppel2022searching,
  title={Searching entangled program spaces},
  author={Koppel, James and Guo, Zheng and De Vries, Edsko and Solar-Lezama, Armando and Polikarpova, Nadia},
  journal={Proceedings of the ACM on Programming Languages},
  volume={6},
  number={ICFP},
  pages={23--51},
  year={2022},
  publisher={ACM New York, NY, USA},
  doi={https://doi.org/10.1145/3547622},
}

@article{lee2021combining,
  title={Combining the top-down propagation and bottom-up enumeration for inductive program synthesis},
  author={Lee, Woosuk},
  journal={Proceedings of the ACM on Programming Languages},
  volume={5},
  number={POPL},
  pages={1--28},
  year={2021},
  publisher={ACM New York, NY, USA},
  doi={https://doi.org/10.1145/3434335},
}

@article{brancas2022cubes,
  author = {Brancas, Ricardo and Terra-Neves, Miguel and Ventura, Miguel and Manquinho, Vasco and Martins, Ruben},
title = {CUBES: A Parallel Synthesizer for SQL Using Examples},
year = {2025},
publisher = {Association for Computing Machinery},
address = {New York, NY, USA},
issn = {0934-5043},
url = {https://doi.org/10.1145/3768578},
doi = {10.1145/3768578},
journal = {Form. Asp. Comput.},
month = sep,
}

@article{orvalho2020squares,
  title={SQUARES: A SQL synthesizer using query reverse engineering},
  author={Orvalho, Pedro and Terra-Neves, Miguel and Ventura, Miguel and Martins, Ruben and Manquinho, Vasco},
  journal={Proceedings of the VLDB Endowment},
  volume={13},
  number={12},
  pages={2853--2856},
  year={2020},
  publisher={VLDB Endowment}
}

@article{patsql,
author = {Takenouchi, Keita and Ishio, Takashi and Okada, Joji and Sakata, Yuji},
title = {PATSQL: efficient synthesis of SQL queries from example tables with quick inference of projected columns},
year = {2021},
issue_date = {July 2021},
publisher = {VLDB Endowment},
volume = {14},
number = {11},
issn = {2150-8097},
url = {https://doi.org/10.14778/3476249.3476253},
doi = {10.14778/3476249.3476253},
journal = {Proc. VLDB Endow.},
month = jul,
pages = {1937–1949},
numpages = {13}
}

@misc{cubes-tool,
    year = {{CUBES tool}},
    howpublished = {\url{https://github.com/OutSystems/CUBES
}},
}

@inproceedings{brancas2024towards,
  title={Towards reliable SQL synthesis: Fuzzing-based evaluation and disambiguation},
  author={Brancas, Ricardo and Terra-Neves, Miguel and Ventura, Miguel and Manquinho, Vasco and Martins, Ruben},
  booktitle={International Conference on Fundamental Approaches to Software Engineering},
  pages={232--254},
  year={2024},
  organization={Springer Nature Switzerland Cham}
}

@inproceedings{tran2009query,
  title={Query by output},
  author={Tran, Quoc Trung and Chan, Chee-Yong and Parthasarathy, Srinivasan},
  booktitle={Proceedings of the 2009 ACM SIGMOD International Conference on Management of data},
  pages={535--548},
  year={2009}
}

@article{gulwani2017program,
  title={Program synthesis},
  author={Gulwani, Sumit and Polozov, Oleksandr and Singh, Rishabh and others},
  journal={Foundations and Trends{\textregistered} in Programming Languages},
  volume={4},
  number={1-2},
  pages={1--119},
  year={2017},
  publisher={Now Publishers, Inc.}
}

@inproceedings{solar2006combinatorial,
  title={Combinatorial sketching for finite programs},
  author={Solar-Lezama, Armando and Tancau, Liviu and Bodik, Rastislav and Seshia, Sanjit and Saraswat, Vijay},
  booktitle={Proceedings of the 12th international conference on Architectural support for programming languages and operating systems},
  pages={404--415},
  year={2006}
}

@inproceedings{srivastava2010program,
  title={From program verification to program synthesis},
  author={Srivastava, Saurabh and Gulwani, Sumit and Foster, Jeffrey S},
  booktitle={Proceedings of the 37th annual ACM SIGPLAN-SIGACT symposium on Principles of programming languages},
  pages={313--326},
  year={2010}
}

@article{guria2023absynthe,
  title={Absynthe: Abstract interpretation-guided synthesis},
  author={Guria, Sankha Narayan and Foster, Jeffrey S and Van Horn, David},
  journal={Proceedings of the ACM on Programming Languages},
  volume={7},
  number={PLDI},
  pages={1584--1607},
  year={2023},
  publisher={ACM New York, NY, USA}
}

@article{10.1145/3591288,
author = {Yoon, Yongho and Lee, Woosuk and Yi, Kwangkeun},
title = {Inductive Program Synthesis via Iterative Forward-Backward Abstract Interpretation},
year = {2023},
issue_date = {June 2023},
publisher = {Association for Computing Machinery},
address = {New York, NY, USA},
volume = {7},
number = {PLDI},
url = {https://doi.org/10.1145/3591288},
doi = {10.1145/3591288},
journal = {Proc. ACM Program. Lang.},
month = jun,
articleno = {174},
numpages = {25}
}

@article{10.1145/3632858,
author = {Mell, Stephen and Zdancewic, Steve and Bastani, Osbert},
title = {Optimal Program Synthesis via Abstract Interpretation},
year = {2024},
issue_date = {January 2024},
publisher = {Association for Computing Machinery},
address = {New York, NY, USA},
volume = {8},
number = {POPL},
url = {https://doi.org/10.1145/3632858},
doi = {10.1145/3632858},
journal = {Proc. ACM Program. Lang.},
month = jan,
articleno = {16},
numpages = {25}
}

@inproceedings{sypet,
  title={Component-based synthesis for complex APIs},
  author={Feng, Yu and Martins, Ruben and Wang, Yuepeng and Dillig, Isil and Reps, Thomas W},
  booktitle={Proceedings of the 44th ACM SIGPLAN Symposium on Principles of Programming Languages},
  pages={599--612},
  year={2017}
}

@article{mandelin2005jungloid,
  title={Jungloid mining: helping to navigate the API jungle},
  author={Mandelin, David and Xu, Lin and Bod{\'\i}k, Rastislav and Kimelman, Doug},
  journal={ACM Sigplan Notices},
  volume={40},
  number={6},
  pages={48--61},
  year={2005},
  publisher={ACM New York, NY, USA}
}

@inproceedings{reinking2015type,
  title={A type-directed approach to program repair},
  author={Reinking, Alex and Piskac, Ruzica},
  booktitle={International Conference on Computer Aided Verification},
  pages={511--517},
  year={2015},
  organization={Springer}
}

@article{li2022fast,
  title={Fast graph simplification for interleaved-dyck reachability},
  author={Li, Yuanbo and Zhang, Qirun and Reps, Thomas},
  journal={ACM Transactions on Programming Languages and Systems (TOPLAS)},
  volume={44},
  number={2},
  pages={1--28},
  year={2022},
  publisher={ACM New York, NY}
}

@article{li2021complexity,
  title={On the complexity of bidirected interleaved Dyck-reachability},
  author={Li, Yuanbo and Zhang, Qirun and Reps, Thomas},
  journal={Proceedings of the ACM on Programming Languages},
  volume={5},
  number={POPL},
  pages={1--28},
  year={2021},
  publisher={ACM New York, NY, USA}
}

@article{li2022efficient,
  title={Efficient algorithms for dynamic bidirected dyck-reachability},
  author={Li, Yuanbo and Satya, Kris and Zhang, Qirun},
  journal={Proceedings of the ACM on Programming Languages},
  volume={6},
  number={POPL},
  pages={1--29},
  year={2022},
  publisher={ACM New York, NY, USA}
}

@misc{sygus-website,
    year = {{Syntax-Guided Synthesis website}},
    howpublished = {\url{https://sygus-org.github.io/}},
}

@article{johnson2024automating,
  title={Automating Pruning in Top-Down Enumeration for Program Synthesis Problems with Monotonic Semantics},
  author={Johnson, Keith JC and Krishnan, Rahul and Reps, Thomas and D’Antoni, Loris},
  journal={Proceedings of the ACM on Programming Languages},
  volume={8},
  number={OOPSLA2},
  pages={935--961},
  year={2024},
  publisher={ACM New York, NY, USA}
}

@article{thorup2004compact,
  title={Compact oracles for reachability and approximate distances in planar digraphs},
  author={Thorup, Mikkel},
  journal={Journal of the ACM (JACM)},
  volume={51},
  number={6},
  pages={993--1024},
  year={2004},
  publisher={ACM New York, NY, USA}
}

@inproceedings{farzan2021counterexample,
  title={Counterexample-guided partial bounding for recursive function synthesis},
  author={Farzan, Azadeh and Nicolet, Victor},
  booktitle={International Conference on Computer Aided Verification},
  pages={832--855},
  year={2021},
  organization={Springer}
}

@inproceedings{gvero2013complete,
  title={Complete completion using types and weights},
  author={Gvero, Tihomir and Kuncak, Viktor and Kuraj, Ivan and Piskac, Ruzica},
  booktitle={Proceedings of the 34th ACM SIGPLAN conference on Programming language design and implementation},
  pages={27--38},
  year={2013}
}

@article{kaplan2020scaling,
  title={Scaling laws for neural language models},
  author={Kaplan, Jared and McCandlish, Sam and Henighan, Tom and Brown, Tom B and Chess, Benjamin and Child, Rewon and Gray, Scott and Radford, Alec and Wu, Jeffrey and Amodei, Dario},
  journal={arXiv preprint arXiv:2001.08361},
  year={2020}
}

@article{li2024efficient,
  title={Efficient bottom-up synthesis for programs with local variables},
  author={Li, Xiang and Zhou, Xiangyu and Dong, Rui and Zhang, Yihong and Wang, Xinyu},
  journal={Proceedings of the ACM on Programming Languages},
  volume={8},
  number={POPL},
  pages={1540--1568},
  year={2024},
  publisher={ACM New York, NY, USA}
}

@article{sakr2012overview,
  title={An overview of graph indexing and querying techniques},
  author={Sakr, Sherif and Al-Naymat, Ghazi},
  journal={Graph Data Management: Techniques and Applications},
  pages={71--88},
  year={2012},
  publisher={IGI Global Scientific Publishing}
}

@inproceedings{jin2008efficiently,
  title={Efficiently answering reachability queries on very large directed graphs},
  author={Jin, Ruoming and Xiang, Yang and Ruan, Ning and Wang, Haixun},
  booktitle={Proceedings of the 2008 ACM SIGMOD international conference on Management of data},
  pages={595--608},
  year={2008}
}

@inproceedings{valstar2017landmark,
  title={Landmark indexing for evaluation of label-constrained reachability queries},
  author={Valstar, Lucien DJ and Fletcher, George HL and Yoshida, Yuichi},
  booktitle={Proceedings of the 2017 ACM International Conference on Management of Data},
  pages={345--358},
  year={2017}
}

@article{peng2020answering,
  title={Answering billion-scale label-constrained reachability queries within microsecond},
  author={Peng, You and Zhang, Ying and Lin, Xuemin and Qin, Lu and Zhang, Wenjie},
  journal={Proceedings of the VLDB Endowment},
  volume={13},
  number={6},
  pages={812--825},
  year={2020},
  publisher={VLDB Endowment}
}

@article{chen2022dlcr,
  title={DLCR: efficient indexing for label-constrained reachability queries on large dynamic graphs},
  author={Chen, Xin and Peng, You and Wang, Sibo and Yu, Jeffrey Xu},
  journal={Proceedings of the VLDB Endowment},
  volume={15},
  number={8},
  pages={1645--1657},
  year={2022},
  publisher={VLDB Endowment}
}

@article{cohen2003reachability,
  title={Reachability and distance queries via 2-hop labels},
  author={Cohen, Edith and Halperin, Eran and Kaplan, Haim and Zwick, Uri},
  journal={SIAM Journal on Computing},
  volume={32},
  number={5},
  pages={1338--1355},
  year={2003},
  publisher={SIAM}
}

@inproceedings{ravi2022cafqa,
  title={CAFQA: A classical simulation bootstrap for variational quantum algorithms},
  author={Ravi, Gokul Subramanian and Gokhale, Pranav and Ding, Yi and Kirby, William and Smith, Kaitlin and Baker, Jonathan M and Love, Peter J and Hoffmann, Henry and Brown, Kenneth R and Chong, Frederic T},
  booktitle={Proceedings of the 28th ACM International Conference on Architectural Support for Programming Languages and Operating Systems, Volume 1},
  pages={15--29},
  year={2022}
}

@misc{rainbow-table,
    year = {{Rainbow table (wiki)}},
    howpublished = {\url{https://en.wikipedia.org/wiki/Rainbow_table}},
}

@misc{window-method,
    year = {{Elliptic curve point multiplication}},
    howpublished = {\url{https://en.wikipedia.org/wiki/Elliptic_curve_point_multiplication}},
}

@article{huggins2023accelerating,
  title={Accelerating Quantum Algorithms with Precomputation},
  author={Huggins, William J and McClean, Jarrod R},
  journal={arXiv preprint arXiv:2305.09638},
  year={2023}
}

@article{thorup2005approximate,
  title={Approximate distance oracles},
  author={Thorup, Mikkel and Zwick, Uri},
  journal={Journal of the ACM (JACM)},
  volume={52},
  number={1},
  pages={1--24},
  year={2005},
  publisher={ACM New York, NY, USA}
}

@inproceedings{patrascu2010distance,
  title={Distance oracles beyond the Thorup-Zwick bound},
  author={Patrascu, Mihai and Roditty, Liam},
  booktitle={2010 IEEE 51st Annual Symposium on Foundations of Computer Science},
  pages={815--823},
  year={2010},
  organization={IEEE}
}
